\documentclass{amsart}
\usepackage{amssymb}
\usepackage{amsmath}
\usepackage{graphicx} % Required for inserting images
\usepackage{color}
\usepackage{xcolor}
\usepackage{appendix}
\usepackage[sep=5pt, offset=1em]{simpler-wick}
\usepackage{verbatim}
\usepackage[breaklinks]{hyperref}
\usepackage{enumerate}
\usepackage{stmaryrd} % for \mapsfrom
\theoremstyle{plain} % 'default
\newtheorem{thm}{Theorem}[section]
\newtheorem{lem}[thm]{Lemma}
\newtheorem{cor}[thm]{Corollary}
\newtheorem{prop}[thm]{Proposition}

\newtheorem{introthm}{Theorem}

\theoremstyle{definition}
\newtheorem{defn}[thm]{Definition}
\newtheorem{expl}[thm]{Example}

\newtheorem*{defn*}{Definition}
\theoremstyle{remark}
\newtheorem{rem}[thm]{Remark}

\newcommand{\g}{\mathfrak{g}}
\newcommand{\mc}{\mathcal}
\newcommand{\mr}{\mathrm}
\newcommand{\bt}{\bullet}
\newcommand{\ra}{\rightarrow}

\newcommand{\wh}{\widehat}
\newcommand{\til}{\widetilde}

\newcommand{\ad}{\mathrm{ad}}
\newcommand{\bl}{\color{blue}}

\newcommand{\F}{\mathcal{F}}
\newcommand{\LL}{\mathcal{L}}
\newcommand{\RR}{\mathbb{R}}
\newcommand{\CC}{\mathbb{C}}
\newcommand{\ZZ}{\mathbb{Z}}
\newcommand{\HD}{\mathrm{Dens}^{\frac12}}
\newcommand{\xra}{\xrightarrow}
\newcommand{\dd}{\partial}
\newcommand{\ii}{\mathrm{int}}
\newcommand{\iii}{\mathrm{i}}
\newcommand{\fun}{\widehat{S}\mathcal{F}^*[[\hbar]]}
\newcommand{\funprime}{\widehat{S}\mathcal{F'}^*[[\hbar]]}
\newcommand{\SDR}[5]{#5\;\;\raisebox{1.3pt}{\rotatebox[origin=c]{90}{$\curvearrowright$}}\;\; #1 \underset{#4}{\stackrel{#3}{\leftrightarrows}} #2}  %arguments: big cx, small cx, i, p, K
\newcommand{\ddd}{\mathrm{d}}
\newcommand{\cl}{\mathrm{cl}}
\newcommand{\formal}{\mathrm{formal}}
\newcommand{\Cr}{\mathrm{cr}}
\newcommand{\ul}{\underline}
\newcommand{\A}{\mathcal{A}}
\newcommand{\B}{\mathcal{B}}
\newcommand{\zm}{\mathrm{zm}}
\newcommand{\N}{\mathcal{N}}
\newcommand{\ol}{\overline}
\newcommand{\HH}{\mathbb{H}}
\newcommand{\xx}{\underline{x}}
\newcommand{\xla}{\xleftarrow}
\newcommand{\dr}{\mathrm{d}}

\title{BV pushforward as a quasi-isomorphism}
\author[A. S. Cattaneo]{Alberto S. Cattaneo}
\address{Institut f\"ur Mathematik, Universit\"at Z\"urich, Winterthurerstrasse 190, CH-8057 Z\"urich,
Switzerland}
\email{cattaneo@math.uzh.ch}

\author[P. Mnev]{Pavel Mnev}
\address{University of Notre Dame, Notre Dame, IN 46556, USA}
\address{Institut f\"ur Mathematik, Universit\"at Z\"urich, Winterthurerstrasse 190, CH-8057 Z\"urich,
Switzerland}
\email{pmnev@nd.edu}

\date{\today}

\begin{document}

\begin{abstract}
    Given a BV theory on a space of fields split into two subspaces (``infrared'' and ``ultraviolet''), one has the BV pushforward map $P_*$, sending observables to observables of the effective theory on the infrared space. This note proves that $P_*$ is a quasi-isomorphism of BV complexes, by realizing it as a part of a strong deformation retraction constructed using the homological perturbation lemma. Two proofs are given: (i) comparing Feynman diagrams for $P_*$ with ``cable diagrams'' arising from homological perturbation theory and (ii) using topological quantum mechanics. This construction gives a formula for the quasi-inverse $i_\ii$ of $P_*$---the map lifting observables of the effective theory to the full theory. The topological quantum mechanics perspective---and its realization as an AKSZ theory---allows one to write $i_\ii$ as a path integral (realizing cable diagrams for $i_\ii$ as Feynman diagrams) and to study its classical limit.
\end{abstract}

\maketitle

\setcounter{tocdepth}{3}
\tableofcontents

\section{Introduction}
Consider a gauge theory in the Batalin--Vilkovisky formalism with the space of fields $\F$ (a graded vector space) equipped with a symplectic structure $\omega$ of degree $-1$ and an action $S$ -- a function satisfying the quantum master equation (QME) ${\frac12 \{S,S\}-\iii\hbar\Delta S=0}$. 

One has the ``BV complex'' of the theory -- the space of functions on $\F$ equipped with differential $Q_\hbar=\{S,-\}-\iii\hbar\Delta$. Its cocycles are the gauge-invariant observables.

Assume that $\F$ is split as a sum of two symplectic subspaces 
\begin{equation}\label{intro F=F'+F''}
\F=\F'\oplus \F''
\end{equation}
-- the ``infrared'' and ``ultraviolet'' subspaces. One has:
\begin{enumerate}[(i)]
    \item The effective action $S'$ on $\F'$ defined by the fiber BV integral
    \begin{equation}\label{intro S'}
        e^{\frac\iii\hbar S'}= \int_{\LL\subset \F''} e^{\frac\iii\hbar S}.
    \end{equation}
    Here $\LL$ is some fixed Lagrangian subspace of $\F''$ and the integral in the r.h.s. is considered as a family over $\F'$. The effective action automatically satisfies the QME on $\F'$.
    \item The BV pushforward map between BV complexes of the full and infrared (effective) theory\footnote{
    Here $\mr{Fun}(\F)$ stands for $\wh{S}\F^*[[\hbar]]$ -- the algebra of polynomials on $\F$ completed to formal power series, with coefficients in formal power series in $\hbar$.
    }
    %\marginpar{comment on what Fun(F) means}
    \begin{equation}
    \begin{array}{cccc}
     P_*\colon &(\mr{Fun}(\F),Q_\hbar) &\ra& (\mr{Fun}(\F'),Q'_\hbar) \\
    & O &\mapsto  & O'=e^{-\frac\iii\hbar S'} \int_\LL e^{\frac\iii\hbar S} O
    \end{array}
    \end{equation}
    It is automatically a chain map and in particular sends gauge-invariant observables of the full theory to those of the effective theory. 
\end{enumerate}

Assume that $S=S_0+S_\mr{int}$ is a perturbation of a free theory $S_0$, associated to a differential $d$ on $\F$ compatible with $\omega$ and with the splitting (\ref{intro F=F'+F''}), so that $d$ is acyclic on $\F''$. Also, assume that $\F$ is finite-dimensional. Under these assumptions we prove the following.
\begin{introthm}\!\!\!\footnote{Corollary \ref{cor: P_* is q-iso} in the main text.}
\label{intro thm A}
    The map $P_*$ is a quasi-isomorphism.
\end{introthm}
The proof is by realizing $P_*$ as the map $p_\ii$ in the strong deformation retraction (SDR) of BV complexes
%\footnote{Theorem \ref{thm: Q' and p_int via BV pushforward} (ii).}
\begin{equation}\label{intro SDR for interacting theory}
    \SDR{(\F,Q_\hbar)}{(\F',Q'_\hbar)}{i_\ii}{p_\ii}{K_\ii}.
\end{equation}
This SDR is constructed by the homological perturbation lemma\footnote{Lemma \ref{lemma: HPL}.} (HPL), starting from the SDR for the free classical BV theory, and then deforming the differential $Q_0=\{S_0,-\}$ of the classical free BV complex to $Q_\hbar=Q_0+\{S_\ii,-\}-\iii\hbar\Delta$ -- the differential of the quantum interacting theory. 
Then, Theorem \ref{intro thm A} is an immediate corollary of (ii) of the following.
\begin{introthm}\!\!\!\footnote{Theorem \ref{thm: Q' and p_int via BV pushforward}.}
\label{intro thm B}
    \begin{enumerate}[(i)]
        \item The induced differential $Q'_\hbar$ in (\ref{intro SDR for interacting theory}), defined by HPL formula, has the form $Q'_\hbar=\{S',-\}-\iii\hbar \Delta'$, i.e., is generated by the effective action $S'$ as defined by the fiber BV integral \eqref{intro S'}.
        \item One has $p_\ii=P_*$.
    \end{enumerate}
\end{introthm}
We give two proofs: %of the fact that $P_*$ coincides with $p_\ii$:
\begin{enumerate}[(1)]
    \item By comparing ``cable diagrams'' of homological perturbation theory on one side with Feynman graphs for the BV integrals defining $P_*$ and $S'$ on the other side, see Section \ref{sec: SDR of interacting theory}.
    \item By realizing the ingredients of the SDR (\ref{intro SDR for interacting theory}) in terms of topological quantum mechanics (TQM) -- an auxiliary 1d quantum field theory $\tau$ constructed out of the BV theory $(\F,\omega,S)$, see Section \ref{ss: TQM proof of thm on Q'_hbar and p_int}.\footnote{In fact, if $(\F,\omega,S)$ is itself an $n$-dimensional theory, then $\tau$ is an $(n+1)$-dimensional theory on a cylinder, cf. Section \ref{ss: tau for F an AKSZ theory}.
    One can think of $\tau$ as a ``bulk theory'' inducing $(\F,\omega,S)$ on the boundary, cf. Remark \ref{rem: S^tau as bulk theory for S}.}
\end{enumerate}
%Simultaneously,\footnote{Theorem \ref{thm: Q' and p_int via BV pushforward} (i).} we prove that the induced differential $Q'_\hbar$ in (\ref{intro SDR for interacting theory}), defined by HPL formula, has the form $Q'_\hbar=\{S',-\}-\iii\hbar \Delta'$, i.e., is generated by the effective action $S'$ as defined by fiber BV integral \eqref{intro S'}.

\textbf{The observable-lifting map $i_\ii$.}
Having $P_*=p_\ii$ as a part of the SDR package (\ref{intro SDR for interacting theory}) guarantees that it is a quasi-isomorphism and gives $i_\ii$ as its quasi-inverse -- a map lifting gauge-invariant observables of the infrared theory to gauge-invariant observables of the full theory.

Having such an observable-lifting map is of particular interest in certain situations. For instance, in \cite[Section 7.1]{C_surface_obs}, \cite[Section 5.2]{BCZ}, 4d Yang--Mills theory is realized as an effective theory for the topological $BF+B^2$ theory. Then, the computation of a correlation function of observables in Yang--Mills (e.g. Wilson loops) can be done by lifting those observables to the topological theory and calculating the correlator there. Cf. Section \ref{ss: example - lifting ab Wilson loop} for the abelian case.

%\marginpar{maybe use ``complicated'' instead of ``hard''?}
\textbf{Two ``easy'' and two ``hard'' maps.}
In the SDR (\ref{intro SDR for interacting theory}), the deformed differential $Q_\hbar$ is given as input and four other objects $Q'_\hbar,i_\ii,p_\ii,K_\ii$ are computed by HPL. By complexity of the result, they split into two pairs:
\begin{itemize}
    \item Two ``easy'' objects, $p_\ii$ and $Q'_\hbar$, which can be expressed in terms of BV integrals -- cable diagrams for them simplify to Feynman graphs computing perturbative BV integrals.
    \item Two ``hard'' objects, $i_\ii$ and $K_\ii$. The combinatorics of cable diagrams for them is more complicated and they do not simplify to Feynman diagrams (for BV integrals associated to theory $(\F,\omega,S)$).
However -- and this is one of the results of the paper -- $i_\ii,K_\ii$ can be recognized as BV integrals (and cable  diagrams can be recognized as Feynman graphs) for an auxiliary AKSZ theory -- a Lagrangian description of the TQM $\tau$.\footnote{
See Theorem \ref{thm: i_int,p_int,K_int from 1d AKSZ} and Section \ref{sss: specialization to i_int}.
}
\end{itemize}

\subsection{Topological quantum mechanics $\tau$}
Topological quantum mechanics $\tau$ is a 1d functorial QFT assigning to a point the space of states $\mc{H}_\mr{pt}=\mr{Fun}(\F)$ equipped with the differential $Q_\hbar$.\footnote{This is an example of topological quantum mechanics in the sense of A. Losev \cite{Losev}; $\tau$ is the name we use for our special example of topological quantum mechanics.
}
Thus, the space of states of $\tau$ is the BV complex of theory $(\F,\omega,S)$. Additionally, $\mc{H}_\mr{pt}$ is equipped with a second differential of degree $-1$ -- the gauge-fixing operator $\wh\kappa$ -- the extension of a chain contraction $\kappa$ of $\F$ (defining the gauge-fixing Lagrangian $\LL=\mr{im}(\kappa)$ in (\ref{intro S'})) to a derivation of $\mr{Fun}(\F)$. Then, the Hamiltonian of the TQM $\tau$ is $H=[Q_\hbar,\wh\kappa]$ and the partition function for an interval of length $T$ is
\begin{equation}\label{intro Z_T,dT}
    Z_{T,\dr T}= e^{-T H+\dr T\, \wh\kappa} \in \Omega^\bt(\RR_+)\otimes \mr{End}(\mc{H}_\mr{pt})
\end{equation}
-- a nonhomogeneous form on the space $\RR_+$ of lengths of an interval valued in operators on $\mc{H}_\mr{pt}$; we will write $Z_T$ for its 0-form component $e^{-TH}$. By construction this partition function satisfies the ``topologicity'' equation ${(\dr_T+[Q_\hbar,-])Z_{T,\dr T}=0}$, which implies that $\int_0^\infty Z_{T,dT}$ is a chain homotopy between identity and $\lim_{T\ra\infty}Z_{T}=P_{\ker H}$ -- the projection onto $\ker H$ along $\mr{im}(H)$. Moreover, from the theorem below we have
$\ker H \simeq \mr{Fun}(\F')$, where the isomorphism is given by $p_\ii|_{\ker H}$ and its inverse is $i_\ii$).

One can recover the data of the SDR (\ref{intro SDR for interacting theory}) from the TQM $\tau$:
\begin{introthm}\!\!\!\footnote{Theorem \ref{thm: interacting SDR via TQM}.}
\label{intro thm C}
    One has 
    \begin{equation}
    \begin{aligned}
        i_\ii&=&\lim_{T\ra\infty} Z_T\circ i,\\
        p_\ii &=&\lim_{T\ra\infty} p\circ Z_T,\\
        K_\ii&=&\int_0^T Z_{T,\dr T},\\
        i_\ii\circ p_\ii&=&\lim_{T\ra\infty} Z_T.
    \end{aligned}
    \end{equation}
\end{introthm}
Here $i$ is the pullback by the projection $\pi\colon\F\ra \F'$ and $p$ is the pullback by the inclusion $\iota\colon \F'\ra \F$ in the splitting (\ref{intro F=F'+F''}). 
%$Z_T$ stands for the zero-form component of the TQM partition function (\ref{intro Z_T,dT}). 

The proof of Theorem \ref{intro thm C} is based on the perturbation series for the exponential of a perturbed operator, where $H=\LL_E+\cdots$ is seen as a perturbation of the Lie derivative along the Euler vector field $E$ on $\F$ assigning degree 1 to ultraviolet fields and degree 0 to infrared ones.

\textbf{TQM application 1: second proof of Theorem \ref{intro thm B}.}
The TQM perspective and Theorem \ref{intro thm C} have two immediate applications.
First, they lead to the second proof of Theorem \ref{intro thm B}, not relying on combinatorics of cable diagrams and Feynman graphs, see Section \ref{ss: TQM proof of thm on Q'_hbar and p_int}. 
To give the idea of the proof, let us specialize to (ii) of Theorem \ref{intro thm B} in the case $\F'=0$. Fixing an observable $O\in\mr{Fun}(\F)$, we have
\begin{equation}\label{intro p_int(O) = <<O,1>>}
    p_\ii(O)=P_{\ker H}(O)=\ll O,1 \gg,
\end{equation}
with $\ll,\gg$ a pairing on $\mr{Fun}(\F)$ such that $H$ is self-adjoint and $1$ has norm $1$. It turns out that the pairing 
\begin{equation}
\ll O_1,O_2 \gg=c\int_{\LL\subset \F} e^{\frac\iii\hbar S} O_1 O_2,
\end{equation}
with $c$ a normalization constant, does the job. This immediately implies that the r.h.s. of (\ref{intro p_int(O) = <<O,1>>}) is the BV pushforward $P_*(O)$.

% \begin{enumerate}[(a)]
%     \item A second proof of Theorem \ref{intro thm B}, not relying on combinatorics of cable diagrams and Feynman graphs.
%     \item Formulae for classical limits $\hbar\ra 0$ of the maps in (\ref{intro SDR for interacting theory}).
% \end{enumerate}

\textbf{TQM application 2: classical limit of maps $i_\ii,p_\ii,K_\ii$.}
Let  $H^\cl = H \bmod \hbar$ be the classical TQM Hamiltonian viewed as a vector field on $\F$. For $X$ an object valued in formal power series in $\hbar$, we will denote $X^\cl$ its constant term in $\hbar$ (i.e., $X^\cl$ is the ``classical limit'' $\hbar\ra 0$ of $X$).
%We will use the superscript $\cl$ for $\bmod\hbar$ reduction of various objects.

The differential $Q^\cl=\{S^\cl,-\}$ is a cohomological vector field on $\F$ vanishing at zero. It equips $\F[-1]$ with the structure of an $L_\infty$ algebra. Likewise, $Q'^\cl=\{S'^\cl,-\}'$ equips $\F'[-1]$ with the structure of an $L_\infty$ algebra, obtained by homotopy transfer from $(\F[-1],Q^\cl)$.\footnote{This statement is the classical limit of (i) of Theorem \ref{intro thm B}.}

Let $\Phi_T\colon \F\ra \F$ be the flow of the vector field $-H^\cl$ in time $T$. 
\begin{introthm}\!\!\!\footnote{Corollary \ref{cor: class limit of interacting ipK via flow of H^cl} and Remark \ref{rem: L_infty language}. }
\label{intro thm D}
One has the following:
    \begin{itemize}
        \item The classical limit of $i_\ii$ is $i_\ii^\cl=\lim_{T\ra\infty} \Phi_T^*\circ i \colon \mr{Fun}(\F')\ra \mr{Fun}(\F)$. It is the pullback by a nonlinear map $\pi^\cl_\ii=\lim_{T\ra\infty} \pi\circ \Phi_T \colon \F\ra \F'$ which is an $L_\infty$ morphism of $L_\infty$ algebras. 
%        $i_\ii^\cl\colon (\mr{Fun}(\F'),Q'^\cl,\cdot)\ra (\mr{Fun}(\F),Q^\cl,\cdot)$ is a morphism of 
    \item The classical limit of $p_\ii$ is $p_\ii^\cl=\lim_{T\ra\infty} p\circ \Phi_T^* \colon \mr{Fun}(\F)\ra \mr{Fun}(\F') $. It is the pullback by a nonlinear map $\iota_\ii^\cl =\lim_{T\ra\infty} \Phi_T\circ\iota \colon \F'\ra \F$ which is an $L_\infty$ morphism of $L_\infty$ algebras. $\iota_\ii^\cl$ maps  a point $x'\in \F'$ to the critical point of $S$ restricted to $x'+\LL$.
    \item The classical limit of $K_\ii$ is
    \begin{equation}
        K_\ii^\cl(O)(x)=\int_0^\infty O\big((\mr{id}+\dr T\,\kappa)\circ\Phi_T(x)\big),
    \end{equation}
    for $O\in \mr{Fun}(\F)$.
    \end{itemize}
\end{introthm}

Here the mutually quasi-inverse $L_\infty$ morphisms $\iota_\ii^\cl$, $\pi_\ii^\cl$ are, respectively ``hard'' and ``easy'' maps (also, the homotopy transferred $L_\infty$ structure $Q'^\cl$ is ``easy'' and the chain homotopy $K_\ii^\cl$ is ``hard'').\footnote{This difference in complexity is related to the behavior of the vector field $-H^\cl$: it is tangent to submanifolds $x'+\LL$ and its flow converges to the critical point of $S$ on $x'+\LL$.
So, in the formula for $\iota_\ii^\cl$, the flow starts at $x'$ and proceeds on the submanifold  $x'+\LL$, converging to the conditional critical point of $S$. In the formulae for $\pi_\ii^\cl$ or $K_\ii^\cl$, the initial condition is generic and the flow is more complicated.
}

\textbf{AKSZ description of the theory $\tau$.}
One can describe the TQM above by a path integral, starting from its %description 
presentation
in the Lagrangian formalism, as a 1d AKSZ sigma model \cite{AKSZ} on the interval $[0,T]$. Its space of fields is $\F^\tau=\Omega^\bt([0,T], T^*\F)$, where $T^*\F$ is the target (with $\F$ the space of fields of the BV theory we started with), equipped with canonical cotangent symplectic structure and the degree $+1$ ``cohomological function'' 
$\Theta = \omega^{-1}(p,\frac{\partial S}{\partial x})+\frac12 \omega^{-1}(p,p)$, satisfying $\{\Theta,\Theta\}=0$. Here $x,p$ are the base and fiber coordinates on $T^*\F$. The action is:
\begin{equation}
    S^\tau=\int_0^T \langle \til{p},\dr_t \til{x} \rangle+ \Theta(\til{x},\til{p}),
\end{equation}
where $\til{x}\in \Omega^\bt([0,T],\F)$, $\til{p} \in \Omega^\bt([0,T],\F^*)$ are the ``AKSZ superfields'' (nonhomogeneous forms on $[0,T]$).

We find a gauge-fixing Lagrangian $\LL^\tau\subset \F^\tau$ (\ref{L^tau}) for the AKSZ theory $\tau$, such that one has the following path integral formulae for the SDR data (\ref{intro SDR for interacting theory}):
\begin{introthm}\!\!\!\footnote{Theorem \ref{thm: i_int,p_int,K_int from 1d AKSZ}.}
\label{intro thm E}
    One has:
    \begin{align}
        i_\ii(O')(x_\mr{out}) &= \int_{\til{p}|_{t=0}=0,\; \til{x}|_{t=T}=x_\mr{out}} O'(\pi(\til{x}|_{t=0}))\; e^{\frac\iii\hbar S^\tau} \Big|_{\LL^\tau}, \\
        p_\ii(O)(x'_\mr{out}) &= \int_{\til{p}|_{t=0}=0,\; \til{x}|_{t=T}=x'_\mr{out}} O(\til{x}|_{t=0})\; e^{\frac\iii\hbar S^\tau} \Big|_{\LL^\tau}, \\
        K_\ii(O)(x_\mr{out}) &= \int_0^\infty\int_{\til{p}|_{t=0}=0,\; \til{x}|_{t=T}=x_\mr{out}} O(\til{x}|_{t=0,\, \dr t=\dr T})\; e^{\frac\iii\hbar S^\tau} \Big|_{\LL^\tau},
    \end{align}
    for any $O\in \mr{Fun}(\F)$, $O'\in \mr{Fun}(\F')$, $x_\mr{out}\in \F$, $x'_\mr{out}\in \F'$.
\end{introthm}
The Feynman diagram expansions of these path integrals recover the HPL cable diagrams for $i_\ii,p_\ii,K_\ii$. In particular, the ``hard'' objects $i_\ii,K_\ii$ are now interpreted as BV integrals for the theory $\tau$.

If $(\F,\omega,S)$ is itself an AKSZ theory on an $n$-manifold $M$, then $\tau$ is an AKSZ theory on an $(n+1)$-dimensional cylinder $[0,T]\times M$, cf. Section \ref{ss: tau for F an AKSZ theory}.

\subsection{Plan of the paper}
We start with reminders on the BV pushforward (Section \ref{sec: BV pushforward}) and on strong deformation retractions and homological perturbation lemma (Section \ref{sec: HPL}). We pay extra attention to the construction of SDR for algebras of polynomial functions (the symmetrized tensor power construction) and show different presentations of the chain homotopy, Section \ref{ss: SDR on polynomial functions}. 

In Section \ref{sec: SDR data for free BV theory} we construct the SDR for a free BV theory. Starting from a retraction of the complex of fields onto infrared fields, 
$\SDR{(\F,d)}{(\F',d')}{\iota}{\pi}{\kappa}$, we apply the polynomial algebra construction to obtain SDR for the classical BV complex of the free BV theory $\SDR{(\mr{Fun}(\F),Q_0)}{(\mr{Fun}(\F'),Q'_0)}{i}{p}{K}$. Then we deform the differential on the BV complex to the quantum one, $Q_0\ra Q_{0,\hbar}=Q_0-\iii\hbar\Delta$, and construct the associated deformed SDR using homological perturbation lemma: 
$\SDR{(\mr{Fun}(\F),Q_{0,\hbar})}{(\mr{Fun}(\F'),Q'_{0,\hbar})}{i}{p_\hbar}{K_\hbar}$. We prove that $p_\hbar$ is the expectation value map (BV pushforward) of the free theory.

In Section \ref{sec: SDR of interacting theory} we consider an interacting BV theory with action $S=S_0+S_\mr{int}$. Accordingly, we deform the differential in the BV complex to the interacting one $Q_{0,\hbar}\ra Q_\hbar=Q_{0,\hbar}+Q_\ii$ and consider the associated deformed SDR constructed by homological perturbation lemma, $\SDR{(\mr{Fun}(\F),Q_\hbar)}{(\mr{Fun}(\F'),Q'_\hbar)}{i_\ii}{p_\ii}{K_\ii}$. We prove Theorem \ref{intro thm B} %(Theorem \ref{thm: Q' and p_int via BV pushforward}) 
by a diagrammatic argument. As a corollary we have Theorem \ref{intro thm A}.

In Section \ref{ss: TQM perspective} we introduce the topological quantum mechanics $\tau$ and prove Theorem \ref{intro thm C}. Then we use it to provide the second, non-diagrammatic proof of Theorem \ref{intro thm B}.
In Section \ref{sec: classical limit of i_int, p_int, K_int} we use the TQM perspective to study the classical limit of the SDR of interacting BV theory (\ref{intro SDR for interacting theory}) and have Theorem \ref{intro thm D} as an immediate corollary of Theorem \ref{intro thm C}.
In Section \ref{ss: TQM as 1d AKSZ} we introduce the Lagrangian description of the TQM $\tau$ as an AKSZ theory and present the corresponding path integral formulae for the maps of the SDR package of the interacting theory (\ref{intro SDR for interacting theory}) -- Theorem \ref{intro thm E}. %Feynman diagram expansions of these path integrals reproduce the HPL cable diagrams.

As a running toy example throughout the paper (alongside more involved examples) we use the ``toy scalar field'' -- a model with $\F=\RR\oplus \RR[1]$ and action $S_0=\frac{x^2}{2}$ or $S=\frac{x^2}{2}+P(x)$ with $P$ a polynomial, with $\F'=0$. This toy example is revisited several times as the machinery is being developed in the paper.

\subsection{Comparison to literature}
Connections between BV formalism, homological perturbation lemma, $L_\infty$ algebras and homotopy transfer of those were 
%\marginpar{\bl edited May 18}
known and discussed for a while, %discussed since (at least) 2005, 
see e.g. \cite{HS}, \cite{Losev_GAP}, \cite{KL}, \cite{simpBF}, \cite{discrBF}, \cite{Losev}, \cite{DJP}, \cite{MSW}, \cite{JMSW}.

The fact that the expectation value of an observable in a free BV theory is given by homological perturbation lemma (the map $p_\hbar$ in our notations) was pointed out in \cite{Gwilliam}, \cite{GJF}. The observation that the BV pushforward of a class of observables in $BF$ theory can be understood as an $L_\infty$ quasi-isomorphism was made in \cite[Theorem 2]{simpBF}.
BV pushforward as a variant of Hochschild--Kostant--Rosenberg map was discussed in \cite{GLX}, \cite{GL2}, \cite{SiLi}.

The fact that the effective BV action can be described in terms of homological perturbation theory was known (a) classically, via comparison to homotopy transfer of $L_\infty$ algebras -- see references above, and (b) at a quantum level, %(in the example of interacting scalar field), 
see \cite{Jurco}, \cite{SS}, \cite{DJP}, \cite{JMSW}.

In that sense, Theorem \ref{intro thm B} is %in part 
covered by existing literature. However, the approach via topological quantum mechanics is, to our knowledge, new. We should mention that formulae for the classical limits of maps of the SDR (\ref{intro SDR for interacting theory}) arising from the classical limit of topological quantum mechanics appeared recently, independently of this work, in \cite{Getzler_talk}.

\subsection{Acknowledgements} 
%\marginpar{\bl should the acknowledgements be merged where the same sources are mentioned for ASC and PM? (SCGCS and SwissMAP)}
A.S.C. acknowledges partial support of the SNF Grant No. 200021\_227719.
%and of the Simons Collaboration on Global Categorical Symmetries.
P.M. acknowledges partial support 
%of the Simons Collaboration on Global Categorical Symmetries, 
of the Simons Foundation Travel Support Grant and of the FIM at ETH Zurich. 
%P.M.'s research was (partly) supported by the NCCR SwissMAP, funded by the Swiss National Science Foundation.
This research was (partly) supported by the Simons Collaboration on Global Categorical Symmetries and by
the NCCR SwissMAP, funded
by the Swiss National Science Foundation. This article is based upon
work from COST Action 21109 CaLISTA, supported by COST (European Cooperation in Science and Technology) (www.cost.eu), MSCA-2021-SE-01-101086123 CaLIGOLA, and
MSCA-DN CaLiForNIA-101119552.

\section{Reminder: BV pushforward}\label{sec: BV pushforward}
%\marginpar{it could be an appendix instead}
%{\color{gray} cite as reference: \cite[Sections 2.2.1, 2.2.2]{CMRpert}}
In this section we give a brief reminder on the BV integral and BV pushforward, following \cite[Sections 2.2.1, 2.2.2]{CMRpert}.

\subsection{BV integral}\label{ss: BV integral}
Let $\F$ (the ``space of fields'') be a %finite-dimensional 
$\mathbb{Z}$-graded manifold equipped with a degree $-1$ symplectic form $\omega$. 

%\marginpar{\bl Edited May 16. might need more editing.}
Throughout the paper we will assume by default that $\F$ is finite-dimensional.\footnote{ 
Extension to field theory and infinite-dimensional spaces requires extra care (in particular, with regularization/renormalization of the BV Laplacian and with perturbative BV integrals) and generally goes beyond the scope of this paper. 
We will, however, have some field-theoretic examples with infinite-dimensional $\F$, see Sections \ref{ss: example - lifting ab Wilson loop}, \ref{sss: 1d BF on an interval}, \ref{ss: 1d BF on S^1}, \ref{ss: tau for F an AKSZ theory} and Examples \ref{example: Chern-Simons H^cl in Lorenz gauge}, \ref{example: H^cl for ab BF+B^2 to ab YM}.
}
% \footnote{\bl
% Ideas and constructions of this paper can be extended to infinite-dimensional spaces needed for local field theory. However, such extension requires extra care, in particular with regularization/renormalization of the BV Laplacian and with perturbative BV integrals, and goes beyond the scope of this paper.
% We will, however, have some field-theoretic examples with infinite-dimensional $\F$.
% }

On the space of $\CC$-valued half-densities on $\F$, $\mr{Dens}^{\frac12}(\F)$, one has a canonical differential operator $\Delta$ (the BV Laplacian) of degree $1$, satisfying $\Delta^2=0$, see \cite{Khudaverdian}, \cite{Severa}. In local Darboux coordinates $(x^i,\xi_i)$ on $\F$, one has $\Delta=\frac{\partial^2}{\partial x^i\partial \xi_i}$. 
%We consider half-densities to be valued in $\CC$ by default.

Given a Lagrangian submanifold $\LL\subset \F$, one has a map
\begin{equation}\label{BV integral}
    \int_\LL\colon \HD(\F) \ra \mathbb{C},\quad \rho\mapsto \int_\LL \rho|_\LL
\end{equation}
-- the BV integral.
It takes a half-density on $\F$, restricts it to $\LL$, resulting in a 1-density on $\LL$, which is subsequently integrated over $\LL$.

Writing integrals, we assume that the input is such that they converge. In the cases relevant to us, the integrals are perturbed Gaussian integrals, nondegenerate on $\LL$, and computed perturbatively, via Feynman diagrams.

The two key properties of the BV integral are:
\begin{enumerate}[(i)]
    \item The integral of a $\Delta$-exact half-density vanishes: 
    \begin{equation}
    \int_\LL \Delta \sigma =0.
    \end{equation}
    \item Given a smooth family of Lagrangians $\LL_t$ and a half-density satisfying $\Delta \rho=0$, the BV integral of $\rho$ is $t$-independent:
    \begin{equation}
        \frac{\ddd}{\ddd t} \int_{\LL_t} \rho =0.
    \end{equation}
\end{enumerate}

\subsection{BV pushforward}\label{ss: BV pushforward}
Let $(\F,\omega)$ be the direct product of $(-1)$-symplectic graded manifolds $(\F',\omega')$ and $(\F'',\omega'')$:
\begin{equation}\label{F=F'xF'', omega=omega'+omega''}
    \F=\F'\times \F'',\quad \omega=\omega'+\omega''
\end{equation}
-- one thinks of this as splitting fields into ``slow'' and ``fast'' fields, or ``infrared'' and ``ultraviolet,'' or ``residual fields'' and ``fluctuations.'' %the complement.

%We think of $\F$ as being fibered over $\F'$ with fiber $\F''$. 
Fix a Lagrangian submanifold $\LL\subset \F''$. One has a map from half-densities on $\F$ to half-densities on $\F'$ given by the BV integral over $\LL\subset \F''$
\begin{equation}\label{BV pushforward}
    \int_\LL\colon \HD(\F)=\HD(\F')\wh\otimes \HD(\F'')\xra{\mr{id}\otimes \int_\LL} \HD(\F').
\end{equation}
This map is the \emph{fiber BV integral} or \emph{BV pushforward} associated to the $(-1)$-symplectic  fibration $P\colon\F\ra \F'$. We will also denote the map (\ref{BV pushforward}) by $P_*^\LL$.

The key properties of the BV pushforward are:
\begin{enumerate}[(i)]
    \item The BV pushforward is a chain map with respect to BV Laplacians $\Delta$, $\Delta'$ on half-densities on $\F$, $\F'$:
    \begin{equation}\label{BV pushforward chain map property}
        %P_*^\LL (\Delta \rho) = \Delta' P_*^\LL(\rho),\quad \rho\in \HD(\F).
        \int_\LL \Delta \rho= \Delta' \int_\LL \rho,\quad \rho\in \HD(\F).
    \end{equation}
    \item Let $\LL_t$ be a family of Lagrangians in $\F''$ 
    given as an image of $\LL_0$ by the Hamiltonian flow generated by a time-dependent Hamiltonian $\Psi_t\in C^\infty(\F)_{-1}$.\footnote{Any smooth family of Lagrangians in an odd-symplectic manifold can be described in such a way, see \cite{Schwarz}.}
    %such that for $\epsilon\ra 0$, $\LL_{t+\epsilon}$ is the image of $\LL_t$ under the  Hamiltonian flow  on $\F''$ in time $\epsilon$, generated by a function $\Psi_t\in C^\infty(\LL_t)_{-1}$.
    Let $\rho$ be a half-density on $\F$ satisfying $\Delta\rho=0$. Then one has\footnote{
    A more general formula, if we do not assume $\Delta\rho=0$, is:
    $\frac{\ddd}{\ddd t}\int_{\LL_t}\rho = \Delta'\int_{\LL_t} \Psi_t \rho+\int_{\LL_t} \Psi_t \Delta \rho$.
    }
    \begin{equation}\label{BV pushforward L dependence}
        \frac{\ddd}{\ddd t}\int_{\LL_t}\rho = \Delta'\int_{\LL_t} \Psi_t \rho.
    \end{equation}
    I.e., the BV pushforward of a $\Delta$-closed half-density changes with deformation of $\LL$ by a $\Delta'$-exact term.
\end{enumerate}

\subsection{\texorpdfstring{Exponential $\Delta$-cocycles}{Exponential Delta-cocycles}%elements
}
\label{ss: exponential cocycles}
\subsubsection{Quantum master equation} Of particular interest are $\Delta$-closed half-densities of the form\footnote{In this section we are thinking of $\hbar$ as a finite positive number, with an eye toward considering $\hbar\ra 0$ asymptotics.} 
\begin{equation}
\rho=\mu\, e^{\frac{\iii}{\hbar}S}
\end{equation}
with $\mu$ some fixed $\Delta$-closed reference half-density on $\F$ and $S\in C^\infty(\F)$ a \emph{BV action}. 
The condition $\Delta\rho=0$ is equivalent to the \emph{quantum master equation} (QME) for $S$:
%Such $\rho$ is $\Delta$-closed if and only if $S$ satisfies the 
\begin{equation}\label{QME}
    \frac12\{S,S\} - \iii\hbar \Delta_\mu S=0.
\end{equation}
Here $\{-,-\}$ is the degree $1$ Poisson bracket on $C^\infty(\F)$ associated with the symplectic form $\omega$ and $\Delta_\mu\colon f\mapsto \mu^{-1} \Delta(\mu f)$ is the BV Laplacian on functions on $\F$.

\subsubsection{Canonical transformations} A family of $\Delta$-closed exponential half-densities $\rho_t=\mu\, e^{\frac{\iii}{\hbar}S_t}$ satisfies the property that $\frac{\ddd}{\ddd t}\rho_t= \Delta(\rho_t R_t)$ for some $R_t\in C^\infty(\F)_{-1}$ (in particular, the variation of $\rho_t$ is $\Delta$-exact) if and only if
\begin{equation}\label{infinitesimal canonical transformation}
    \frac{\ddd}{\ddd t} S_t=\{S_t,R_t\}-\iii\hbar \Delta_\mu R_t
\end{equation}
-- then one says that $S_{t}$ for different $t$ are related by canonical transformations (or ``quantum canonical BV transformations'').

\subsubsection{Effective action}
Assume the splitting (\ref{F=F'xF'', omega=omega'+omega''}) and let $\mu,\mu',\mu''$ be reference $\Delta$-closed half-densities on $\F,\F',\F''$ satisfying $\mu=\mu'\otimes\mu''$.
Fix a solution $S$ of the QME on $\F$. Then we have a $\Delta$-closed element $\rho=\mu \, e^{\frac{\iii}{\hbar}S}$. Its BV pushforward has the form
\begin{equation}\label{S' definition via BV pushforward}
    \mu' e^{\frac{\iii}{\hbar}S'}=\int_\LL \mu\, e^{\frac{\iii}{\hbar}S}
\end{equation}
with $S'\in C^\infty(\F')$ the \emph{effective action}, induced on infrared fields. 
%\marginpar{can switch to ``residual fields'' terminology}
% $\F'$ (i.e., on ``infrared fields''). 
%seen as ``infrared fields'' (or ``residual fields'').

The effective action satisfies the following properties due to (\ref{BV pushforward chain map property}), (\ref{BV pushforward L dependence}):
\begin{enumerate}[(i)]
    \item The QME for $S$ on $\F$ implies the QME for $S'$ on $\F'$.
    \item A canonical transformation of $S$ induces, via BV pushforward, a canonical transformation of $S'$.
    \item A variation of the Lagrangian $\LL$ induces a canonical transformation of $S'$.
\end{enumerate}

\subsubsection{Observables, BV complex}\label{sss: observables, BV complex}
Given a solution $S$ of QME on $\F$, one can study half-densities of the form
\begin{equation}
    \rho_O=\mu\, e^{\frac{\iii}{\hbar}S} O
\end{equation}
with $O\in C^\infty(\F)$. One has
\begin{equation}
    -\iii\hbar\Delta \rho_O= \rho_{Q_\hbar O},
\end{equation}
where
\begin{equation}
    Q_\hbar=\{S,-\}-\iii\hbar \Delta_\mu
\end{equation}
is the BV differential on $C^\infty(\F)$. We call $(C^\infty(\F),Q_\hbar)$ the \emph{quantum BV complex}. 
%\marginpar{Maybe call it ``quantum BV complex?''}
Its cocycles -- elements $O$ satisfying $Q_\hbar O=0$ -- are the \emph{quantum observables}.\footnote{In the following we will often use the word ``observable'' more broadly -- for any element of the BV complex. In that terminology, a cocycle is a ``gauge-invariant observable.''}
%\marginpar{In a broader sense, sometimes we call any element of $C^\infty(\F)$ an observable}

The full BV integral (\ref{BV integral}) maps $\rho_O$ to the vacuum expectation value (or ``correlator'') $\langle O \rangle$ of $O$ in the theory determined by the action $S$.

The BV pushforward maps $\rho_O$ to a half-density on $\F'$ of the form $\rho_{O'}=\mu'\, e^{\frac{\iii}{\hbar}S'}O'$ with $O'$ the ``effective'' (or ``induced'') observable on the infrared fields. By construction, $O'$ has the same expectation value as $O$. 

The map
\begin{equation}\label{map O to O'}
O\mapsto O'=\mu'^{-1}e^{-\frac{\iii}{\hbar}S'}\int_\LL \mu\,e^{\frac{\iii}{\hbar}S} O
\end{equation} 
determined by the BV pushforward is a chain map of BV complexes
\begin{equation}
    (C^\infty(\F), Q_\hbar) \ra (C^\infty(\F'),Q'_\hbar)
\end{equation}
with $Q'_\hbar=\{S',-\}'-\iii\hbar\Delta'$ the BV differential on functions on infrared fields.

\section{Reminder: homological perturbation lemma}
\label{sec: HPL}
%{\color{gray} def: SDR data. differential perturbation (HPL). constructions of SDR: dual, tensor algebra, symmetric algebra. (Maybe: interpretation via TQM).}

In this section we review the definition of strong deformation retraction between cochain complexes and state the result on the effect of a deformation of the differential -- the homological perturbation lemma. These are well known in the literature, see \cite{GL}, \cite{Crainic} for details. 

By default, we consider cochain complexes of finite-dimensional vector spaces over $\mathbb{R}$.

\begin{defn}
A \emph{strong deformation retraction} (SDR) from a cochain complex $V^\bt$ with differential $d$ onto a cochain complex $V'^{\bt}$ with differential $d'$ is a triple of (i) a chain inclusion $i\colon V'^{\bt}\ra V^\bt$, (ii) a chain projection $p\colon V^\bt\ra V'^{\bt}$ and (iii) a chain homotopy $V^\bt\ra V^{\bt-1}$ satisfying
\begin{equation}\label{dK+Kd=id-ip}
    dK+Kd=\mr{id}-ip,
\end{equation}
\begin{equation}
    pi =\mr{id}
\end{equation}
and ``side conditions''
\begin{equation}\label{side conditions}
    K^2=0,\; Ki=0,\; pK=0.
\end{equation}
\end{defn}
We call the triple $(i,p,K)$ ``SDR data'' and denote it by %write it as a diagram
\begin{equation}\label{ipK squiggly notation}
    (V,d)\stackrel{(i,p,K)}{\rightsquigarrow} (V',d')
\end{equation}
or equivalently
\begin{equation}\label{ipK}
    % K\;\;\raisebox{1.3pt}{\rotatebox[origin=c]{90}{$\curvearrowright$}}\;\; (V,d) \underset{p}{\stackrel{i}{\leftrightarrows}} (V',d').
    \SDR{(V,d)}{(V',d')}{i}{p}{K}.
\end{equation}

Note that, given an SDR, $i$ and $p$ are automatically \emph{quasi-isomorphisms}, i.e., they induce mutually inverse isomorphisms $i_*$, $p_*$ between the cohomology of $(V,d)$ and cohomology of $(V',d')$. We say that $i$, $p$ are mutually quasi-inverse.

\begin{lem}[Homological perturbation lemma]\label{lemma: HPL}
    Fix SDR data (\ref{ipK}) and consider a deformation of the differential on $V$ from $d$ to $\til{d}=d+\delta$ for some $\delta\colon V\ra V$ satisfying $(d+\delta)^2=0$. Then the deformed complex
    $(V,\til{d})$ is quasi-isomorphic to $V'$ equipped with
    deformed differential
    \begin{equation}\label{HPL tilde d'}
        \til{d}'=d'+\sum_{k\geq 0}p(-\delta K)^k\delta i
    \end{equation}
    and one has the deformed SDR data
    \begin{equation}
%        (V,\til{d}) \stackrel{(\til{i},\til{p},\til{K})}{\rightsquigarrow} (V',\til{d}')
  % \til{K}\;\;\rotatebox[origin=c]{90}{$\curvearrowright$}\;\; (V,\til{d}) \underset{\til{p}}{\stackrel{\til{i}}{\leftrightarrows}} (V',\til{d}').
  \SDR{(V,\til{d})}{(V',\til{d}')}{\til{i}}{\til{p}}{\til{K}}.
    \end{equation}
    with
    \begin{equation}\label{HPL tilde ipK}
        \til{i}= \sum_{k\geq 0}(-K\delta)^k i,\;\;
        \til{p}= \sum_{k\geq 0}p(-\delta K)^k,\;\;
        \til{K}= \sum_{k\geq 0} K(-\delta K)^k.
    \end{equation}
    Here we are assuming that $\delta$ is ``sufficiently small,'' so that the sums over $k$ above are convergent.\footnote{Typically, in examples, $\delta$ shifts some auxiliary grading on $V$ (e.g. power in a formal parameter $\hbar$ or polynomial degree if $V$ is the algebra of polynomials) -- and $K$ does not -- which makes the convergence of sums above automatic. Also, it happens in some examples that sums terminate at finite $k$ for a degree reason.}
    %\marginpar{Mar 3: added a footnote}
\end{lem}

%\marginpar{add a bit on composition of SDRs? (is it needed?)}

\subsection{SDR on polynomial functions}\label{ss: SDR on polynomial functions}
Given SDR (\ref{ipK}), one has a dual SDR between the dual complexes,
\begin{equation}\label{SDR dual}
    % K^\vee\;\;\rotatebox[origin=c]{90}{$\curvearrowright$}\;\; (V^*,d^\vee) \underset{i^\vee}{\stackrel{p^\vee}{\leftrightarrows}} (V'^*,d'^\vee).
    \SDR{(V^*,d^\vee)}{(V'^*,d'^\vee)}{i^\vee}{p^\vee}{ K^\vee}.
\end{equation}
Passing to the tensor algebras $TV^*=\bigoplus_{n\geq 0} (V^*)^{\otimes n}$ and $TV'^*$, with differentials given by extension of $d^\vee, d'^\vee$ by Leibniz identity, one has the SDR data
\begin{equation}
    % K^\otimes\;\;\rotatebox[origin=c]{90}{$\curvearrowright$}\;\; (TV^*,d^\vee) \underset{i^*}{\stackrel{p^*}{\leftrightarrows}} (TV'^*,d'^\vee),
    \SDR{(TV^*,d^\vee)}{(TV'^*,d'^\vee)}{p^*}{i^*}{K^\otimes},
\end{equation}
where $p^*$, $i^*$ are pullbacks by $p,i$ (or, equivalently, extensions of $p^\vee,i^\vee$ to algebra maps). The chain homotopy on $(V^*)^{\otimes n}$ is given by
\begin{equation}\label{Kotimes}
    K^\otimes= \sum_{k=1}^{n} %(p^\vee i^\vee)
    P'^{\otimes k-1}\otimes K^\vee \otimes (\mr{id})^{\otimes n-k},
\end{equation}
where we denoted $P'=p^\vee i^\vee$.
This formula corresponds to the composition of SDRs
\begin{equation}
V\otimes V\otimes\cdots \otimes V \rightsquigarrow V'\otimes V\otimes \cdots \otimes V\rightsquigarrow \cdots \rightsquigarrow V'\otimes V'\cdots \otimes V',
\end{equation}
contracting the $V$ factors one-by-one, going from left to right. One can choose a different order for contractions, which would result in permuting the factors in (\ref{Kotimes}).

Passing to the symmetric algebras of $V^*$, $V'^*$ (algebras of polynomials on $V$, $V'$) by quotienting the tensor algebras by the symmetric group, one has the SDR
\begin{equation}
        % K^\mr{sym}\;\;\rotatebox[origin=c]{90}{$\curvearrowright$}\;\; (S^\bt V^*,d^\vee) \underset{i^*}{\stackrel{p^*}{\leftrightarrows}} (S^\bt V'^*,d'^\vee).
        \SDR{(S^\bt V^*,d^\vee)}{(S^\bt V'^*,d'^\vee)}{p^*}{i^*}{K^\mr{sym}}.
\end{equation}
Here $p^*$, $i^*$ the pullback of polynomials by $p,i$. The chain homotopy acting on $S^nV^*$ is the symmetrization of (\ref{Kotimes}) -- averaging over possible orderings of factors:\footnote{
It might not be immediately obvious that the symmetrized chain homotopy satisfies $(K^\mr{sym})^2=0$. However, this follows from the fact that all terms in the sum (\ref{Ksym}) mutually anti-commute.
}
\begin{equation}\label{Ksym}
    K^\mr{sym}=\sum_{k=1}^{n} \sum_{X_i\in \{P',\mr{id}\}}c_{\#\{P'\},\#\{\mr{id}\}} X_1\otimes\cdots\otimes X_{k-1}\otimes K^\vee \otimes X_{k+1}\otimes \cdots \otimes X_n.
\end{equation}
%Here $P''=\mr{id}-P'$.
In the sum each factor (except the one where $K^\vee$ is inserted) can be either $P'$ or $\mr{id}$ and the coefficient is $c_{M,N}=\frac{M!N!}{(M+N+1)!}$ with  $M$, $N$ the numbers of $P'$ and $\mr{id}$ factors.
% \begin{equation}
%     K^\mr{sym}=\sum_{k=1}^{n} \sum_{X_i\in \{P',P''\}}\frac{1}{\#\{P''\}} X_1\otimes\cdots\otimes X_{k-1}\otimes K^\vee \otimes X_{k+1}\otimes \cdots \otimes X_n.
% \end{equation}
% In the sum each factor (except the one where $K^\vee$ is inserted) can be either $P'$ or $P''=\mr{id}-P'$ and the coefficient is $1/N$ with $N$ the number of $P''$ factors.

\subsection{Equivalent formulae for \texorpdfstring{$K^\mr{sym}$}{Ksym}
}
\subsubsection{Topological quantum mechanics formula}
Equivalently to (\ref{Ksym}), 
one can write $K^\mr{sym}$ as an integral\footnote{
This formula is related by a change of variable $t$ to formulae (325,326) in \cite{discrBF} and (19) in \cite{CMRcell}: $t_\mr{loc.\,cit.}=1-e^{-t}$.
}
%in terms of topological quantum mechanics (or heat flow with exact heat operator) on $V$:
\begin{equation}\label{Ksym via TQM}
    K^\mr{sym}(f)(x)= \int_0^\infty f(e^{-(\mr{id}-ip)t+dt\, K}x)
\end{equation}
in the variable $t$. Here $f$ is a polynomial on $V$ and $x\in V$ is a variable. 
\begin{rem}\label{rem: TQM}
Topological quantum mechanics in the sense of A. Losev \cite[Section 2.4]{Losev} is a graded vector space $V$ (the space of states) equipped with two differentials $d$ and $G$ of degree $+1$ and $-1$, respectively, such that the ($d$-exact) Hamiltonian $H=[d,G]$ is non-negative, with kernel $\ker H= V'$.\footnote{This structure is also known as $\mc{N}=2$ supersymmetric quantum mechanics with supercharges $Q_1=d+G$, $Q_2=i(d-G)$.} Additionally, we assume that $G$ vanishes on $\ker H$.
The partition function of the TQM on an interval is defined as
\begin{equation}\label{Z_TQM}
    Z(t,\dr t)= e^{-tH+\dr t\,G}
\end{equation}
-- a nonhomogeneous form on $[0,+\infty)$ (the space of lengths of the interval) valued in operators on $V$.\footnote{Alternatively, one can think of $Z(t,\dr t)$ as a heat flow for a $d$-exact heat operator (the $0$-form component of $Z$), together with its infinitesimal homotopy (the $1$-form component of $Z$). } 
Note that one has $Z|_{t=0}=\mr{id}$, $Z|_{t\ra \infty}$ is the projection $P'$ onto $\ker H=V'$, $\int_0^\infty Z(t,\dr t)=G (H+c P')^{-1}$ is the chain homotopy contracting $V$ onto $V'$ (for any $c>0$).

If we set $G=K$, $Z(t,\dr t)$ becomes the operator $e^{-t[d,K]+\dr t\,K}$ appearing in the integrand in (\ref{Ksym via TQM}).
\end{rem}

\begin{proof}[Proof of (\ref{Ksym via TQM})]
The r.h.s. of (\ref{Ksym via TQM}) can be written as
\begin{equation}\label{Ksym TQM integral 2}
    \int_0^\infty f((e^{-t}\mr{id}+(1-e^{-t})ip+\dr t\, e^{-t}K)x).
\end{equation}
Evaluated on a monomial $f$ on $V$, it returns the r.h.s. of (\ref{Ksym}) (evaluated on $x\odot\cdots \odot x$), with correct coefficients, since 
\begin{multline}
\int_0^\infty \dr t\, e^{-t} (e^{-t})^{N} (1-e^{-t})^M\underset{\tau=1-e^{-t}}{=}\int_0^1 \dr \tau\, \tau^M (1-\tau)^N\\=B(M+1,N+1)=\frac{M!N!}{(M+N+1)!}=c_{M,N}
\end{multline}
is the Euler's beta function integral.
\end{proof}

%\marginpar{added Mar 4}
\begin{rem}
The fact that the r.h.s. of (\ref{Ksym via TQM}) satisfies the chain homotopy property (\ref{dK+Kd=id-ip}) can be shown directly as follows: the TQM partition function (\ref{Z_TQM}) satisfies $(\dr t \frac{\ddd}{\ddd t}-[d,-])Z=0$. Hence, the integrand $Z^*f$ (the pullback by $Z$) in (\ref{Ksym via TQM}) satisfies
\begin{equation}
    (\dr t \frac{\ddd}{\ddd t}-[d^\vee,-])Z^*=0.
\end{equation}
Integrating over $t$ and using Stokes', we obtain
\begin{equation}
    Z^*\big|_{t=\infty}-Z^*\big|_{t=0}+\left[d^\vee,\int_0^\infty Z^*\right]= p^*i^*-\mr{id}+[d^\vee, K^\mr{sym}]=0
\end{equation}
-- the chain homotopy property.
\end{rem}

\subsubsection{The \texorpdfstring{$P'$, $P''$}{P',P''} formula}
Let $P''=\mr{id}-p^\vee i^\vee$ -- the projector onto the second summand in $V=p^\vee(V'^*)\oplus V''^*$, with $P'=p^\vee i^\vee$ the projector onto the first summand.
Sometimes it is convenient to rewrite (\ref{Ksym}) in terms of insertions of $P',P''$ instead of $P',\mr{id}$ as
\begin{equation}\label{Ksym P',P'' formula}
    K^\mr{sym}=\sum_{k=1}^n \sum_{Y_i\in \{P',P''\}} \frac{1}{\#\{P''\}+1} Y_1\otimes\cdots \otimes Y_{k-1}\otimes K^\vee \otimes Y_{k+1}\otimes \cdots \otimes Y_n.
\end{equation}

The fact that the coefficient of the term with $M$ insertions of $P'$ and $N$ of $P''$ is $1/(N+1)$ is due to the fact that (\ref{Ksym TQM integral 2}) can be written as
\begin{equation}
    \int_0^\infty f((ip+e^{-t}(\mr{id}-ip)+\dr t\, e^{-t}K)x)
\end{equation}
and 
\begin{equation}
    \int_0^\infty \dr t\, e^{-t} 1^M (e^{-t})^N=\frac{1}{N+1}.
\end{equation}

\subsubsection{%Yet another 
Euler vector field
formula for \texorpdfstring{$K^\mr{sym}$}{Ksym}}\label{sss: Ksym=nu whK}
Let $\nu\colon S V^*\ra SV^*$ be the linear operator mapping $f$ to $\frac{1}{N} f$ for $f$ homogeneous, of polynomial degree $N\geq 1$ in $V''$, and returning zero on $S %p^\vee
V'^*$. Let $\wh{K}$ be the extension of $K^\vee$ to a derivation of $SV^*$.

One can rewrite (\ref{Ksym P',P'' formula}) as\footnote{
This formula appears in the proof of Proposition 2.5.5 in \cite{Gwilliam}.} 
\begin{equation}\label{Ksym=nu whK}
    K^\mr{sym}=\nu \wh{K}.
\end{equation}

Equivalently, introducing the ``Lie derivative along the Euler vector field on $V''$'' $\LL_E$ -- the extension of $\mr{id}-p^\vee i^\vee$ to a derivation of $SV^*$ -- one can write
%\marginpar{Morally, $E^{-1}$ is something like indefinite integral on a ray in $V''$ in the logarithmic radial coordinate. It reminds of Poincar\'e chain homotopy.. (to which there should also be a relation via odd Fourier transform)}
\begin{equation}\label{Ksym L_E formula}
    K^\mr{sym}=(\LL_E+c\, p^* i^*)^{-1} \wh{K}
\end{equation}
with $c>0$ an arbitrary positive constant.

\section{SDR data for free BV theory}
\label{sec: SDR data for free BV theory}
Consider a linear version of the setting of Section \ref{sec: BV pushforward}.

Let the space of fields $(\F,d)$ be a cochain complex equipped with a degree $-1$ nondegenerate pairing $\omega$ (a constant $(-1)$-symplectic form) compatible with $d$ (i.e., $\omega(x,dy)=\pm \omega(dx,y)$) and a constant ``reference'' half-density $\mu \in (\mr{Det}\, \F^*)^{\otimes \frac12}$.

The algebra of polynomial functions on fields $S \F^*$ has the structure of a BV algebra, with degree $1$ Poisson bracket and the BV Laplacian induced by $\omega$. Using a basis in $\F$, one has
\begin{eqnarray}
\{f,g\}&=&(\omega^{-1})^{ij}f\frac{\overleftarrow\partial}{\partial x^i}\frac{\overrightarrow\partial}{\partial x^j} g, \label{Poisson bracket in coords}\\ 
\Delta f&=&\frac12(\omega^{-1})^{ij}\frac{\partial^2}{\partial x^i \partial x^j} f. \label{Delta in coords}
\end{eqnarray}

Assume that $\F$ is split as a direct sum of complexes
\begin{equation}\label{F=F' oplus F''}
    \F=\F'\oplus \F''
\end{equation}
with $d',d''$ the differentials on $\F',\F''$ and assume that $(\F'',d'')$ is acyclic. Furthermore, assume that $\F',\F''$ are equipped with compatible degree $-1$ pairings $\omega',\omega''$ and constant reference half-densities $\mu',\mu''$ satisfying
$\omega=\omega'+\omega''$ and $\mu=\mu'\otimes \mu''$. 
%\marginpar{notations: $\omega'+\omega''$ or $\omega'\oplus \omega''$?}

Consider the quadratic function on $\F$
\begin{equation}
    S_0(x)=\frac12\omega(x,dx)  
\end{equation}
-- the ``free BV action.'' It satisfies the QME (\ref{QME}) by virtue of $d^2=0$ and $d$ being traceless (by degree reason).

Similarly, we have free BV actions for $\F'$ and $\F''$: 
\begin{equation}
S_0'(x')=\frac12\omega'(x',d'x'),\quad S_0''(x'')=\frac12 \omega''(x'',d'' x'')
\end{equation}
with $x'\in \F'$, $x''\in \F''$.

\subsection{SDR for classical free BV theory}
Let $\iota,\pi$ be the inclusion of the first summand and the projection to the first summand in (\ref{F=F' oplus F''}). Choose a chain homotopy $\kappa$ between $\mr{id}$ and $\iota\pi$ satisfying the side conditions (\ref{side conditions}), so that one has a SDR
\begin{equation}\label{SDR F to F'}
    \SDR{(\F,d)}{(\F',d')}{\iota}{\pi}{\kappa}
\end{equation}
Additionally, we assume compatibility of the SDR data with $\omega,\omega'$: 
\begin{equation}\label{SDR self-duality condition}
    \omega(x,\iota(y'))=\omega'(\pi(x),y')),\quad \omega(x,\kappa(y))=\pm \omega(\kappa(x),y).
\end{equation}
Equivalently, we require that the SDR (\ref{SDR F to F'}) is self-dual with respect to $\omega,\omega'$, cf. (\ref{SDR dual}) for the definition of duality of SDR.

Next, we promote the SDR (\ref{SDR F to F'}) to an SDR for polynomial functions using the construction of Section \ref{ss: SDR on polynomial functions}:
\begin{equation}\label{SDR of classical free BV theory}
    \SDR{(S\F^*,Q_0)}{(S\F'^*,Q'_0)}{i=\pi^*}{p=\iota^*}{K=\kappa^\mr{sym}}.
\end{equation}
Here $Q_0=\{S_0,-\}$, $Q'_0=\{S'_0,-\}'$ is the extension of dual differentials $d^\vee, d'^\vee$ to derivations of $S\F^*,S\F'^*$.

\begin{expl}\!\!\!\footnote{This example corresponds to ``abstract abelian $BF$ theory'' in the terminology of \cite{simpBF}, \cite{discrBF}, or ``cotangent theory'' in the terminology of \cite{CG}.}
    Given an SDR of cochain complexes
    \begin{equation}\label{SDR for W}
        \SDR{(W,d_W)}{(W',d_{W'})}{\iota_W}{\pi_W}{\kappa_W},
    \end{equation}
    set the spaces of fields to be the cotangent complexes
    \begin{equation}
    \F=T^*[-1]W=W\oplus W^*[-1],\quad \F'=T^*[-1]W'=W'\oplus W'^*[-1]
    \end{equation} 
    %These are naturally cochain complexes 
    equipped
    with differentials $d=d_W\oplus d_W^\vee$, $d'=d_{W'}\oplus d^\vee_{W'}$ and with symplectic forms $\omega$, $\omega'$ induced by the canonical pairing $\langle,\rangle$ between $W$ and $W^*$ and between $W'$ and $W'^*$. 

    The BV action on $W$ is 
    \begin{equation}
    S_0(A,B)=\langle B,d_W A \rangle
    \end{equation} 
    with fields $A\in W$, $B\in W^*[-1]$, and $S_0'$ on $W'$ is similar.
    
    One obtains SDR (\ref{SDR F to F'}) as a direct sum of (\ref{SDR for W}) and its dual:
    \begin{equation}\label{cotangent SDR}
        \SDR{(\F,d) %(W\oplus W^*[-1])
        }{(\F',d') %(W'\oplus W'^*[-1])
        }{\iota_W\oplus \pi_W^\vee}{\pi_W\oplus \iota_W^\vee}{\kappa_W\oplus \kappa_W^\vee}.
    \end{equation}
    Note that there are no self-duality conditions (\ref{SDR self-duality condition}) imposed on (\ref{SDR for W}), whereas for (\ref{cotangent SDR}) they are fulfilled automatically.
\end{expl}

\begin{expl}[Toy free scalar theory]\label{example: toy scalar S_0=x^2/2}
    Set $\F=\RR\oplus \RR[-1]$ with even degree $0$ coordinate $x$ and odd degree $-1$ coordinate $\xi$. We take $\omega=\delta x\wedge \delta\xi$ for the symplectic form and set $S_0=\frac12 x^2$ for the action. This action corresponds to the differential $d\colon e\mapsto \epsilon$ on $\F$, with $e,\epsilon$ basis vectors of degree $0,1$ in the two summands in $\F$. We set $\F'=0$ (note that $(\F,d)$ is a contractible complex). The SDR (\ref{SDR F to F'}) is unique
    \begin{equation}
        \SDR{(\RR\oplus \RR[-1],d)}{0}{0}{0}{\kappa}
    \end{equation}
    with $\kappa\colon \epsilon \mapsto e$. The associated SDR (\ref{SDR of classical free BV theory}) on polynomial functions on fields is
    \begin{equation}\label{toy free scalar with S_0=x^2/2: classical SDR for fun(fields)}
        \SDR{(\CC[x,\xi],Q_0)}{\CC}{i}{p}{K}.
    \end{equation}
    Here $Q_0=x\frac{\partial}{\partial \xi}\colon  x^{N}\xi \mapsto x^{N+1}$; $i$ is the inclusion of constants into polynomials of $x,\xi$; $p$ extracts the constant term of a polynomial and
    \begin{equation}
    % \begin{aligned}
    %     K\colon &x^N& \mapsto&\;\; x^{N-1}\xi, \;\;N\geq 1,\\
    %     &1&\mapsto&\;\; 0.
    % \end{aligned}
    K\colon 
    \begin{cases}
        x^N \mapsto x^{N-1}\xi,\;\; N\geq 1,\\
        1\mapsto 0.\\
    \end{cases}
    \end{equation}
\end{expl}

\begin{expl}[Toy free scalar cont'd]\label{example: toy scalar field with matrix M}
    Generalizing the previous example, set $\F=\RR^n\oplus \RR^n[-1]$ with even degree $0$ coordinates $x^1,\ldots,x^n$ and odd degree $-1$ coordinates $\xi_1,\ldots,\xi_n$. Set $\omega=\delta x^i\wedge \delta\xi_i$ and set $S_0 =\frac12 M_{ij} x^i x^j$, where $M$ is a fixed \emph{nondegenerate} symmetric $n\times n$ matrix. This $S_0$ corresponds to the differential $d\colon e_i\mapsto M_{ij}\epsilon^j$ on $\F$, with $\{e_i\},\{\epsilon^i\}$ the standard basis on the first and second summand in $\F$. By nondegeneracy of $M$, $(\F,d)$ is again contractible, so one can set $\F'=0$ and one has SDR
    \begin{equation}
        \SDR{(\RR^n\oplus \RR^n[-1],d)}{0}{0}{0}{\kappa}
    \end{equation}
    with $\kappa\colon \epsilon^i\mapsto (M^{-1})^{ij}e_j$ as the unique possible chain contraction. Passing to polynomial functions of fields, one has
    \begin{equation}
        \SDR{(\CC[x^1,\ldots,x^n,\xi_1,\ldots,\xi_n],Q_0)}{\CC}{i}{p}{K}
    \end{equation}
    with $Q_0=M_{ij}x^i\frac{\partial}{\partial\xi_j}$; $i,p$ are the inclusion of constants and extraction of constant term. The chain homotopy acts on a homogeneous polynomial as
    \begin{equation}\label{K toy free scalar with matrix M}
        K\colon 
        \begin{cases}
        f(x,\xi) \mapsto\frac{1}{\deg_{x,\xi}(f)}(M^{-1})^{ij}\xi_i \frac{\partial}{\partial x^j}f(x,\xi)\;\;\mr{if}\;\; \deg_{x,\xi}(f)\geq 1,\\
        \mr{const} \mapsto 0.
        \end{cases}
    \end{equation}
    with $\deg_{x,\xi}(f)$ the polynomial degree of $f$ (where $x$ and $\xi$ are counted with degree $+1$).

    We remark that formula (\ref{K toy free scalar with matrix M}) is a special case of (\ref{Ksym=nu whK}).
\end{expl}

\begin{expl}[Free massive scalar on a graph]
    Given a graph $\Gamma$ (possibly with lengths assigned to edges) with $V$ the set of vertices, one can consider the toy free massive scalar field theory on $\Gamma$ -- the special case of Example \ref{example: toy scalar field with matrix M} with $\F=T^*[-1]\RR^V$ and $M=\Delta_\Gamma+m^2$ with $\Delta_\Gamma$ the graph Laplacian and $m>0$ a fixed mass parameter.
\end{expl}

%{\color{gray} Maybe add: (\ref{Ksym=nu whK}) in general case, written in derivative notations; %free scalar theory on a graph; abelian $BF$ on an interval in Dupont gauge.}

\subsection{%Quantization 
%SDR data for free quantum BV theory
%Passing to free quantum theory
Free quantum theory
}

\begin{thm}\label{thm SDR for free quantum BV complexes}
    Deformation of the differential $Q_0\ra Q_{0,\hbar}=Q_0-\iii\hbar \Delta$ on $S\F^*$ induces, via the homological perturbation lemma (Lemma \ref{lemma: HPL}), the following SDR of BV complexes (Section \ref{sss: observables, BV complex}):
    \begin{equation}\label{ihbar phbar Khbar}
        \SDR{(S\F^*,Q_{0,\hbar} %Q_0-\iii\hbar \Delta
        )}{(S\F'^*,Q'_{0,\hbar} %Q'_0-\iii\hbar \Delta'
        )}{i_\hbar}{p_\hbar}{K_\hbar},
    \end{equation}
    where
    \begin{enumerate}[(i)]
        \item \label{thm SDR for free quantum BV complexes, (i)} $Q'_{0,\hbar}=Q'_0-\iii\hbar\Delta'$.
        \item \label{thm SDR for free quantum BV complexes, (ii)} $i_\hbar=i=\pi^*$.
        \item \label{thm SDR for free quantum BV complexes, (iii)} $p_\hbar$ is the BV pushforward map (\ref{map O to O'}) defined by free actions $S_0,S'_0$,
        \begin{equation}\label{p_hbar}
            p_\hbar\colon O\mapsto O'=\mu'^{-1}e^{-\frac{\iii}{\hbar}S'_0}\int_\LL \mu\, e^{\frac{\iii}{\hbar}S_0}O,
        \end{equation}
        with the Lagrangian subspace
        \begin{equation}\label{L=im(kappa)}
            \LL=\mr{im}(\kappa)\subset \F''.
        \end{equation}
    \end{enumerate}
\end{thm}
This theorem (in the case $\F'=0$) is essentially contained in \cite{Gwilliam}, \cite{GJF}; the case with $\F'\neq 0$ is contained in \cite[Proof of Theorem 2]{DJP}. %(in the proof of Theorem 2). 
%-- at least partially -- contained in \cite{GJF}.\marginpar{Is that the right thing to say? ``Alluded to in \cite{GJF}?''}

\begin{proof}
    Item (\ref{thm SDR for free quantum BV complexes, (i)}): by homological perturbation lemma (\ref{HPL tilde d'}), the deformed differential on the retract is
    \begin{equation}\label{thm SDR for free quantum BV complexes, eq1}
    Q'_0-\iota^* (\iii\hbar \Delta) \pi^*-\sum_{k\geq 1}\iota^* (\iii\hbar \Delta)(K(\iii\hbar \Delta))^k \pi^*.
    \end{equation}
    Note that $\Delta \pi^*=\pi^* \Delta'$. This implies that (a) the second term in (\ref{thm SDR for free quantum BV complexes, eq1}) is $-\iii\hbar \Delta'$ and (b) $K\Delta\pi^*=0$ and hence the third term in (\ref{thm SDR for free quantum BV complexes, eq1}) vanishes. Thus, the deformed differential is $Q'_{0,\hbar}=Q'_0-\iii\hbar\Delta'$.

    Item (\ref{thm SDR for free quantum BV complexes, (ii)}): by equation (\ref{HPL tilde ipK}) of homological perturbation lemma,
    the deformed inclusion is
    \begin{equation}
        i_\hbar=\pi^*+\sum_{k\geq 1}(K(\iii\hbar \Delta))^k \pi^*.
    \end{equation}
    Using again that $K\Delta\pi^*=0$, the second term above vanishes, yielding $i_\hbar=\pi^*$.

    Item (\ref{thm SDR for free quantum BV complexes, (iii)}): 
    By homological perturbation lemma (\ref{HPL tilde ipK}),
    we have
    \begin{equation}\label{thm SDR for free quantum BV complexes, eq2}
        p_\hbar=\sum_{k\geq 0} \iota^* (\iii\hbar \Delta K)^k.
    \end{equation}
    Let $x^i$ be a basis in $\F^*$ splitting into a basis $x^{i'}$ in $\F'^*$ and $x^{i''}$ in $\F''^*$. 
    % Using (\ref{Ksym=nu whK}) and (\ref{Delta in coords}), we have 
    % \begin{equation}
    % \Delta K=h^{ij} \frac{\dd^2}{\dd x^i \dd x^j}\nu
    % \end{equation}
    % with $h=\omega''^{-1}\kappa^\vee \in S^2 \F''$ and with $\nu$ as in Section \ref{sss: Ksym=nu whK}.
    Consider the $k=1$ term in (\ref{thm SDR for free quantum BV complexes, eq2}) applied to some $O\in S\F^*$:
    \begin{multline}\label{thm SDR for free quantum BV complexes, eq3}
        \iota^* \iii\hbar\Delta K O \underset{(\ref{Ksym=nu whK})}{=}
        \iii\hbar\, \iota^* \Delta \wh\kappa \nu O
        \underset{(\ref{Delta in coords})}{=} \left(\iii\hbar \frac12(\omega^{-1})^{ij}\frac{\dd^2}{\dd x^i\dd x^j} (\kappa^\vee)^{k''}_{\;\;\;l''}\, x^{l''} \frac{\dd}{\dd x^{k''}} \nu O\right) \Big|_{x''=0}\\=
        \left(\iii\hbar (\omega''^{-1}\kappa^\vee)^{i''k''}\frac{\dd^2}{\dd x^{i''}\dd x^{k''}} \nu O\right)\Big|_{x''=0}
        = \left(\frac{\iii\hbar}{2} h^{i''k''}\frac{\dd^2}{\dd x^{i''}\dd x^{k''}} O\right)\Big|_{x''=0}\\=\frac{\iii\hbar}{2} \eta O|_{x''=0}.
    \end{multline}
    Here $\nu$ is as in Section \ref{sss: Ksym=nu whK} and we denoted 
    \begin{equation}\label{h}
    h=\omega''^{-1}\kappa^\vee \in \mr{Hom}(\F''^*,\F'')_\mr{self-dual}\cong S^2 \F''
    \end{equation}
    -- the \emph{propagator}. We also denoted\footnote{
    We remark that in similar notations the derivation $\wh\kappa$ is $\langle x'',\kappa^\vee \frac{\dd}{\dd x''} \rangle$ with $\langle,\rangle$ the canonical pairing between $\F$ and $\F^*$.
    } 
    \begin{equation}\label{H second order derivation}
    \eta= \omega''^{-1}\left(\frac{\dd}{\dd x''},\kappa^\vee \frac{\dd}{\dd x''}\right)
    \end{equation} 
    the extension of $h$ to a second-order derivation of $S\F^*$. Transitioning to the second line, we used that the derivative in $x^i$ or $x^j$ must hit $x^{l''}$ factor, otherwise the term would not survive when restricting to $x''=0$. In the last equality we used that for the restriction to $x''=0$ to be nonzero, $O$ must be of order $2$ in $x''$ and hence $\nu O=\frac12 O$.
    
    The $k=2$ term in (\ref{thm SDR for free quantum BV complexes, eq2}) is computed by replacing $O\ra \iii\hbar \Delta K O$ in (\ref{thm SDR for free quantum BV complexes, eq3}):
    \begin{multline}
        \iota^* (\iii\hbar \Delta K)^2 O=\\=
        (\iii\hbar)^2\frac12 h^{i_1'' k_1''} \frac{\dd^2}{\dd x^{i_1''}x^{k_1''}} \frac12 (\omega^{-1})^{i_2j_2} \frac{\dd^2}{\dd x^{i_2}\dd x^{j_2}} (\kappa^\vee)^{k''_2}_{\;\;\; l''_2}\, x^{l''_2} \frac{\dd}{\dd x^{k''_2}}\nu O \Big|_{x''=0}.
    \end{multline}
    Here we note that $x^{l''_2}$ must be ``eaten'' by the derivative in $x^{i_2}$ or $x^{j_2}$, but not by derivative in $x^{i_1''}$ or $x^{k''_1}$, since in the latter cases the result would vanish due $h\kappa^\vee=\omega''^{-1}(\kappa^\vee)^2=0$, which follows from $\kappa^2=0$. Also, for the expression to be nonzero $O$ must be of order $4$ in $x''$ and so $\nu O=\frac14 O$.
    Thus, one has
    \begin{multline}
     \iota^* (\iii\hbar \Delta K)^2 O=\\
     =\frac{(\iii\hbar)^2}{2\cdot 4}
     h^{i_1'' k_1''} \frac{\dd^2}{\dd x^{i_1''} \dd x^{k_1''}} 
     h^{i_2'' k_2''} \frac{\dd^2}{\dd x^{i_2''} \dd x^{k_2''}} O\Big|_{x''=0}
     =\frac{(\iii\hbar)^2}{2\cdot 4}\eta^2 O\Big|_{x''=0}.
    \end{multline}
    
    Iterating the construction, for the $k$-th term in (\ref{thm SDR for free quantum BV complexes, eq3}) we obtain
    \begin{equation}
        \iota^* (\iii\hbar \Delta K)^k O=
        \frac{(\iii\hbar)^k}{2^k k!}\eta^k O\Big|_{x''=0}.
    \end{equation}
    Thus, for $p_\hbar$ we obtain
    \begin{multline}\label{thm SDR for free quantum BV complexes, eq4}
        p_\hbar O=\sum_{k\geq 0} \frac{(\iii\hbar)^k}{2^k k!}\eta^k O\Big|_{x''=0}
        =e^{\frac{\iii\hbar}{2} \eta} O\Big|_{x''=0}
        \\= \left.\left(c\int_\LL \mu'' e^{\frac{\iii}{\hbar}S_0''(x'')} O(x'+x'')\right)\right|_{x''=0},
    \end{multline}
    with $c=\left(\int_\LL \mu''\, e^{\frac{\iii}{\hbar}S''_0}\right)^{-1}$ a normalization constant.
    Here the last transition uses a standard formula for the perturbative expansion of a Gaussian momentum of $O$ in terms of Wick contractions. The operator $\eta$ applies a single Wick contraction, and $\eta^k$ applies $k$ Wick contractions; at the end the remaining variables are set to be infrared. This finishes the proof of (\ref{p_hbar}).
\end{proof}

\begin{rem}
    We augment Theorem \ref{thm SDR for free quantum BV complexes} with a formula for $K_\hbar$: for $O\in S\F^*$ of degree $n\geq 1$ in $x''$, one has
    \begin{equation}\label{Khbar O}
        K_\hbar O=
        \sum_{j=0}^{\left[\frac{n-1}{2}\right]} \frac{(\iii\hbar)^j }{n(n-2)\cdots (n-2j)}\wh\kappa \eta^j O,
    \end{equation}
    and for $O$ independent of $x''$, $K_\hbar O=0$.
    
    The proof follows the similar logic to that of (\ref{thm SDR for free quantum BV complexes, (iii)}) of Theorem \ref{thm SDR for free quantum BV complexes}. By (\ref{HPL tilde ipK}), 
    \begin{equation}\label{Khbar HPL f-la}
        K_\hbar=K+\iii\hbar K\Delta K + (\iii\hbar)^2 K\Delta K \Delta K+\cdots
    \end{equation}
    Applying to $O$, by (\ref{Ksym=nu whK}), the first term is $\frac{1}{n}\wh\kappa O$, which is the first term in (\ref{Khbar O}). The second term in (\ref{Khbar HPL f-la}) yields
    \begin{equation}\label{Khbar computation eq1}
        \iii\hbar\,\wh\kappa \nu \Delta\wh\kappa \nu O
        = \frac{\iii\hbar}{n(n-2)} (\kappa^\vee)^{a''}_{\;\;\;b''}\, x^{a''}\frac{\dd}{\dd x^{b''}} 
        \frac12(\omega^{-1})^{ij}\frac{\dd^2}{\dd x^i\dd x^j} (\kappa^\vee)^{k''}_{\;\;\;l''}\, x^{l''} \frac{\dd}{\dd x^{k''}} O.
    \end{equation}
    Here the factor $\frac{1}{n(n-2)}$ comes from the two $\nu$ factors. The factor $x^{l''}$ must be ``eaten'' by a derivative in $x^i$ or $x^j$: otherwise it is either (a) eaten by the derivative in $x^{b''}$, in which case the expression vanishes by $\kappa^2=0$, or (b) not eaten at all, in which case one has a cancellation due to signs.\footnote{
    The cancellation is $(\kappa^\vee\otimes \mr{id})\circ (\mr{id}\otimes\kappa^\vee)+(\mr{id}\otimes\kappa^\vee)\circ (\kappa^\vee\otimes \mr{id}) = \kappa^\vee\otimes \kappa^\vee-\kappa^\vee\otimes \kappa^\vee=0$.
    }
    Thus, (\ref{Khbar computation eq1}) %continues as
    evaluates to
    \begin{equation}
        \frac{\iii\hbar}{n(n-2)}\wh\kappa \eta O.
    \end{equation}
    Subsequent terms in (\ref{Khbar HPL f-la}) are computed similarly, yielding the terms of (\ref{Khbar O}). The range for $j$ in (\ref{Khbar O}) is such that %consists of values of $j$ for which 
    $\eta^jO$ is of positive degree in $x''$ and so $\wh\kappa \eta^jO$ can be nonzero.

    An equivalent way to write $K_\hbar$, as follows immediately from (\ref{Khbar O}), is
    \begin{equation}\label{Khbar via H}
        K_\hbar=\sum_{j\geq 0} \wh\kappa \nu(\iii\hbar \eta\nu)^j
        = \wh\kappa \nu(1-\iii\hbar \eta \nu)^{-1}.
    \end{equation}
\end{rem}

%\marginpar{add a remark on SDR for $\Delta$ (conjugate to the one we are looking at), Poincar\'e contraction and OFT..}

\begin{rem}
    Since the differential $Q_{0,\hbar}$ is the conjugation of the BV Laplacian, $Q_{0,\hbar}=e^{-\frac{\iii}{\hbar}S_0}(-\iii\hbar \Delta)e^{\frac{\iii}{\hbar}S_0}$, we have an isomorphism of complexes
    \begin{equation}
        (S\F^*,Q_{0,\hbar}) \xra{\cdot e^{\frac{\iii}{\hbar}S_0}} (\mr{Fun}(\F),-\iii\hbar\Delta),
    \end{equation}
    where $\mr{Fun}(\F)=\{e^{\frac{\iii}{\hbar}S_0}O\;|\; O\;\mr{polynomial\;on\;}\F\}$. We also have a similar isomorphism for functions of infrared fields:
    %Similarly, we have an isomorphism of complexes
    \begin{equation}
        (S\F'^*,Q'_{0,\hbar}) \xra{\cdot e^{\frac{\iii}{\hbar}S'_0}} (\mr{Fun}(\F'),-\iii\hbar\Delta'),
    \end{equation}
    with $\mr{Fun}(\F')=\{e^{\frac{\iii}{\hbar}S'_0}O'\;|\; O'\;\mr{polynomial\;on\;}\F'\}$. One can then transport SDR data (\ref{ihbar phbar Khbar}) along these isomorphisms to an SDR
    \begin{equation}\label{SDR Delta}
        \SDR{(\mr{Fun}(\F),-\iii\hbar\Delta)}{(\mr{Fun}(\F'),-\iii\hbar\Delta')}{i_\Delta}{p_\Delta}{K_\Delta}%{\bar{i}_\hbar}{\bar{p}_\hbar}{\bar{K}_\hbar},
    \end{equation}
    where
    \begin{equation}\label{ipK Delta}
        %\bar{i}_\hbar
        i_\Delta
        \colon f'\mapsto e^{\frac{\iii}{\hbar}S''_0}\pi^*(f'),
        \quad
        %\bar{p}_\hbar
        p_\Delta
        \colon f\mapsto c\int_\LL \mu''\,f,
        \quad
        %\bar{K}_\hbar
        K_\Delta
        \colon f\mapsto e^{\frac{\iii}{\hbar}S_0}K_\hbar(e^{-\frac{\iii}{\hbar}S_0}f).
    \end{equation}

    %\marginpar{unfinished thought..}
    % {\color{gray} As an example, consider the case $\F'=0$ and let $\F=T^*[-1]W$ with canonical symplectic structure of the cotangent bundle, with $W$ evenly graded. 
    % %The BV Laplacian $\Delta$ is the odd Fourier transform of the de Rham operator $d_W$ on forms on $W$.
    % The complex $(C^\infty(\F),\Delta)$ is isomorphic by odd Fourier transform (OFT) to de Rham complex of $W$, $(\Omega^\bt(W),d_W)$. Thus, one should expect $K_\Delta$ to be related by OFT to a contracting homotopy for forms on $W$ as in Poincar\'e lemma. The situation however is complicated by the fact that we are considering functions on $\F$ satisfying a special ansatz $e^{\frac{\iii}{\hbar}S_0}O$ with $O$ a polynomial, thus the relevant Poincar\'e lemma is also for special forms on $W$.
    % }
\end{rem}

%\begin{expl}[Toy free scalar field: quantum version]
\subsubsection{Example: toy free scalar field -- quantum version}
\label{sss: example -- quantum toy free scalar}
    Returning to Example \ref{example: toy scalar S_0=x^2/2} and turning on the deformation of the differential $Q_0=\xi\frac{\dd}{\dd x}$ in (\ref{toy free scalar with S_0=x^2/2: classical SDR for fun(fields)}) by $\Delta=\frac{\dd^2}{\dd x\dd\xi}$, we obtain
    \begin{equation}
        \SDR{(\CC[x,\xi],Q_0-\iii\hbar\Delta)}{\CC}{i_\hbar=i}{p_\hbar}{K_\hbar}.
    \end{equation}
    Here the deformed projection is
    \begin{equation}
        p_\hbar\colon O(x)\mapsto c\int_{-\infty}^\infty \dr x\, e^{\frac{\iii}{\hbar}\frac{x^2}{2}} O(x)
    \end{equation}
    and maps functions linear in $\xi$ to zero; the normalization constant is 
    \begin{equation}\label{c toy scalar}
    c=(2\pi\hbar)^{-\frac12}e^{-\frac{\pi \iii}{4}}.
    \end{equation}
    In particular, $p_\hbar$ acts on monomials by
    \begin{equation}
        p_\hbar(x^n)=
        \left\{\begin{array}{cl}
        (\iii\hbar)^m(2m-1)!!,& n=2m\;\mr{even},\\
        0,& n\; \mr{odd}.
        \end{array}\right.
    \end{equation}
    The deformed chain homotopy can be easily computed from (\ref{Khbar O}), with $\wh\kappa=\xi\frac{\dd}{\dd x}$, $\eta= \frac{\dd^2}{\dd x^2}$. One finds
    \begin{multline}\label{Khbar toy scalar 1}
        K_\hbar\colon x^n \mapsto \xi (x^{n-1}+\iii\hbar (n-1)x^{n-3}+(\iii\hbar)^2 (n-1)(n-3)x^{n-5}+\cdots)
        \\=
        \sum_{j=0}^{\left[ \frac{n-1}{2}\right]}
        (\iii\hbar)^j %(n-1)(n-3)\cdots (n+1-2j)
        \frac{(n-1)!!}{(n-1-2j)!!} \xi x^{n-1-2j}
    \end{multline}
    for $n\geq 1$.
    In the first line, the sum stops when the exponent for $x$ reaches $0$ (for $n$ odd) or $1$ (for $n$ even). One can also write the result as follows:
    \begin{equation}\label{Khbar toy scalar 2}
        K_\hbar\colon O(x) \mapsto \left[\xi \left(1-\frac{\iii\hbar}{x}\frac{\dd}{\dd x}\right)^{-1}\frac{1}{x}O(x)\right]_\mr{reg},
    \end{equation}
    where $[\cdots]_\mr{reg}$ means that we subtract the singular part of the Laurent series at $x=0$.
    Yet another way to write $K_\hbar$ is
    \begin{equation}\label{Khbar toy scalar 3}
        K_\hbar\colon O(x) \mapsto
        \xi \frac{\iii}{\hbar}e^{-\frac{\iii}{\hbar}\frac{x^2}{2}}\int^x %_{\gamma_x} 
        \dr y\; e^{\frac{\iii}{\hbar}\frac{y^2}{2}}(O(y)-p_\hbar(O)).
    \end{equation}
    Here the integral is taken over the contour
    $\gamma_x$ from $e^{\iii\alpha}\cdot(+\infty)$ 
    to $x$ on the complex $y$ plane, with some fixed $\alpha\in (0,\frac{\pi}{2})\cup (\pi,\frac{3\pi}{2})$. It is easy to see that the result does not depend on $\alpha$.\footnote{
    (\ref{Khbar toy scalar 3}) is proven by evaluating on monomials $x^n$ using iterated integration by parts and comparing to (\ref{Khbar toy scalar 1}). Indeed: let $\Phi_n=\frac{\iii}{\hbar} e^{-\frac{\iii}{\hbar}\frac{x^2}{2}} \int^x %_{\gamma_x} 
    \dr y\, e^{\frac{\iii}{\hbar}\frac{y^2}{2}} y^n$. Integrating by parts, we have $\Phi_n=x^{n-1}+\iii\hbar(n-1)\Phi_{n-2}$. Iterating this relation, for $n=2m+1$ odd, we obtain
    $\Phi_n=x^{n-1}+\iii\hbar(n-1) x^{n-3}+(\iii\hbar)^2 (n-1)(n-3) x^{n-5}+\cdots+ (\iii\hbar)^m (n-1)!!$. Thus, taking into account $p_\hbar(x^{2m+1})=0$, (\ref{Khbar toy scalar 3}) agrees with (\ref{Khbar toy scalar 1}) on odd-degree monomials. For $n=2m$ even, we get $\Phi_n=x^{n-1}+\iii\hbar(n-1) x^{n-3}+(\iii\hbar)^2 (n-1)(n-3) x^{n-5}+\cdots+(\iii\hbar)^{m-1}(n-1)!!x+(\iii\hbar)^m (n-1)!! \Phi_0$.
    Thus, (\ref{Khbar toy scalar 1}) gives $K_\hbar(x^n)=\xi (\Phi_n-(\iii\hbar)^m (n-1)!!\Phi_0)$, which agrees with (\ref{Khbar toy scalar 3}).
    }

    The conjugate SDR data (\ref{SDR Delta}) in this example is
    \begin{equation}
        \SDR{(e^{\frac{\iii}{\hbar}\frac{x^2}{2}}\CC[x,\xi],-\iii\hbar\Delta)}{\CC}{  e^{\frac{\iii}{\hbar}\frac{x^2}{2}}\mapsfrom 1\colon i_\Delta}{p_\Delta\colon f\mapsto c\int_{-\infty}^\infty \dr x f(x)}{K_\Delta}
    \end{equation}
    with
    %$K_\Delta$ counterpart (\ref{ipK Delta}) of $K_\hbar$ is
    \begin{equation}\label{K_Delta toy scalar}
        K_\Delta\colon f(x)\mapsto \xi \frac{\iii}{\hbar}\int^x \dr y\left(f(y)-e^{\frac{\iii}{\hbar}\frac{y^2}{2}}c\int_{-\infty}^\infty \dr z f(z)\right)
    \end{equation}
    for $f(x)=e^{\frac{\iii}{\hbar}\frac{x^2}{2}}O(x)$, $O(x)$ a polynomial. Written in this form, it is obvious that $K_\Delta$ inverts $-\iii\hbar\Delta = -\iii\hbar \frac{\dd^2}{\dd x\dd\xi}$.

\subsection{Cable diagrams}
%{\color{gray} Formalism of cable diagrams. Example $S_0=x^2/2$.}
It is convenient to visualize operators on the symmetric algebra $S\F^*$ arising in homological perturbation lemma by ``cable diagrams.'' A cross-section of a cable with $n$ threads represents an element of $S^n\F^*$; individual threads carry elements of $\F^*$. 
An operation can apply linear operators to individual threads, or it can have several inputs and/or outputs in $\F^*$, merging and/or splitting threads. 
A composition of operations is represented as a horizontal concatenation of cables.

We will show elements of $\F''^*$ as solid lines, elements of $\F'^*$ as dashed lines and elements of $\F^*$ that could be either in $\F'^*$ or in $\F''^*$ as thick solid lines.

Terms in formula (\ref{Ksym P',P'' formula})  in the context of (\ref{SDR of classical free BV theory}) apply $\kappa^\vee$ to one thread (denoted by a black dot) and distribute $P'$, $P''$ over the other threads (shown as dashed or solid threads passing through). BV Laplacian $\Delta$ merges a pair of lines, replacing two linear forms $\alpha,\beta\in \F^*$ with the number $\omega^{-1}(\alpha,\beta)$. Moreover $\Delta$ splits into $\Delta'+\Delta''$, with $\Delta'$ merging a pair of dashed lines and $\Delta''$ merging a pair of solid lines.

With these conventions, terms in the homological perturbation formula (\ref{thm SDR for free quantum BV complexes, eq2}) for $p_\hbar$, 
\begin{equation}\label{phbar HPL formula}
p_\hbar=p+\iii\hbar p\Delta K+(\iii\hbar)^2p\Delta K\Delta K+\cdots
\end{equation}
correspond to cable diagrams as in Figure \ref{fig:phbar cable diagram}. Note that %(a) threads in the cable are considered up to reordering, (b)
a black dot on one of the two merging threads can be slid over to the other thread, by self-adjointness of $\kappa$ (\ref{SDR self-duality condition}).

% \begin{figure}[h]
%     \centering
%     $$\vcenter{\hbox{ \includegraphics[width=0.3\linewidth]{phbar3.eps} }}
%     \quad
%     \vcenter{\hbox{  \includegraphics[width=0.2\linewidth]{petal_graph.eps} }}
%     \quad
%     \vcenter{\hbox{ \includegraphics[width=0.3\linewidth]{Khbar2.eps} }} $$
%     \caption{Caption}
%     \label{fig:placeholder}
% \end{figure}

\begin{figure}[h]
    \centering
    \includegraphics[width=0.4\linewidth]{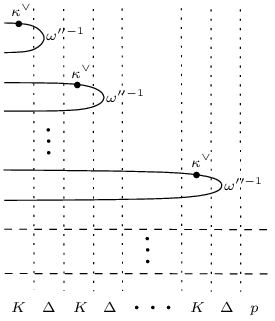}
    \caption{General cable diagram for a term in the homological perturbation series (\ref{phbar HPL formula} for $p_\hbar$.}
    \label{fig:phbar cable diagram}
\end{figure}

Such cable diagrams correspond -- by pinching the left side of the cable, forgetting the horizontal tiered structure, and recalling that $\omega''^{-1}\kappa^\vee=h$ -- to Feynman diagrams computing the BV pushforward of an observable (\ref{p_hbar}) by Wick's contractions, Figure \ref{fig:phbar Feynman graph}.

\begin{figure}[h]
    \centering
    \includegraphics[width=0.2\linewidth]{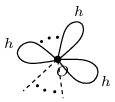}
    \caption{Feynman graph for the BV pushforward of an observable $O$ in a free theory. ($O$ corresponds to the input monomial on the left side of the cable in Figure \ref{fig:phbar cable diagram}.) }
    \label{fig:phbar Feynman graph}
\end{figure}

In Figure \ref{fig:Khbar cable diagrams} we show the general nonvanishing cable diagrams computing the terms in the homological perturbation series (\ref{Khbar HPL f-la}) for $K_\hbar$.
\begin{figure}[h]
    \centering
    \includegraphics[width=0.4\linewidth]{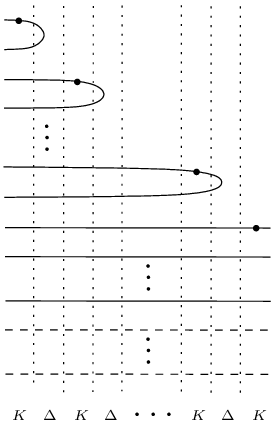}
    \caption{Cable diagrams for $K_\hbar$.}
    \label{fig:Khbar cable diagrams}
\end{figure}

%\section{Interacting theory}
\section{SDR of the interacting theory and BV pushforward. Cable diagrams for $Q'_\hbar$, $p_\ii$ as Feynman diagrams}
%{\color{gray} Relation of $p_{int}$ and induced differential to BV integral. Sketch of proof: relating Feynman graphs to cable diagrams by Morse theory on a Feynman graph.Classical truncation; $L_\infty$-morphism property of $i_{int},p_{int}$; interpretation (conditional critical point map, gauge transf to a residual field).  An alternative ipK triple from conjugation be $e^{\frac{\iii}{\hbar}S}$: alternative observable-lifting map.}
\label{sec: SDR of interacting theory}

In this section and onward, we will consider $\hbar$ to be a formal infinitesimal parameter and will switch the model for functions on $\F$ from polynomials $S\F^*$ to $%\mr{Fun}(\F)=
\wh{S}\F^*[[\hbar]]$, where $\wh{S}\F^*$ is the completion of polynomials to formal power series. This change is needed for the convergence of homological perturbation series we encounter below.\footnote{Another possibility is to keep $\hbar$ finite but introduce a formal infinitesimal coupling constant $g$, require that $S_\ii\in g S\F^*[[g]]$ is of degree $\geq 1$ in $g$, and set the model for function on $\F$ to be $S\F^*[[g]]$.} For functions on $\F'$, we similarly take $\wh{S}\F'^*[[\hbar]]$. Integrals (BV pushforwards) we encounter in this section are understood perturbatively -- as perturbed Gaussian integrals computed via Feynman diagram expansion.

In the setting of Section \ref{sec: SDR data for free BV theory}, consider a deformation of the BV action on $\F$, $S_0\ra S=S_0+S_\ii$, where $S_\ii=\sum_{n\geq 3}S_n$, with $S_n$ a polynomial of degree $n$ on $\F$. We assume that $S$ satisfies QME (\ref{QME}) or equivalently 
\begin{equation}
Q_0 S_\ii+\frac12\{S_\ii,S_\ii\}-\iii\hbar\Delta S_\ii=0.
\end{equation}
We denote $Q_\ii=\{S_\ii,-\}$ -- a derivation of $\fun$ (a formal vector field on $\F$). We have 
\begin{equation}\label{Q_int = sum Q_n}
Q_\ii=\sum_{n\geq 3}Q_n
\end{equation} 
with $Q_n=\{S_n,-\}$
-- a vector field on $\F$, with components of polynomial degree $n-1$.

%\marginpar{Need to change from $S\F^*$ to $\wh{S}\F^*$ or $S\F^*[[g]]$ (with $S=S_0+g S_\ii$) or $\wh{S}\F^*[[\hbar]]$ or $S\F^*[[g,\hbar]]$}

Consider the deformation of differential $Q_{0,\hbar}\ra Q_\hbar=Q_{0,\hbar}+Q_\ii=\{S,-\}-\iii\hbar\Delta$ on $\fun$. It induces a deformation of SDR data of free quantum theory (\ref{ihbar phbar Khbar}) by homological perturbation lemma to the SDR
\begin{equation}\label{SDR for interacting quantum theory}
    \SDR{(\fun,Q_\hbar)}{(\funprime,Q'_\hbar)}{i_{\ii}}{p_{\ii}}{K_{\ii}}.
\end{equation}

%\subsection{SDR of the interacting theory and BV pushforward. Cable diagrams for $Q'_\hbar$, $p_\ii$ as Feynman diagrams}
\begin{thm}\label{thm: Q' and p_int via BV pushforward}
\!\!\!\footnote{Part (i) of the theorem is proven in \cite{SS}, in the case of scalar theory perturbed by a polynomial potential. This reference gives a combinatorial proof of the fact that the coefficients of cable diagrams and Feynman graphs agree. Our proof of that fact is based on a different idea, see Remark \ref{rem: two proofs of Thm 5.1 imply correct coeffs}.
%Also, after posting the preprint, we were pointed to 
Paper \cite{DJP} contains a different proof of part (ii) of the theorem (Theorem 3 in loc. cit.). 
%is contained in \cite[Theorem 3]{DJP}. 
Part (i) of the theorem is present in \cite{DJP} as well, obtained from part (ii) by essentially the same argument as the one we give in the beginning of Section \ref{ss: TQM proof of thm on Q'_hbar and p_int}.
        } 
    \begin{enumerate}[(i)]
        \item
        %\marginpar{\bl May 3 footnote added.}
        The induced differential $Q'_\hbar$ in (\ref{SDR for interacting quantum theory}) has the form
        \begin{equation}\label{Q'_hbar=(S',-)+Delta'}
            Q'_\hbar=\{S',-\}'-\iii\hbar\Delta'
        \end{equation}
        with $S'$ the effective action induced on $\F'$ via %perturbative 
        BV pushforward (\ref{S' definition via BV pushforward}), with Lagrangian (\ref{L=im(kappa)}) determined by chain homotopy $\kappa$.
        \item The deformed projection in $p_\ii$ is the %perturbative 
        BV pushforward map of observables (\ref{map O to O'}),
        \begin{equation}\label{p_int}
            p_\ii\colon O\mapsto O'=\mu'^{-1}e^{-\frac{\iii}{\hbar}S'}\int_\LL \mu\, e^{\frac{\iii}{\hbar}S}O.
        \end{equation}
    \end{enumerate}
\end{thm}

%\marginpar{\bl Cor added May 14}
An immediate consequence of (ii) of the theorem is the statement in the title of the paper:
\begin{cor}\label{cor: P_* is q-iso}
    The BV pushforward map $P_*$ (\ref{map O to O'}) is a quasi-isomorphism of BV complexes.
\end{cor}
\begin{proof}
Indeed, $P_*=p_\ii$, and $p_\ii$, being part of an SDR (\ref{SDR for interacting quantum theory}), is automatically a quasi-isomorphism.
\end{proof}

\begin{rem}[Alternative SDR from conjugation]\label{rem: SDR for interacting theory from conjugation}
    Note that one can also obtain ``alternative'' SDR data for the interacting theory from the free quantum case (\ref{ihbar phbar Khbar}) by %conjugation by $e^{\frac{\iii}{\hbar}S_\ii}$
    exploiting the chain isomorphisms of BV complexes
    \begin{eqnarray}
        (\mr{Fun}(\F), Q_{\hbar} %Q_0-\iii \hbar\Delta
        ) &\xra{\cdot e^{\frac{\iii}{\hbar}S_\ii}}&
        (\mr{Fun}(\F),Q_{0,\hbar}),\\ 
        (\mr{Fun}(\F'), Q'_{\hbar} %Q_0-\iii \hbar\Delta
        ) &\xra{\cdot e^{\frac{\iii}{\hbar}(S'-S'_0)}}&
        (\mr{Fun}(\F'),Q'_{0,\hbar}).
    \end{eqnarray}
    (Here we are being intentionally vague with the classes of allowed functions on $\F,\F'$ and allowed dependence on $\hbar$.)
    Transported along these isomorphisms, the SDR dta (\ref{ihbar phbar Khbar}) yields
    \begin{equation}\label{SDR for interacting theory from conjugation}
        \SDR{(\mr{Fun}(\F,Q_\hbar))}{(\mr{Fun}(\F'),Q'_\hbar)}{i^\mr{conj}_\ii}{p^\mr{conj}_\ii}{K^\mr{conj}_\ii}
    \end{equation}
    with 
    \begin{eqnarray}
        i^\mr{conj}_\ii\colon O' &\mapsto& e^{-\frac{\iii}{\hbar}S_\ii}i( e^{\frac{\iii}{\hbar}(S'-S'_0)}O'),\\
        p^\mr{conj}_\ii\colon O &\mapsto & e^{-\frac{\iii}{\hbar}(S'-S'_0)} p_\hbar (e^{-\frac{\iii}{\hbar}S_\ii}O)=\mu'^{-1}e^{-\frac{\iii}{\hbar}S'} \int_\LL \mu\, e^{\frac{\iii}{\hbar}S} O
        ,\\
        K^\mr{conj}_\ii\colon O &\mapsto & e^{-\frac{\iii}{\hbar}S_\ii} K_\hbar
        (e^{\frac{\iii}{\hbar}S_\ii} O).
    \end{eqnarray}
    Comparing this ``conjugation'' SDR data with the one obtained by homological perturbation (\ref{SDR for interacting quantum theory}), we have that:
    \begin{itemize}
        \item The induced differentials coincide -- it is $Q'_\hbar$ in both cases, by (i) of Theorem \ref{thm: Q' and p_int via BV pushforward}.
        \item The projections coincide $p^\mr{conj}_\ii=p_\ii$, by (ii) of Theorem \ref{thm: Q' and p_int via BV pushforward}.
        \item Inclusions $i^\mr{conj}_\ii$ and $i_\ii$ generally do not coincide. Moreover, the first one does not map formal power series in $\hbar$ to formal power series in $\hbar$, and in particular does not have a ``classical limit'' -- truncation $\bmod\; \hbar$. 
        %\marginpar{this is a bit rambly. edit.}
        On the other hand $i_\ii$ has a classical limit, which moreover is a morphism of commutative dg algebras, and thus is a pullback by an $L_\infty$ morphism (a nonlinear map) $\pi_\ii\colon \F\ra \F'$, see Corollary \ref{cor: class limit of interacting ipK via flow of H^cl} (\ref{cor: class limit of interacting ipK via flow of H^cl (b)}) and Remark \ref{rem: L_infty language} below.
        \item Chain homotopies $K_\ii^\mr{conj}$ and $K_\ii$ generally do not coincide. For instance, the former one does not vanish on constants while the latter one does. Also, as in the case of $i_\ii^\mr{conj}$ above, $K_\ii^\mr{conj}$ does not map formal power series in $\hbar$ to formal power series in $\hbar$, while $K_\ii$ does, see Example \ref{example: Kint^conj(1) in toy interacting scalar}.
    \end{itemize}
    %\marginpar{maybe add example: toy interacting scalar $\frac{x^2}{2}+P(x)$}
\end{rem}

\begin{rem}\label{rem: two proofs of Thm 5.1 imply correct coeffs}
Here below we will give a sketch of a diagrammatic proof of Theorem \ref{thm: Q' and p_int via BV pushforward}. The ``sketchy'' part is that we don't check that coefficients of the diagrams agree. Later, in Section \ref{ss: TQM proof of thm on Q'_hbar and p_int}, we will give an independent proof based on a presentation of homological perturbation theory via topological quantum mechanics.
%a proof (not a sketch) based on a different idea (topological quantum mechanics). 
In particular, it implies that the coefficients of diagrams in the sketch below agree.
\end{rem}

\begin{proof}[Sketch of proof of Theorem \ref{thm: Q' and p_int via BV pushforward}]
Item (i): by homological perturbation lemma, we have
\begin{equation}\label{Q'_hbar HPL}
    Q'_\hbar=Q'_{0,\hbar}+p_\hbar Q_\ii i-p_\hbar Q_\ii K_\hbar Q_\ii i+\cdots
\end{equation}
Contributions to terms in this series (excluding $Q'_{0,\hbar}$) are given by cable diagrams of the form shown in Figure \ref{fig:Q'_hbar cable diagram}. 
    \begin{figure}[h]
        \centering
        \includegraphics[width=0.4\linewidth]{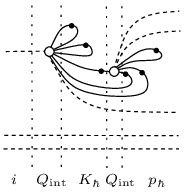}
        \caption{A typical cable diagram for a term in $Q'_\hbar$.}
        \label{fig:Q'_hbar cable diagram}
    \end{figure}
Graphically, $Q_n$ splits a thread into $n-1$ threads (each of the incoming and outgoing threads for $Q_n$ can be infrared or ultraviolet).\footnote{\label{footnote: black dot property}
An important property of admissible cable diagrams for $Q'_\hbar$ (and for $i_\ii$ below) is that the black dot $\kappa^\vee$ occurs only in two possible situations: (a) just before the merging of two threads on one of the merging threads, by the mechanism of Figures \ref{fig:phbar cable diagram}, \ref{fig:Khbar cable diagrams} (or equivalently formulae (\ref{thm SDR for free quantum BV complexes, eq4}), (\ref{Khbar via H})) and (b) on a thread that is about to split by $Q_\ii$, just before $Q_\ii$. One can show the latter by looking at the transpose of terms in (\ref{Q'_hbar HPL}). (Note that passing to the transpose corresponds to reading the cable diagram right-to-left instead of left-to-right.)
One has 
$S\F'\xra{p_\hbar^\vee} S\F'\otimes S\LL \xra{Q_\ii^\vee}S\F'\otimes S\LL\otimes \F \xra{K_\hbar^\vee}   S\F'\otimes S\LL \xra{Q_\ii^\vee}\cdots \xra{Q_\ii^\vee} S\F'\otimes S\LL\otimes \F \xra{K_\hbar^\vee}   S\F'\otimes S\LL
$.
Each time the separate $\F$ factor appears, it corresponds to the input of $Q_\ii$, and that is the only factor which $\kappa$ (the rightmost dot in Figure \ref{fig:Khbar cable diagrams}) from the subsequent $K^\vee$ can hit.
}
Note that here a nontrivial transformation occurs only to one input thread,\footnote{Proven by induction, going left-to-right in the cable diagram: 
% one has 
% $\F'\odot S^n\F'^* \xra{Q_\ii} S\F^*\odot S^n \F'^* \xra{K_\hbar} S\F^*\odot S^n\F'^* \cdots \xra{K_\hbar} S\F^*\odot S^n\F'^* \xra{p_\hbar} S\F'^*\odot S^n \F'^* $.
% -- 
The first $Q_\ii$ splits one thread, acting as identity on the remaining dashed threads $T_1,\ldots,T_n$. The subsequent $K_\hbar$ also acts trivially on $T_i$ and puts the rightmost black dot $\kappa^\vee$ on one of the output threads of the first $Q_\ii$, which then has to be followed by the second $Q_\ii$ by footnote \ref{footnote: black dot property}. Repeating the argument, threads $T_i$ are left untransformed by all subsequent operations, and finally are left untransformed by $p_\hbar$ (cf. Figure \ref{fig:phbar cable diagram}).
}
thus the diagram determines a derivation of $\fun$ (a formal vector field on $\F'$). Such a diagram can be seen as representing $\{S'_\Gamma,O'\}'$ -- the action of the Hamiltonian vector field generated by a contribution $S'_\Gamma$ of a connected Feynman graph $\Gamma$ to the effective action (\ref{S' definition via BV pushforward}) on an element $O'\in \funprime$ (an infrared observable\footnote{Here when we say ``observable,'' we just mean an element of $\funprime$, not necessarily a $Q'_\hbar$-closed element (i.e., we don't insist on gauge-invariant observables).}), see Figure \ref{fig:Q'_hbar Feynman diagram}. 
%\marginpar{Is it ok to have this ambiguity with the word ``observable''}

The translation of a cable diagram to a Feynman graph is by pinching the left side of the cable diagram and pairing to $O'$, interpreting solid lines with a black dot as propagators $h$, and forgetting the left-to-right tiered structure of the cable diagram.
\begin{figure}
    \centering
    \includegraphics[width=0.5\linewidth]{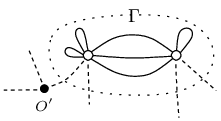}
    \caption{Hamiltonian vector field generated by the contribution of a Feynman graph  $\Gamma$ to the effective action $S'$ acting on an infrared observable $O'$.}
    \label{fig:Q'_hbar Feynman diagram}
\end{figure}
The reverse translation of a connected Feynman graph $\Gamma$ (with external dashed edges and internal solid edges) to a cable diagram is as follows: choose a ``height function'' (or ``Morse function'') $\chi$ on $\Gamma$ (seen as a 1d cell complex), satisfying the following: 
\begin{itemize}
    \item $\chi$ is a continuous function on $\Gamma$ valued in an interval $[a,b]$.
    %smooth on edges and continuous at vertices.
    \item $\chi=a$ at the tip of one dashed edge, $\chi=b$ at the tips of the other dashed edges. 
    \item $\chi$ is monotonous on dashed edges. On solid edges it is either monotonous or has a unique local maximum in the interior of the edge.
    \item For every vertex $v$, $\chi$ is decreasing along exactly one of the incident edges (dashed or solid).
\end{itemize}
Such a Morse function converts $\Gamma$ into a cable diagram, where cross-sections of the cable are given by level sets of $\chi$. We understand that a black dot (corresponding to an insertion of $\kappa^\vee$ on the thread, or $K$ on the cable) should be inserted just before the merging of two threads on one of them, and also on a thread that is about to split, see Figure \ref{fig:Q'_hbar cable diagram}.
Local maxima of $\chi$ on edges correspond to merging a pair of threads with  $\omega''^{-1}$  (or insertion of $\Delta$ on the cable), see Figure \ref{fig:phbar cable diagram}.  Vertices of $\Gamma$ correspond to $Q_\ii$ factors in (\ref{Q'_hbar HPL}).\footnote{
An equivalent construction is to 
%As a variant of the construction, one can instead 
consider a combinatorial Morse function $\bar\chi$ on $\Gamma$ in the sense of R. Forman, assigning values in $[a,b]$ to vertices and edges, such that for each edge $e$,  $\bar\chi(e)$ is smaller than at $\bar\chi$ on at most one of the incident vertices, and for each vertex $v$, $\bar\chi(v)$ is greater than $\bar\chi$ on at most one incident edge.
}

This process describes the creation of the nontrivial part of the cable diagram of $Q'_\hbar$, to which one can add an arbitrary number of dashed threads passing through, like the bottom two threads in Figure \ref{fig:Q'_hbar cable diagram}.

% assigning distinct real values to vertices and edges, assigning $\chi=-\infty$ to one dashed edge and $\chi=+\infty$ to all other dashed edges, and satisfying the following:
% \begin{itemize}
%     \item For each vertex $v$, there is exactly one incident edge $e$ with $\chi(e)<\chi(v)$.
%     \item For each internal edge $e$ connecting vertices $v_1$ and $v_2$, one has either 
%     \begin{enumerate}[(a)]
%         \item $\chi(v_1)<\chi(e)<\chi(v_2)$ or 
%     \item $\chi(v_1)<\chi(e)>\chi(v_2)$.
%     \end{enumerate}
%     The situation $\chi(v_1)>\chi(e)<\chi(v_2)$ is not allowed.
% \end{itemize}

In summary, terms in (\ref{Q'_hbar HPL}) (cable diagrams) match the terms in the Feynman graph expansion of (\ref{Q'_hbar=(S',-)+Delta'}), $Q'_\hbar=Q'_{0,\hbar}+\sum_\Gamma \{S'_\Gamma,-\}'$,  
which proves item (i) of the theorem.

Item (ii):
By homological perturbation lemma, we have 
\begin{equation}\label{p_int HPL}
    p_\ii=p_\hbar-p_\hbar Q_\ii K_\hbar+p_\hbar Q_\ii K_\hbar Q_\ii K_\hbar-\cdots
\end{equation}
Terms here correspond to cable diagrams like the one in Figure \ref{fig:p_int cable diagram}. Note that footnote \ref{footnote: black dot property} applies here as well.

\begin{figure}[h]
    \centering
    \includegraphics[width=0.7\linewidth]{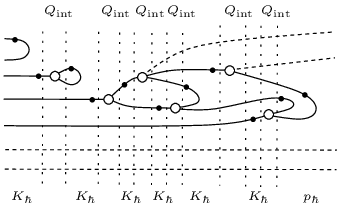}
    \caption{Typical cable diagram for a term in $p_\ii$.}
    \label{fig:p_int cable diagram}
\end{figure}

Acting on an element $O\in \fun$ (an input for the left side of the cable), such a diagram reproduces a Feynman graph in the perturbative expansion of (\ref{p_int}), see Figure \ref{fig:p_int Feynman graph}. A Feynman graph here is a connected graph with a distinguished vertex decorated by $O$, other vertices decorated by $S_n$ (with $n$ the valence of the vertex), external (dashed) edges allowed (decorated by the input in $\F'$) and internal (solid) edges decorated by the propagator $h$.

\begin{figure}[h]
    \centering
    \includegraphics[width=0.35\linewidth]{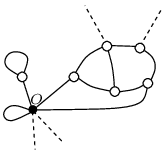}
    \caption{A Feynman graph for the integral (\ref{p_int}) computing the BV pushforward of an observable $O$.}
    \label{fig:p_int Feynman graph}
\end{figure}

Translation of a cable diagram to a Feynman graph is similar to (i). The reverse process is also similar, with an amendment that $\chi$ assigns the (minimal) value $a$ to the $O$-vertex and (maximal) value $b$ to tips of dashed edges.

Again, terms in (\ref{p_int HPL}) match the Feynman graph contributions to (\ref{p_int}), which proves (ii) of the theorem.
\end{proof}

%\subsection{Observable-lifting map \texorpdfstring{$i_\ii$}{iint}}

\begin{rem} 
The observable-lifting map $i_\ii$ and the chain homotopy $K_\ii$ in (\ref{SDR for interacting quantum theory}) are more complicated than $Q_\hbar$ and $p_\ii$, and the corresponding cable diagrams cannot be immediately interpreted as Feynman graphs for a BV integral for the theory $(\F,\omega,S)$, see Figures \ref{fig:i_int cable diagram}, \ref{fig:K_int cable diagram}. Nevertheless, in Section %\ref{ss: TQM as 1d AKSZ} 
\ref{sss: cable diagrams for i_int via 1d AKSZ}
below we will interpret them as Feynman diagrams for the BV integral for a different theory build out of $(\F,\omega,S)$.
\begin{figure}
    \centering
    \includegraphics[width=0.7\linewidth]{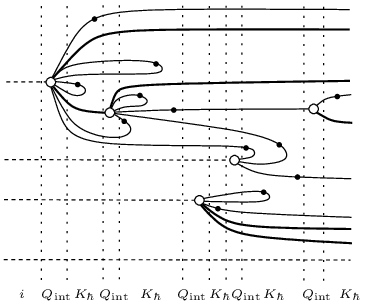}
    \caption{
    %{\color{red} Correction: I think this is not a general type of diagram for $i_\ii$. A more general one can contain dashed edges between $Q_\ii$ vertices and possibly solid edges (though those may cancel out by CME).} 
    A typical cable diagram for $i_\ii$. Thick threads stand for elements of $\F^*$ (i.e., thick thread = dashed thread $\F'^*$ plus thin thread $\F''^*$).}
    \label{fig:i_int cable diagram}
\end{figure}
\begin{figure}
    \centering
    \includegraphics[width=0.55\linewidth]{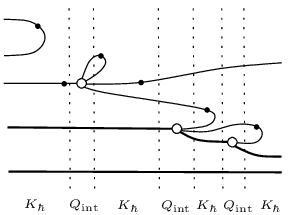}
    \caption{A typical cable diagram for $K_\ii$.}
    \label{fig:K_int cable diagram}
\end{figure}
\end{rem}

% \begin{figure}
%     \centering
%     \includegraphics[width=0.5\linewidth]{i_int_Feynman_graph_4.eps}
%     \caption{Cable diagram of Figure \ref{fig:i_int cable diagram} drawn as a ``Feynman'' graph}
%     \label{fig:placeholder}
% \end{figure}

%\subsection{An alternative SDR for the interacting theory}

\subsection{Example: toy interacting quantum scalar field}
\label{sss example: toy interacting quantum scalar}
%\begin{expl}
    Consider the toy scalar theory (Example \ref{example: toy scalar S_0=x^2/2} and Section \ref{sss: example -- quantum toy free scalar}) and perturb the action $S_0(x)=\frac{x^2}{2}$ 
    to
    \begin{equation}
    S(x)=\frac{x^2}{2}+P(x)
    \end{equation} 
    with $P(x)=S_\ii(x)=\sum_{n\geq 3} \frac{P_n}{n!}x^n$ a polynomial potential. The corresponding perturbation of the BV differential is 
    \begin{equation}
    Q_{0,\hbar}= x\frac{\dd}{\dd \xi}-\iii\hbar \frac{\dd^2}{\dd x\dd \xi} \quad \ra\quad 
    Q_\hbar= x\frac{\dd}{\dd \xi}-\iii\hbar \frac{\dd^2}{\dd x\dd \xi}+\underbrace{P'(x)\frac{\dd}{\dd \xi}}_{Q_\ii}
    \end{equation}
    In this case the SDR (\ref{SDR for interacting quantum theory}) is
    \begin{equation}\label{SDR interacting toy scalar}
        \SDR{\left(\CC[[x,\xi]][[\hbar]], Q_\hbar\right)}{\CC[[\hbar]]}{i_\ii=i}{p_\ii}{K_\ii}.
    \end{equation}
    Here $i_\ii$ is still the tautological inclusion of constants.
    The projection
    \begin{equation}
        p_\ii\colon O(x) \mapsto c \int_{-\infty}^\infty \dr x\, e^{\frac{\iii}{\hbar}(\frac{x^2}{2}+P(x))} O(x)
    \end{equation}
    -- expectation value of an observable w.r.t. the perturbed action, with $c$ as in (\ref{c toy scalar}).

    Consider the chain homotopy $K_\ii$. %It suffices to consider it on observables of the form $O(x)$, since 
    It vanishes for degree reason on elements of the form $O(x,\xi)=\xi f(x)$, so it suffices to inspect $K_\hbar$ on observables of the form $O(x)$. Using (\ref{Khbar toy scalar 2}), we find
    \begin{multline}
        O(x)\xra{K_\hbar} \xi \mc{R}\left(1-\frac{\iii\hbar}{x} %\dd_x \frac{\dd}{\dd x}
        \right)^{-1}\frac{O(x)}{x} \\
        \xra{Q_\ii} 
        P'(x) \mc{R} \left(1-\frac{\iii\hbar}{x} \frac{\dd}{\dd x}\right)^{-1}\frac{O(x)}{x}\\
        \xra{K_\hbar}
        \xi \mc{R}\left(1-\frac{\iii\hbar}{x} \frac{\dd}{\dd x}\right)^{-1}\frac{P'(x)}{x} \mc{R}\left(1-\frac{\iii\hbar}{x} \frac{\dd}{\dd x}\right)^{-1}\frac{O(x)}{x} \ra\cdots
        %\\\ra (K_\hbar Q_\ii)^n K_\hbar O = \xi \mc{R} \left(\left(1-\frac{\iii\hbar}{x} \frac{\dd}{\dd x}\right)^{-1}\frac{P'(x)}{x} \mc{R}\right)^n\left(1-\frac{\iii\hbar}{x} \frac{\dd}{\dd x}\right)^{-1}\frac{O(x)}{x}
    \end{multline}
    where the operator $\mc{R}$ subtracts the singular part of the Laurent series at $x=0$, as $[\cdots]_\mr{reg}$ in (\ref{Khbar toy scalar 2}). Thus, we have
    \begin{multline}\label{Kint toy scalar}
        K_\ii\colon O(x)\mapsto (K_\hbar-K_\hbar Q_\ii K_\hbar+K_\hbar Q_\ii K_\hbar Q_\ii K_\hbar-\cdots)(O)\\
        =\sum_{n\geq 0} \xi \mc{R}\left(-\left(1-\frac{\iii\hbar}{x} \frac{\dd}{\dd x}\right)^{-1}\frac{P'(x)}{x}\mc{R}\right)^n \left(1-\frac{\iii\hbar}{x} \frac{\dd}{\dd x}\right)^{-1}\frac{O(x) }{x}\\
        =\xi \mc{R}\left(1-\frac{\iii\hbar}{x}\frac{\dd}{\dd x}+\frac{P'(x)}{x}\mc{R} \right)^{-1}\frac{O(x)}{x} .
    \end{multline}
We make the following remarks:
\begin{enumerate}[(a)]
\item Formula (\ref{Kint toy scalar}) has a simple limit $\hbar\ra 0$ (classical limit):\footnote{Here and later when discussing toy interacting scalar field, the prime on $S$ denotes the derivative in $x$, not the effective action.} 
%\marginpar{\bl May 10: $\dd_x S(x) \ra S'(x)$?}
\begin{equation}
    K_\ii^\cl(O)=\xi \frac{O(x)-O(0)}{x+P'(x)}=\xi \frac{O(x)-O(0)}{ S'(x)}.
\end{equation}
    \item If $P(x)$ is an even function (contains only even powers of $x$) and $O(x)$ is an odd function, $\mc{R}$ acts tautologically in (\ref{Kint toy scalar}) and one has
    \begin{equation}\label{Kint toy scalar for S even, O odd}
        K_\ii(O)= \xi \left(1-\frac{\iii\hbar}{x}\frac{\dd}{\dd x}+\frac{P'(x)}{x} \right)^{-1}\frac{O(x)}{x} 
    \end{equation}
\end{enumerate}

Let us compare the ``conjugation SDR'' (\ref{SDR for interacting theory from conjugation}) of Remark \ref{rem: SDR for interacting theory from conjugation}  with the SDR (\ref{SDR interacting toy scalar}) coming from homological perturbation theory. It is obvious that $i_\ii^\mr{conj}=i_\ii$, $p_\ii^\mr{conj}=p_\ii$. As for $K_\ii^\mr{conj}$, we first observe that 
    \begin{equation}
        e^{-\frac{\iii}{\hbar}P(x)} \frac{\dd}{\dd x} e^{\frac{\iii}{\hbar}P(x)} = \frac{\iii}{\hbar} P'(x)+\frac{\dd}{\dd x}.
    \end{equation}
Thus, we have
\begin{multline}\label{Kint^conj toy scalar}
    K_\ii^\mr{conj}(O)=e^{-\frac\iii\hbar P(x)} K_\hbar (e^{\frac\iii\hbar P(x)}O(x))\\
    = \xi e^{-\frac\iii\hbar P(x)} \mc{R}\left(1-\frac{\iii\hbar}{x}\frac{\dd}{\dd x} \right)^{-1}e^{\frac\iii\hbar P(x)}\frac{O(x)}{x}\\
    =\xi e^{-\frac\iii\hbar P(x)} \mc{R}e^{\frac\iii\hbar P(x)}\left(1-\frac{\iii\hbar}{x}\frac{\dd}{\dd x}+\frac{P'(x)}{x} \right)^{-1}\frac{O(x)}{x}. 
\end{multline}

In the special case when $S(x)$ is an even function and $O(x)$ is an odd function, $\mc{R}$ in (\ref{Kint^conj toy scalar}) acts tautologically and one has
\begin{equation}\label{toy interacting scalar Kint=Kint^conj for P even, O odd}
    K_\ii^\mr{conj}(O)=K_\ii(O),
\end{equation}
cf. (\ref{Kint toy scalar for S even, O odd}).

Without the assumption that $S$ is even and $O$ is odd, $K_\ii^\mr{conj}(O)$ generally does not coincide with $K_\ii(O)$. Moreover, $K_\ii^\mr{conj}(O)$ is generally not a power series in $\hbar$ (and so does not have a classical limit).

\begin{expl}\label{example: Kint^conj(1) in toy interacting scalar}
    Let $S=\frac{x^2}{2}+g x^4$, $O=1$. Then, expanding $K_\ii^\mr{conj}(1)$ in powers of the coupling constant $g$, we have
    \begin{multline}
        K_\ii^\mr{conj}(1)=\xi \left( \frac\iii\hbar g(x^3+\iii\hbar\, 3x)\right. \\
        \left.+\left(\frac\iii\hbar\right)^2 \frac{g^2}{2}\left(-x^7+\iii\hbar x^5+(\iii\hbar)^27\cdot 5 x^3 +(\iii\hbar)^3 7\cdot 5\cdot 3 x\right) + \mc{O}(g^3)
        \right),
    \end{multline}
    which contains negative powers of $\hbar$.
\end{expl}

%\marginpar{Added Apr 11}
\subsubsection{An alternative computation (and another formula) for $K_\ii$.}
We are looking for an operator $K_\ii\colon O(x)\mapsto \xi \ul{K}O,\, \xi O(x) \mapsto 0$ satisfying
\begin{equation}\label{toy interacting scalar [K,Q]=id-ip}
    [K_\ii,Q_\hbar]=\mr{id}-i\circ p_\ii
\end{equation}
with $Q_\hbar=\left(S'(x)-\iii\hbar \frac{\dd}{\dd x}\right)\frac{\dd}{\dd\xi}$.
%\footnote{Here and later when discussing toy interacting scalar field, the prime on $S$ denotes the derivative in $x$, not the effective action.} 
Equation (\ref{toy interacting scalar [K,Q]=id-ip}), by specializing to functions $O(x)$ and $\xi O(x)$, is equivalent to the system
\begin{eqnarray}\label{toy interacting scalar K system 1}
    \left(S'(x)-\iii\hbar \frac{\dd}{\dd x}\right)\ul{K} O&=& O(x)-c_\ii \int_\RR \dr y\, e^{\frac{\iii}{\hbar} S(y)}O(y),\\
    \label{toy interacting scalar K system 2}
    \ul{K}\left(S'(x)-\iii\hbar \frac{\dd}{\dd x}\right)O &=&O,
\end{eqnarray}
with $c_\ii=\left(\int_\RR \dr x\, e^{\frac{\iii}{\hbar}S(x)}\right)^{-1}$. Note that  the differential operator
\begin{equation}
\left(S'(x)-\iii\hbar \frac{\dd}{\dd x}\right)=e^{-\frac\iii\hbar S} (-\iii\hbar)\frac{\dd}{\dd x} e^{\frac\iii\hbar S}
\end{equation}
contained in $Q_\hbar$ is a conjugation of the derivative. Let $\ul{K}=e^{-\frac\iii\hbar S} K^c e^{\frac\iii\hbar S}$ for some operator $K^c$. The system (\ref{toy interacting scalar K system 1}), (\ref{toy interacting scalar K system 2}) reads
\begin{equation}
    -\iii\hbar \frac{\dd}{\dd x} K^c \til{O}=\til{O}-e^{\frac{\iii}{\hbar}S}c_\ii \int_\RR \dr y\, \til{O}(y),
    \qquad K^c\left(-\iii\hbar \frac{\dd}{\dd x}\right) \til{O}=\til{O},
\end{equation}
where $\til{O}=e^{\frac\iii\hbar S}O$.
This system has a \emph{unique} solution\footnote{Uniqueness of chain-homotopy is a feature of two-term complexes: Let $\SDR{(V,d)}{(V',d')}{i}{p}{K}$ be an SDR of chain complexes  with $V$ concentrated in two neighboring degrees $n$, $n+1$. Then if $i$ and $p$ are fixed, $K$ is determined uniquely (is non-deformable). In the case at hand, $(V,d)$ is the BV complex in the left side of (\ref{SDR interacting toy scalar}), concentrated in degrees $-1,0$.
%complex $V$ concentrated in two neighboring degrees
%(once $\iota$ and $\pi$ are fixed) is a special property of two-term chain complexes $\F,d$, concentrated in degrees $0$ and $1$.
}
\begin{equation}
    K^c(\til{O})(x)=\frac\iii\hbar \left(\int^x_{-\infty} \dr y\, \til{O}(y)-\int^x_{-\infty} \dr y\, e^{\frac\iii\hbar S(y)} c_\ii \int_\RR \dr z\, \til{O}(z)\right)
\end{equation}
%\marginpar{remind about the integration contour}
and thus $K_\ii$ is
\begin{multline}\label{toy interacting scalar Kint integral formula}
    K_\ii(O)(x)=\xi \frac\iii\hbar e^{-\frac\iii\hbar S(x)} \int^x_{-\infty} \dr y\, e^{\frac\iii\hbar S(y)}\left(O(y)-c_\ii \int_\RR \dr z\, e^{\frac\iii\hbar S(z)}O(z)\right)\\
    = \xi \int_\RR \dr y\, G(x,y)O(y)
\end{multline}
with the ``Green's function''
\begin{equation}\label{toy interacting scalar: Green's function}
    G(x,y)=\frac\iii\hbar e^{\frac\iii\hbar(-S(x)+S(y))}\left(\theta(x-y)-c_\ii \int^x_{-\infty} \dr z\, e^{\frac\iii\hbar S(z)}\right).
\end{equation}
Formula (\ref{toy interacting scalar Kint integral formula}) generalizes the integral formula (\ref{Khbar toy scalar 3}) to the interacting case. By uniqueness, this chain homotopy coincides with (\ref{Kint toy scalar}). 
%Note that from the formula (\ref{toy interacting scalar Kint integral formula}), it is not immediately obvious why it maps formal power series in $\hbar$ to formal power series in $\hbar$ (and has a classical limit), whereas from (\ref{Khbar toy scalar 3}) it is immediate.

For comparison, as follows from (\ref{Khbar toy scalar 3}), $K_\ii^\mr{conj}$ is given by 
\begin{multline}\label{toy interacting scalar Kint^conj integral formula}
    K_\ii^\mr{conj}(O)(x)=\xi \frac\iii\hbar e^{-\frac\iii\hbar S(x)} \int^x_{-\infty} \dr y \left( e^{\frac\iii\hbar S(y)} O(y)- e^{\frac\iii\hbar S_0(y)}c \int_\RR \dr z\, e^{\frac\iii\hbar S(z)}O(z)\right)\\
    = \xi \int_\RR \dr y\, G^\mr{conj}(x,y)O(y)
\end{multline}
with 
\begin{equation}
    G^\mr{conj}(x,y)=\frac\iii\hbar e^{\frac\iii\hbar(-S(x)+S(y))}\left(\theta(x-y)-c \int^x_{-\infty} \dr z\, e^{\frac\iii\hbar S_0(z)}\right).
\end{equation}
Comparing to (\ref{toy interacting scalar: Green's function}), we note that in the second term the interacting action $S(x)$ is replaced with the free action $S_0(x)$ and the normalization constant $c_\ii$ of the interacting theory is replaced with its free counterpart (\ref{c toy scalar}).

We also note that if $P(x)$ is an even function and $O(x)$ is an odd function, the second terms in (\ref{toy interacting scalar Kint integral formula}) and in (\ref{toy interacting scalar Kint^conj integral formula}) vanish and we recover the result (\ref{toy interacting scalar Kint=Kint^conj for P even, O odd}).

% \begin{rem}\label{rem: K_int vs K_int^conj}
%     It is generally not true that $K_\ii$ and $K_\ii^\mr{conj}$ coincide. Analyzing the cable diagrams for the two expressions, one notices that diagrams containing a thick solid thread connecting two $Q_\ii$ vertices (e.g. such a thread is present in Figure \ref{fig:K_int cable diagram}) can generally appear in $K_\ii$ but not in $K_\ii^\mr{conj}$. In the example of toy scalar, such thick solid edges in $K_\ii$ do not appear for a degree reason (a $Q_\ii$ vertex has $\xi$ as input and $x$s as outputs, which cannot be an input for another $Q_\ii$ vertex) -- ultimately because $(\F,d)$ in this case is just a two-term cochain complex.
% \end{rem}

%\marginpar{subsection added Apr 26}
\subsection{Example: lifting Wilson loop observable from Yang--Mills to $BF+B^2$ theory, abelian case}
\label{ss: example - lifting ab Wilson loop}
\subsubsection{Abelian $BF+B^2$ theory}
Consider the abelian $BF+B^2$ theory on a closed oriented 4-manifold $M$ defined by the BV action
\begin{equation}\label{ab BF+BB: S}
    S=\int_M \B \wedge \dr\A+\frac{\lambda}{2} \B\wedge \B
\end{equation}
with $\lambda\neq 0$ a coupling constant. The fields $(\A,\B)$ (the AKSZ superfields) are two nonhomogeneous forms on $M$,
\begin{equation}\label{ab BF+BB: F}
    (\A,\B)\in\;\; \Omega^\bt(M)[1]\oplus \Omega^\bt(M)[2] =\colon \F.
\end{equation}
We introduce the following notations for components of $\A,\B$:
\begin{equation}
    \A=c+A+B^++\tau^++\phi^+,
    %\underset{0,1}{c}+\underset{1,0}{A}+\underset{2,-1}{B^+}+\underset{3,-2}{\tau^+}+\underset{4,-3}{\phi^+},
    \qquad \B= \phi+\tau+B+A^++c^+,
\end{equation}
where form degrees and ghost numbers are as follows:
\begin{center}
\begin{tabular}{c|c|c|c|c|c|c|c|c|c|c}
     & $c$& $A$ & $B^+$ & $\tau^+$ & $\phi^+$ & $\phi$ & $\tau$ & $B$ & $A^+$ & $c^+$  \\ \hline
     form degree & 0 & 1 & 2 & 3 & 4 & 0&1&2&3&4 \\
     ghost number & 1 & 0 & -1&-2&-3&2&1&0&-1&-2
\end{tabular}
\end{center}

The $(-1)$-symplectic form on $\F$ is 
\begin{equation}
    \omega=\int_M \delta \B\wedge \delta \A = \int_M \sum_{\Phi\in \{A,B,c,\tau,\phi\}}\delta \Phi^+\wedge \delta \Phi.
\end{equation}

\subsubsection{Gauge-fixing}
Next, assume that $M$ is equipped with a Riemannian metric $g$ and let $*$ be the corresponding Hodge star operator. We split the 2-forms into self-dual forms $*\alpha=\alpha$ and anti-self-dual forms $*\beta=-\beta$: $\Omega^2(M)=\Omega^2_+(M)\oplus \Omega^2_-(M)$, with projectors to the two summands being $P_\pm=\frac{\mr{id}\pm *}{2}$.\footnote{Note that for $\alpha_+$ a self-dual 2-form and $\beta_-$ an anti-self-dual 2-form, one has $\int_M \alpha_+\wedge \beta_-=0$.}
Thus, we split the 2-forms fields into self-dual and anti-self-dual parts:
\begin{equation}
    B=B_++ B_-,\quad B^+=B_+^++B_-^+.
\end{equation}

The metric $g$ on $M$ also induces the Hodge decomposition of forms into harmonic, exact and coexact, $\Omega^\bt(M)=\Omega_\mr{harm}^\bt(M)\oplus \Omega_\mr{ex}^\bt(M)\oplus \Omega_\mr{coex}^\bt(M)$.

We introduce the infrared and ultraviolet complexes as follows.
First, we introduce two complexes $\F_1,\F_2$:
\begin{equation}
    \F_1=\qquad \begin{array}{ccccccc}
         % \Omega^0_{(c)}& \xra{d}& \Omega^1_{(A)}&\xra{P_+d}& \Omega^2_{+(B_+^+)} &&  \\
         % &&\oplus&&\oplus&& \\
         % && \Omega^2_+&\xra{d}& \Omega^3 &\xra{d}& \Omega^4
         \Omega^0& \xra{\dr}& \Omega^1 &\xra{P_+\dr}& \Omega^2_+ &&  \\
         &&\oplus&&\oplus&& \\
         && \Omega^2_+&\xra{\dr}& \Omega^3 &\xra{\dr}& \Omega^4
    \end{array}
\end{equation}
Here the first line corresponds to fields $c,A,B^+_+$ and the second line to fields $B_+,A^+,c^+$. We define $\F_2$ as the complex
\begin{equation}
    \F_2=\qquad 
    \begin{array}{ccccccccccc}
         \Omega^0 &\xra{\dr} &\Omega^1 &\xra{P_- \dr}& \Omega^2_- &\xra{0}& \Omega^2_-& \xra{\dr} &\Omega^3 &\xra{\dr} &\Omega^4, 
    \end{array}
    %\Omega^0\xra{d} \Omega^1 \xra{P_- d} \Omega^2_- \xra{0} \Omega^2_- \xra{d} \Omega^3 \xra{d} \Omega^4,
\end{equation}
where, going left-to-right, the terms correspond to fields $\phi,\tau,B_-,B_-^+,\tau^+,\phi^+$.

We define the infrared complex as $\F_1$ plus the harmonic part of $\F_2$ (which is isomorphic to cohomology of $\F_2$):
\begin{equation}
    \F'=\F_1\oplus \underbrace{\left( \F_2\cap \Omega^\bt_\mr{harm}\right)}_{\simeq H^\bt(\F_2)}
\end{equation}
and the ultraviolet complex as $\F_2$ with harmonic part removed (so as to make $\F''$ acyclic):
\begin{equation}
    \F''=\F_2\cap (\Omega^\bt_\mr{ex}\oplus \Omega^\bt_\mr{coex}).
\end{equation}
Viewed as a graded vector space, the space of fields (\ref{ab BF+BB: F}) splits as $\F=\F'\oplus \F''$. 
%\marginpar{comment on the differential in $\F$, $d_\F=d'+d''+d_\mr{int}$?}

We can write the action (\ref{ab BF+BB: S}) as
\begin{multline}
    S=\underbrace{\int_M B_+ \dr A + A^+ \dr c}_{S'_0}+\underbrace{\int_M B^+_- \dr \tau+\tau^+\dr \phi}_{S''_0} \\
    + \underbrace{\int_M B_- \dr A+B_+^+ \dr \tau+\lambda \left( \frac12 B_+ B_+ +\frac12 B_-B_- +A^+\tau+c^+\phi\right)}_{S_\ii}.
\end{multline}
Here $S'_0$, $S''_0$ are the free BV actions associated to the complexes $\F'$, $\F''$, and $S_\ii$ is the perturbation.

We want to consider the BV pushforward defined by the gauge-fixing Lagrangian $\LL=\F_2\cap \Omega_\mr{coex}^\bt \subset \F''$ -- the coexact part of the complex $\F''$ (or equivalently $\F_2$). This Lagrangian can be written as $\LL=\mr{im}(\kappa)$ with $\kappa$ the chain contraction of $\F''$ induced by the following contraction of the complex $\F_2$:
\begin{equation}\label{ab BF+BB: kappa}
    %\F_2\colon\qquad 
    \kappa\colon\quad
    \begin{array}{ccccccccccc} 
       \Omega^0 &\xla{\dr^*G} &\Omega^1 &\xla{2\dr^* G}& \Omega^2_- &\xla{0}& \Omega^2_-& \xla{2P_-\dr^* G} &\Omega^3 &\xla{\dr^*G} &\Omega^4.
    \end{array}
\end{equation}
Here $G=(\Delta_g+P_\mr{harm})^{-1}$ is the Green's function; $\Delta_g=[\dr,\dr^*]$ is the Laplace--de Rham operator. The dual of $\kappa$ extended to a derivation is then
\begin{equation}\label{ab BF+BB: kappahat}
    \wh\kappa=\int_M G \dr^*\tau \frac{\delta}{\delta \phi}+2G \dr^*B_- \frac{\delta}{\delta \tau}+ 2P_- G  \dr^*\tau^+ \frac{\delta}{\delta B_-^+}+G\dr^*\phi^+\frac{\delta}{\delta\tau^+}.
\end{equation}

Thus, we have an SDR (\ref{SDR F to F'}) with the inclusion/projection $\iota,\pi$ obvious by construction and with $\kappa$ as in (\ref{ab BF+BB: kappa}).

The effective action on $\F'$, computed by BV pushforward, reads (up to a constant)
\begin{multline}
    S'=\underbrace{\int_M B_+\dr A + A^+ \dr c +\frac{\lambda}{2} B_+ B_+}_{S'_\mr{YM}} \\
    + \underbrace{\lambda\int_M \frac12 B_{-\mr{harm}} B_{-\mr{harm}} +A^+_\mr{harm} \tau_\mr{harm} + c^+_\mr{harm} \phi_\mr{harm}
    }_{S'_\mr{zm}}.
\end{multline}
%Here the subscript ``harm'' stands for the harmonic part of a form.
Here the first term in the abelian Yang--Mills action in the first-order formalism.\footnote{Note that integrating out the field $B_+$ in $S'_\mr{YM}$, one obtains $\displaystyle \int_M -\frac{1}{4\lambda} \dr A\wedge *\dr A + A^+\dr c $ -- the BV action of abelian Yang--Mills in the second-order formalism.} The second term is the contribution of zero-modes (cohomology) of $\F_2$.

\subsubsection{Lifting abelian Wilson loop observable}
\label{sss: lifting ab Wilson loop observable}
Assume that the first Betti number of $M$ vanishes (so that $A^+_\mr{harm}=0$).

Consider the abelian Wilson loop  
\begin{equation}
    W_\gamma(A)=\exp \left(\iii\alpha \oint_\gamma A\right)
\end{equation}
as an element of the infrared BV complex.
Here $\gamma\subset M$ a closed curve and $\alpha\in \mathbb{R}$ a coupling constant. Note that $W_\gamma$ is a BV cocycle: $\{S',W_\gamma\}'-\iii\hbar \Delta' W_\gamma=0$.\footnote{Here we use the assumption of vanishing first Betti number. Otherwise, we would have $\displaystyle \{S',W_\gamma\}'-\iii\hbar\Delta'W_\gamma=\iii\alpha \lambda \left(\oint_\gamma \tau_\mr{harm}\right)\cdot W_\gamma(A)\neq 0$. 
Note that instead of assuming $B_1=0$ we could have assumed that $\gamma$ is trivial in the first homology of $M$ -- that would also guarantee that $W_\gamma$ is a BV cocycle.
}

\begin{prop}\label{prop: lifting ab Wilson loop}
    The lifting of $W_\gamma$ to abelian $BF+B^2$ theory is given by
    \begin{equation}\label{ab BF+BB: i_int(W)}
        i_\ii(W_\gamma)=\exp \left(\iii\alpha \oint_\gamma (A+\lambda 2 G \dr^*B_-)\right).
    \end{equation}
\end{prop}
\begin{proof}
%Let us compute $i_\ii(W_\gamma)$ -- the lifting of $W_\gamma$ to the abelian $BF+B^2$ theory. 
By homological perturbation lemma, we have
\begin{equation}\label{ab BF+BB: HPT series}
    i_\ii(W_\gamma)=(i-K(Q_\ii-\iii\hbar \Delta) i + K(Q_\ii-\iii\hbar \Delta)K(Q_\ii-\iii\hbar \Delta) i-\cdots) W_\gamma
\end{equation}
with 
\begin{multline}
    Q_\ii=\{S_\ii,-\}= \int_M (P_+ \dr\tau) \frac{\delta}{\delta B_+}+ (P_- \dr A) \frac{\delta}{\delta B_-^+}+\dr B_- \frac{\delta}{\delta A_+} + \dr B_+^+ \frac{\delta}{\delta \tau^+}\\
    +\lambda \left(
    \phi\frac{\delta}{\delta c}+\tau \frac{\delta}{\delta A}+B_+\frac{\delta}{\delta B_+^+}+B_-\frac{\delta}{\delta B_-^+}+A^+\frac{\delta}{\delta \tau^+}+c^+\frac{\delta}{\delta \phi^+}
    \right).
\end{multline}
Next we compute the terms in the series (\ref{ab BF+BB: HPT series}):
\begin{multline}
    W_\gamma(A) \xra{Q_\ii-\iii\hbar\Delta} \iii\alpha\left(-\oint_\gamma \lambda \tau\right)W_\gamma(A) \xra{K} \iii\alpha\left(-\oint_\gamma \lambda 2G\dr^* B_-\right)W_\gamma(A) \\
    \xra{Q_\ii-\iii\hbar\Delta} (\iii\alpha)^2\left(-\oint_\gamma \lambda 2G\dr^* B_-\right)\left(-\oint_\gamma \lambda \tau\right)W_\gamma(A) 
    \xra{K} 
    \frac12  (\iii\alpha)^2\left(-\oint_\gamma \lambda 2G\dr^* B_-\right)^2 W_\gamma(A)\\
    \cdots \ra (K(Q_\ii-\iii\hbar\Delta))^k (W_\gamma)=
    \frac{1}{k!}(\iii\alpha)^k \left(-\oint_\gamma \lambda 2G\dr^* B_-\right)^k W_\gamma(A).
\end{multline}
Here we use the formula (\ref{Ksym=nu whK}) for $K$, $K=\nu\wh\kappa$, and the operator $\nu$ is responsible for the appearance of the factor $\frac{1}{k!}$ above. Thus, we have
\begin{multline}
    i_\ii(W_\gamma)= \sum_{k\geq 0}\frac{(\iii\alpha)^k}{k!} \left(\oint_\gamma \lambda 2G\dr^* B_-\right)^k W_\gamma(A)\\
    =\exp \left(\iii\alpha \oint_\gamma (A+\lambda 2 G \dr^*B_-)\right).
\end{multline}
\end{proof}

As a consistency check, we can check that the r.h.s. of (\ref{ab BF+BB: i_int(W)}) -- denote it $\til{W}_\gamma$ -- is indeed a BV cocycle in the abelian $BF+B^2$ theory:
\begin{multline}\label{ab BF+BB: tilW cocycle check}
    (Q-\iii\hbar \Delta)\til{W}_\gamma
    %\exp \left(\iii\alpha \oint_\gamma A+\lambda 2 G d^*B_-\right) 
    = 
    \til{W}_\gamma\cdot \iii\alpha\, Q\oint_\gamma \left(A+\lambda 2 G \dr^*B_-\right)\\=
    \til{W}_\gamma\cdot (-\iii\alpha) \oint_\gamma \left(\dr c+\lambda\tau -\lambda 2G \dr^* P_- \dr\tau\right).
\end{multline}
Note that we have
\begin{equation}
    2G\dr^* \underbrace{P_-}_{\frac12-\frac12*}\dr\tau=Gd^*\dr\tau = G(\Delta_g-\dr \dr^*)\tau= \tau-\underbrace{P_\mr{harm}\tau}_{=0\;\mr{by\;}B_1=0} -\dr \dr^*G\tau.
\end{equation}
%where in the last step we used the assumption of vanishing first Betti number.
Therefore, continuing the computation (\ref{ab BF+BB: tilW cocycle check}), we have
\begin{equation}
    (Q-\iii\hbar \Delta)\til{W}_\gamma= \til{W}_\gamma\cdot (-\iii\alpha)\oint_\gamma \dr(c+\lambda \dr^* G\tau) =0
\end{equation}
since $\gamma$ is closed.

\begin{rem}
    If we drop the assumption that the first Betti number of $M$ vanishes, the computation of the lifting (\ref{ab BF+BB: i_int(W)}) still goes through. However, neither $W_\gamma$ nor its lift $\til{W}_\gamma$ will be BV cocycles.
\end{rem}

\section{A topological quantum mechanics perspective}
\label{ss: TQM perspective}
Consider a topological quantum mechanics $\tau$ (cf. Remark \ref{rem: TQM}) defined by the space of states for a point $V=\fun$ equipped with the BV differential $Q_\hbar$ and a second degree $-1$ differential $G=\wh\kappa$. The Hamiltonian is
\begin{equation}\label{Hhat}
\begin{aligned}
    H=[Q_\hbar,\wh\kappa]&= [Q_0,\wh\kappa]-\iii\hbar [\Delta,\wh\kappa]+ [Q_\ii,\wh\kappa]\\
    &=\LL_E-\iii\hbar \eta + [Q_\ii,\wh\kappa]
\end{aligned}
\end{equation}
with $\LL_E$ the Lie derivative along the Euler vector field assigning degree $1$ to coordinates on  $\F''$ and degree $0$ to coordinates on $\F'$, $\eta
%=\frac12 \omega^{-1}(\frac{\dd}{\dd x},\kappa^\vee \frac{\dd}{\dd x})
$
as in (\ref{H second order derivation}), $\wh\kappa$ the extension of $\kappa^\vee$ to a derivation of $S\F^*$.

Let 
\begin{equation}\label{Z^tau}
    Z_{T,\dr T}=e^{-T H+\dr T G} = e^{-TH}+\dr T\, G e^{-TH}\in \Omega^\bt(\RR_+)\otimes \mr{End}(\wh{S}\F^*)
\end{equation}
be the partition function of the topological quantum mechanics. We denote its zero-form part along $T$ by $Z_T=e^{-TH}$.

\begin{thm}
\label{thm: interacting SDR via TQM}
    SDR data (\ref{SDR for interacting quantum theory}) can be expressed as follows in terms of the TQM partition function (\ref{Z^tau}):\footnote{
While working on the paper, we realized that independently the $\hbar\ra 0$ limit of the formulae of this theorem were obtained by Ezra Getzler, \cite{Getzler_talk}.
}
    \begin{eqnarray}
        i_\ii&=& \lim_{T\ra \infty} Z_T\circ i,\label{thm 5.3 i}\\
        p_\ii &=& \lim_{T\ra\infty} p\circ Z_T,\label{thm 5.3 p}\\
        K_\ii&=& \int_0^\infty Z_{T,\dr T}. \label{thm 5.3 K}
    \end{eqnarray}
    Also, one has
    \begin{equation}
        \lim_{T\ra\infty} Z_T = i_\ii\circ p_\ii \label{thm 5.3 Zinfty}
    \end{equation}
\end{thm}

\begin{proof}
Viewing the hamiltonian $H$ as a perturbation of $\mc{L}_E$ by $[Q_\ii-\iii\hbar\Delta,\wh\kappa]$, 
we have
\begin{multline}\label{thm 5.3 proof eq1}
Z_T=e^{-T (\mc{L}_E+[Q_\ii-\iii\hbar\Delta,\wh\kappa])}
\\
=\sum_{k\geq 0}(-1)^k\int_{t_0,\ldots,t_k>0, t_0+\cdots+t_k=T}\dr t_1\ldots \dr t_k
e^{-t_0 \LL_E} [Q_\ii-\iii\hbar\Delta,\wh\kappa] e^{-t_1 \LL_E} [Q_\ii-\iii\hbar\Delta,\wh\kappa] \cdots \\
\cdots [Q_\ii-\iii\hbar\Delta,\wh\kappa] e^{-t_k \LL_E}
\end{multline}
Note that
\begin{eqnarray}
    \lim_{t\ra \infty} e^{-t \LL_E}= i\circ p, \label{thm 5.3 proof eq2}\\
    e^{-t \LL_E} \circ i = i, \label{thm 5.3 proof eq3}\\
    p\circ e^{-t \LL_E} = p. \label{thm 5.3 proof eq4}
\end{eqnarray}
Precomposing (\ref{thm 5.3 proof eq1}) with $i$ and taking $T\ra\infty$ we have
\begin{multline}
    \lim_{T\ra \infty}Z_T \circ i=\\=
    \sum_{k\geq 0}(-1)^k \int_0^\infty\cdots \int_0^\infty \dr t_0\cdots \dr t_{k-1}
    e^{-t_0 \LL_E} [Q_\ii-\iii\hbar\Delta,\wh\kappa] \cdots e^{-t_{k-1}\LL_E} [Q_\ii-\iii\hbar\Delta,\wh\kappa] i\\
    =\sum_{k\geq 0}(-1)^k \int_0^\infty\cdots \int_0^\infty \dr t_0\cdots \dr t_{k-1}
    e^{-t_0 \LL_E} \wh\kappa(Q_\ii-\iii\hbar\Delta) \cdots e^{-t_{k-1}\LL_E} \wh\kappa(Q_\ii-\iii\hbar\Delta) i\\
    =\sum_{k\geq 0} (-1)^k(\underbrace{\LL_E^{-1}\wh\kappa}_K (Q_\ii-\iii\hbar\Delta))^k i =
    i_\ii.
\end{multline}
This recovers the homological perturbation formula for $i_\ii$ and proves (\ref{thm 5.3 i}). Here we used the formula (\ref{Ksym L_E formula}) for $K$; $\LL_E^{-1}\wh\kappa$ is a shorthand for $(\LL_E+c\, ip)^{-1}\wh\kappa$ with arbitrary $c>0$.

Likewise, postcomposing (\ref{thm 5.3 proof eq1}) with $p$ and taking $T\ra\infty$, we get
\begin{multline}
    \lim_{T\ra\infty} p\circ Z_T =\\
    =\sum_{k\geq 0}(-1)^k \int_0^\infty \cdots \int_0^\infty \dr t_1\cdots \dr t_k
    p[Q_\ii-\iii\hbar\Delta,\wh\kappa] e^{-t_1 \LL_E}
    \cdots [Q_\ii-\iii\hbar\Delta,\wh\kappa] e^{-t_k\LL_E} \\
    =\sum_{k\geq 0}(-1)^k \int_0^\infty \cdots \int_0^\infty \dr t_1\cdots \dr t_k
    p(Q_\ii-\iii\hbar\Delta)\wh\kappa e^{-t_1 \LL_E}
    \cdots (Q_\ii-\iii\hbar\Delta)\wh\kappa e^{-t_k\LL_E}\\
    =\sum_{k\geq 0}(-1)^k p((Q_\ii-\iii\hbar\Delta)\underbrace{\wh\kappa \LL_E^{-1}}_K)^k = p_\ii.
\end{multline}
This recovers the homological perturbation formula for $p_\ii$, proving (\ref{thm 5.3 p}).

Note that in the limit $T\ra\infty$, the integral in (\ref{thm 5.3 proof eq1}) precomposed with $i$ is supported in the region $t_k\ra\infty$, $t_0,\ldots,t_{k-1}$ finite. If instead we postcompose with $p$, the integral is supported at $t_0\ra\infty$, $t_1,\ldots,t_k$ finite.

Next, composing (\ref{thm 5.3 proof eq1}) with $\wh\kappa$ and integrating over $T$, we have
\begin{multline}
    \int_0^\infty Z_{T,\dr T}= \int_0^\infty \dr T\, \wh\kappa Z_T\\
    =\sum_{k\geq 0} (-1)^k\int_{t_0,\ldots,t_k>0} \dr t_0\cdots \dr t_k \wh\kappa e^{-t_0\LL_E} [Q_\ii-\iii\hbar\Delta,\wh\kappa] e^{-t_1 \LL_E} \cdots 
    [Q_\ii-\iii\hbar\Delta,\wh\kappa] e^{-t_k \LL_E}\\
    =\sum_{k\geq 0} (-1)^k\int_{t_0,\ldots,t_k>0} \dr t_0\cdots \dr t_k \wh\kappa e^{-t_0\LL_E} (Q_\ii-\iii\hbar\Delta)\wh\kappa e^{-t_1 \LL_E} \cdots 
    (Q_\ii-\iii\hbar\Delta)\wh\kappa e^{-t_k \LL_E}\\
    =\sum_{k\geq 0}(-1)^k\underbrace{\wh\kappa \LL_E^{-1}}_K  ((Q_\ii-\iii\hbar\Delta)\underbrace{\wh\kappa \LL_E^{-1}}_K)^k= K_\ii.
\end{multline}
This recovers the homological perturbation formula for $K_\ii$, proving (\ref{thm 5.3 K}).

To prove (\ref{thm 5.3 Zinfty}), note that in the limit $T\ra\infty$, in the integral (\ref{thm 5.3 proof eq1}) one $t_i$ must be very large. The corresponding factor in the integrand  $e^{-t_i \LL_E}\ra i\circ p$ splits the the integrand into a composition of the integrands for $i_\ii$ and for $p_\ii$, see Figure \ref{fig:Z at T to infinity}.
\end{proof}
\begin{figure}
    \centering
    \includegraphics[width=0.8\linewidth]{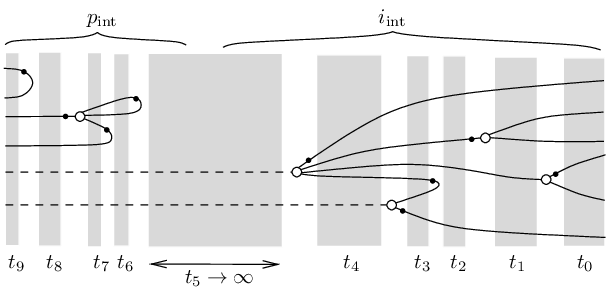}
    \caption{Graphical representation for a term in $Z_{T\ra\infty}$. Grey boxes are $e^{-t_i \LL_E}$, with $t_i$ the width of the box. Between the boxes, terms of $[\wh\kappa,Q_\ii]$ or $\eta$ are applied.}
    \label{fig:Z at T to infinity}
\end{figure}

%[PICTURE - DIAGRAM FOR $Z_{T\ra\infty}$?]

\begin{defn}
    Let us call a linear map $\bar\kappa\colon \colon \F^\bt\ra \F^{\bt-1}$ a ``non-normalized chain homotopy'' (or ``non-normalized gauge-fixing operator'') compatible with the SDR (\ref{SDR F to F'}) if
\begin{itemize}
    \item $\bar\kappa$ induces an isomorphism $\mr{im}(d)\ra \mr{im}(\kappa)$,
    \item $\bar\kappa$ vanishes on $\F'\oplus \mr{im}(\kappa)$,
    \item $[\bar\kappa,d]$ is a non-negative operator (with kernel $\F'$, as follows from the points above).
\end{itemize}
\end{defn}

\begin{rem}[Generalization to non-normalized gauge-fixing operators]
\label{rem: non-normalized chain homotopy}
 In the construction of this subsection, instead of choosing $G=\wh\kappa$ for TQM, one can take $G=\wh{\bar\kappa}$ -- an extension to a derivation of $\wh{S}\F^*$ of the dual of a non-normalized chain homotopy $\bar\kappa\colon \F^\bt\ra \F^{\bt-1}$ compatible with the SDR (\ref{SDR F to F'}). Then for the topological quantum mechanics (\ref{Z^tau}) with $G=\wh{\bar\kappa}$ and the Hamiltonian $H=[Q_\hbar,\wh{\bar\kappa}]$, the statement of Theorem \ref{thm: interacting SDR via TQM} still holds. 
 
The reason is that one can arrange the formulae in the proof of Theorem \ref{thm: interacting SDR via TQM} as involving $\wh\kappa$ and only in the combination $[Q_0,\wh\kappa]^{-1}\wh\kappa $ and one has
$[Q_0,\wh\kappa]^{-1}\wh\kappa =[Q_0,\wh{\bar\kappa}]^{-1}\wh{\bar\kappa}$.
\end{rem}

% \subsection{Classical limit of $i_\ii,p_\ii$}
% {\color{gray} $i_\ii,p_\ii\bmod\hbar\ra $  $L_\infty$-morphisms}

% Let us describe the classical limit (reduction modulo $\hbar$) of the SDR (\ref{SDR for interacting quantum theory}) -- let us denote this reduction by
% \begin{equation}
%     \SDR{(\wh{S}\F^*,Q)}{(\wh{S}\F'^*,Q')}{i_\ii^\cl}{p_\ii^\cl}{K_\ii^\cl}
% \end{equation}
% with $Q=\{S,-\}=Q_0+Q_\ii$,  $Q'=\{S',-\}'$.

% Consider the reduction modulo $\hbar$ of the TQM Hamiltonian (\ref{Hhat}),
% \begin{equation}
%     H^\cl=[\wh\kappa, Q]=
% \end{equation}

\begin{expl}
    In the example of toy interacting scalar field (Section \ref{sss example: toy interacting quantum scalar}) with action $S(x)=\frac{x^2}{2}+P(x)$, we have
    \begin{equation}
        \wh{H}=\Big[\underbrace{\xi \frac{\dd}{\dd x}}_{\wh\kappa}, \underbrace{S'(x)\frac{\dd}{\dd \xi}-\iii\hbar \frac{\dd^2}{\dd x\dd \xi}}_{Q_\hbar}\Big]=
        S'(x)\frac{\dd}{\dd x}+S''(x) \xi \frac{\dd}{\dd \xi}-\iii\hbar \frac{\dd^2}{\dd x^2},
    \end{equation}
    and the TQM formula (\ref{thm 5.3 K}) for the chain homotopy becomes
    % \marginpar{I don't know how useful this formula is...}
    \begin{equation}
        K_\ii(O)=\xi \int_0^\infty \dr T \frac{\dd}{\dd x} e^{-T\left(S'(x)\frac{\dd}{\dd x}-\iii\hbar \frac{\dd^2}{\dd x^2}\right)} O(x).
    \end{equation}
\end{expl}

\subsection{Proof of Theorem \ref{thm: Q' and p_int via BV pushforward} via TQM}\label{ss: TQM proof of thm on Q'_hbar and p_int}
We give a second proof of Theorem \ref{thm: Q' and p_int via BV pushforward}, relying on the topological quantum mechanics picture above.

\textbf{(i)$\Leftarrow$(ii).} First, remark that if (ii) of Theorem \ref{thm: Q' and p_int via BV pushforward} is known, i.e., $p_\ii$ is the BV pushforward $P_*$ of observables (\ref{p_int}), then (i) of the Theorem follows automatically, by the following argument. We know by Remark \ref{rem: SDR for interacting theory from conjugation} that $P_*=p_\ii^\mr{conj}$ intertwines the BV differential $Q_\hbar$ on the full BV complex with the BV differential $Q'_\hbar$ on the infrared BV complex. On the other hand, by homological perturbation lemma, $p_\ii$ intertwines $Q_\hbar$ with some induced differential ${Q'_\hbar}^{\mr{HPL}}$. But, by surjectivity of $p_\ii$, it can only intertwine $Q_\hbar$ with a unique differential on $\wh{S}\F'^*[[\hbar]]$. Hence, $Q'_\hbar={Q'_\hbar}^{\mr{HPL}}$.

Now we proceed to prove (ii) of Theorem \ref{thm: Q' and p_int via BV pushforward}. We will first consider the case $\F'=0$ (we will generalize afterwards).
%and then generalize to $\F'$ possibly nonzero.

\textbf{Proof of (ii) of Theorem \ref{thm: Q' and p_int via BV pushforward} 
in the case $\F'=0$.}
From Theorem \ref{thm: interacting SDR via TQM} we know that, for $O\in \fun$, 
\begin{equation}\label{p_int O as projection to ker whH}
p_\ii O=\lim_{T\ra\infty} p e^{-TH}O=P_{\ker H}O
\end{equation}
-- the projection onto the kernel of $H$, given by constants (for $\F'=0$), along $\mr{im}H$.

To anticipate the next step: 
%Idea of the argument for $\F'=0$:
if we manage to find a pairing on $\fun$ with respect to which $\mr{ker}H$ and $\mr{im}H$ are orthogonal (and $1$ has norm $1$), we can construct the desired projection onto $\ker H$ as pairing with $1$.

Consider the following pairing on observables
\begin{equation}\label{pairing <<,>>}
    \ll,\gg \colon \wh{S}\F^*\otimes \wh{S}\F^* \ra \CC[[\hbar]],\quad
    \ll O_1,O_2\gg = c \int_{\LL} e^{\frac{\iii}{\hbar}S}O_1 O_2
\end{equation}
with $c=(\int_\LL e^{\frac{\iii}{\hbar}S})^{-1}$ the normalization factor, so that $\ll 1,1\gg=1$.

The key property of pairing (\ref{pairing <<,>>}) is that $H$ is self-adjoint with respect to it:
\begin{equation}\label{self-adjointness of whH}
    \ll HO_1,O_2 \gg = \ll O_1,H O_2 \gg.
\end{equation}
To prove this, note that
\begin{equation}
    H=[\wh\kappa,Q_\hbar]= \left[\left\langle \kappa(x),\frac{\dd}{\dd x} \right\rangle, \omega^{-1}\left(\frac{\dd S}{\dd x},\frac{\dd}{\dd x}\right)-\frac{\iii\hbar}{2}\omega^{-1}\left(\frac{\dd}{\dd x},\frac{\dd}{\dd x}\right)\right],
\end{equation}
hence the action of $H$ restricted to $\LL$ is
\begin{equation}
    (HO)|_\LL= h^{ij}\left(\frac{\dd S(x_\kappa)}{\dd x^i_\kappa} \frac{\dd}{\dd x^j_\kappa} - \iii\hbar  \frac{\dd}{\dd x^i_\kappa}\frac{\dd}{\dd x^j_\kappa}\right) O(x_\kappa).
\end{equation}
Here $h=\omega^{-1}\kappa^\vee$ as in (\ref{h}) (recall that $h$ is symmetric); %$x_\kappa$ denotes the projection of $x\in \F$ onto the second summand in $\F=\mr{im}(d)\oplus \mr{im}(\kappa)$; 
$x_\kappa^i$ are coordinates on $\LL=\mr{im}(\kappa)\subset \F$. Then, we have
\begin{multline}\label{self-adjointness of whH eq1}
    \ll HO_1,O_2\gg =\\
    = c\int_\LL Dx_\kappa e^{\frac{\iii}{\hbar}S%(x_\kappa)
    } h^{ij}\left(
    \frac{\dd S%(x_\kappa)
    }{\dd x^i_\kappa} \frac{\dd }{\dd x^j_\kappa}O_1%(x_\kappa)
    -\iii\hbar \frac{\dd}{\dd x_\kappa^i}\frac{\dd}{\dd x_\kappa^j} O_1
    %(x_\kappa)
    \right) O_2%(x_\kappa)
    \\
    =c\int_\LL Dx_\kappa e^{\frac{\iii}{\hbar}S}O_1 h^{ij}\Big(
    -\Big(\underbrace{\frac{\iii}{\hbar}\frac{\dd S}{\dd x^i_\kappa}\frac{\dd S}{\dd x^j_\kappa}}_a+\underbrace{\frac{\dd^2 S}{\dd x^i_\kappa \dd x^j_\kappa}}_b+\underbrace{\frac{\dd S}{\dd x^i_\kappa}\frac{\dd}{\dd x^j_\kappa}}_c\Big)O_2-\\
    -\iii\hbar\Big(
    \underbrace{2\frac{\iii}{\hbar}\frac{\dd S}{\dd x^i_\kappa}\frac{\dd}{\dd x^j_\kappa}}_d+\underbrace{\left(\frac{\iii}{\hbar}\right)^2 \frac{\dd S}{\dd x^i_\kappa}\frac{\dd S}{\dd x^j_\kappa}}_e+\underbrace{\frac{\iii}{\hbar}\frac{\dd^2 S}{\dd x^i_\kappa \dd x^j_\kappa}}_f+\frac{\dd^2}{\dd x^i_\kappa \dd x^j_\kappa}
    \Big) O_2
    \Big).
\end{multline}
Here transition to the third line is integration by parts (once for the first term, twice for the second), transferring the derivatives from $O_1$ to $O_2$. Note that terms $a$ and $e$ cancel, $b$ and $f$ cancel, and $c$ cancels half of $d$, leaving
\begin{equation}\label{self-adjointness of whH eq2}
    \ll HO_1,O_2 \gg = c\int_\LL Dx_\kappa e^{\frac{\iii}{\hbar}S} O_1 h^{ij}\left(\frac{\dd S}{\dd x^i_\kappa} \frac{\dd}{\dd x^j_\kappa}-\iii\hbar \frac{\dd^2}{\dd x^i_\kappa \dd x^j_\kappa}\right)O_2
    = \ll O_1,HO_2 \gg,
\end{equation}
which proves (\ref{self-adjointness of whH}).

Next, given an observable $O\in \fun$, we split it as 
\begin{equation}\label{O=O_kerH+O_imH}
O=O_{\ker H}+O_{\mr{im}H},
\end{equation}
according to the splitting $\fun=\ker H\oplus \mr{im}H$. Here $O_{\ker H}=p_\ii O$ and is a constant, cf. (\ref{p_int O as projection to ker whH}).
We have
\begin{multline}
    \ll O,1 \gg= \ll O_{\ker H},1\gg + \ll H(\cdots),1 \gg=
    p_\ii(O)+ \ll H(\cdots),1 \gg \\
    \underset{(\ref{self-adjointness of whH}}{=}p_\ii(O)+ \ll \cdots, \underbrace{H(1)}_0\gg
    =p_\ii(O).
\end{multline}
The left hand side is the BV integral $c\int_\LL e^{\frac{\iii}{\hbar}S}O$, which proves (ii) of Theorem \ref{thm: Q' and p_int via BV pushforward} in the case $\F'=0$.

\textbf{Generalization of the argument to $\F'$ possibly nonzero.}
We consider the pairing
\begin{equation}
    \ll,\gg\colon \fun \otimes \fun\ra \funprime,\quad
    \ll O_1,O_2\gg=e^{-\frac{\iii}{\hbar}S'}\int_\LL e^{\frac{\iii}{\hbar}S}O_1 O_2.
\end{equation}
With respect to it, self-adjoint of $H$ (\ref{self-adjointness of whH}) still holds (and the proof (\ref{self-adjointness of whH eq1}), (\ref{self-adjointness of whH eq2}) goes through without change, except that now $S,O_1,O_2$ additionally depend on $x'\in \F'$ and the normalization is $c= e^{-\frac{\iii}{\hbar}S'}$).

Consider an observable $O\in\fun$ and split it as in (\ref{O=O_kerH+O_imH}). We have
\begin{multline}\label{computation of <<O,1>>}
    %e^{-\frac{\iii}{\hbar}S'}\int_\LL e^{\frac{\iii}{\hbar} S}O
    \ll O,1 \gg = \ll O_{\ker H}+ O_{\mr{im} H},1\gg =
    \ll O_{\ker H},1 \gg + \underbrace{\ll H(\cdots), 1\gg}_{\underset{(\ref{self-adjointness of whH})}{=}\ll \cdots, H(1) \gg=0}\\
    = \ll O_{\ker H},1 \gg
\end{multline}
Note that, by Theorem \ref{thm: interacting SDR via TQM}, we have
\begin{equation}
    O_{\ker H}=\lim_{T\ra\infty}Z_T O=i_\ii p_\ii O= \lim_{T\ra \infty} Z_T i p_\ii O.
\end{equation}
Using this, we continue (\ref{computation of <<O,1>>}):
\begin{multline}\label{computation of <<O,1>> eq2}
    \ll O,1\gg = \lim_{T\ra\infty} \ll e^{-TH}i p_\ii O,1 \gg \underset{(\ref{self-adjointness of whH})}{=} 
    \lim_{T\ra\infty} \ll i p_\ii O, \underbrace{e^{-TH}1}_1\gg\\=
    \ll i p_\ii O,1 \gg = p_\ii O.
\end{multline}
Here in the last step we use the obvious property $\ll i(O'),1 \gg=O'$. Since the l.h.s. of (\ref{computation of <<O,1>> eq2}) is the BV pushforward of $O$, we have a proof of (ii) of Theorem \ref{thm: Q' and p_int via BV pushforward}.

Together with the argument in the beginning of this subsection ((ii) implies (i)), this concludes the second proof of Theorem \ref{thm: Q' and p_int via BV pushforward}.

\section{Classical limit of $i_\ii,p_\ii$, $K_\ii$}
%{\color{gray} $i_\ii,p_\ii\bmod\hbar\ra $  $L_\infty$-morphisms}
\label{sec: classical limit of i_int, p_int, K_int}

Let us describe the classical limit (reduction modulo $\hbar$) of the SDR (\ref{SDR for interacting quantum theory}) -- let us denote this reduction by
\begin{equation}\label{SDR for classical interacting theory}
    \SDR{(\wh{S}\F^*,Q^\cl)}{(\wh{S}\F'^*,Q'^\cl)}{i_\ii^\cl}{p_\ii^\cl}{K_\ii^\cl}
\end{equation}
with $Q^\cl=\{S^\cl,-\}=Q_0+Q^\cl_\ii$,  $Q'^\cl=\{S'^\cl,-\}'$. The superscript $\cl$ means ``modulo $\hbar$,'' i.e., $S^\cl=S\bmod\hbar$ etc.
%\marginpar{Edit: we are taking $S,S'$ mod $\hbar$ here. (Can assume that $S$ does not depend on $\hbar$, but $S'$ still might..)}
The maps $i^\cl_\ii, p^\cl_\ii,K^\cl_\ii$ are given by the homological perturbation formulae (\ref{HPL tilde ipK}) with $i,p,K$ of the free classical theory (\ref{SDR of classical free BV theory}) and with  $\delta=Q_\ii^\cl$ the perturbation of the differential.

Consider the reduction modulo $\hbar$ of the TQM Hamiltonian (\ref{Hhat}),
\begin{multline}\label{Hhat classical}
    H^\cl=[\wh\kappa, Q^\cl]= \left[ \left\langle\kappa(x),\frac{\dd}{\dd x} \right\rangle, \omega^{-1}\left(\frac{\dd S^\cl}{\dd x},\frac{\dd}{\dd x}\right)\right]
    \\
    = \underbrace{\omega^{-1}\left(\frac{\dd S^\cl}{\dd x},\kappa^\vee \frac{\dd }{\dd x}\right)}_{H^\cl_1} + \underbrace{\omega^{-1}\left(\left\langle\kappa(x),\frac{\dd}{\dd x} \right\rangle \frac{\dd S^\cl}{\dd x},\frac{\dd}{\dd x} \right)}_{H^\cl_2}
    \\
    = h^{ij}\frac{\dd S^\cl}{\dd x^i_\kappa}\frac{\dd}{\dd x^j_\kappa}+\kappa^i_{\; j} x^j_d (\omega^{-1})^{IJ}\frac{\dd^2 S^\cl}{\dd x^i_\kappa\dd x^I}\frac{\dd}{\dd x^J}.
\end{multline}
Here $x^i_d$ are coordinates on $\mr{im}(d'')$ and $x^I$ are coordinates on the entire $\F=\F'\oplus \mr{im}(\kappa)\oplus \mr{im} (d'')$.

Note that $H^\cl$ is a derivation of $\wh{S}\F^*$ or, equivalently, a formal vector field on $\F$ (or: a vector field on a formal neighborhood of zero in $\F$; we denote the latter $\F^\formal$).\footnote{
In fact, if $S$ has positive convergence radius in $\F$, one can replace formal neighborhoods of zero in this subsection with finite open neighborhoods of zero.
}

Let $\Cr\colon \F'^\formal\ra \LL^\formal$ be the formal nonlinear map sending $x'$ to the solution $x_\kappa\in \LL$ of the equation
\begin{equation}\label{conditional crit point eq}
    \frac{\dd S^\cl(x'+x_\kappa)}{\dd x_\kappa}=0.
\end{equation}
Note that in the formal setting such a solution exists and is unique (i.e., for $x'$ infinitesimal there exists a unique infinitesimal $x_\kappa$). Also note that (\ref{conditional crit point eq}) means that $x'+x_\kappa$ is a critical point of $S^\cl$ restricted to the affine space $x'+\LL$. Denote
\begin{equation}
    \mc{Z}\colon=\mr{graph}(\Cr)=\{x'+x_\kappa(x')\;|\; (\ref{conditional crit point eq})\;\mr{holds}\}\subset (\F'\oplus\LL)^\formal.
\end{equation}
By the remark above, the projection $\pi$ onto $\F'$ restricted to $\mc{Z}$ yields a (nonlinear) isomorphism 
\begin{equation}\label{pi|_Z}
    \pi|_{\mc{Z}}\colon \mc{Z} \xra{\sim} \F'^\formal.
\end{equation}
%$\mc{Z}$ and $\F'^\formal$.

Properties of the vector field $H^\cl$:
\begin{enumerate}[(i)]
    \item Locus of zeros (fixed points) of $H^\cl$ is $\mc{Z}$. This zero locus is attracting for the flow of $-H^\cl$.\footnote{Thus, $-H^\cl$ behaves like the gradient vector field of a  Morse-Bott function with critical locus $\mc{Z}$ of index zero.}
    \item $H^\cl$ is tangent to $x'+\LL$ for any $x'\in \F'^\formal$. In the restriction of $H^\cl$ to $x'+\LL$, only the term $H^\cl_1$ in (\ref{Hhat classical}) survives, whereas $H^\cl_2$ vanishes.
\end{enumerate}

%\marginpar{\bl May 16: moved the defn of $\Phi_T$ here}
Denote 
\begin{equation}
    \Phi_T=\mr{Flow}_T(-H^\cl)\colon \F^\formal \ra \F^\formal
\end{equation}
the flow of the vector field $-H^\cl$ in time $T\in \RR$.

\begin{expl}
    In free theory, $S=S_0=\frac12\omega(x,dx)$, we have $H^\cl=\LL_E=\langle x'',\frac{\dd}{x''} \rangle$ the Euler vector field. Its zero locus is $\mc{Z}=\F'$; it is repulsive for $\LL_E$ and attractive for $-\LL_E$. The flow $\Phi_T$ maps $x=x'+x''$ to $x'+e^{-T}x''$.
\end{expl}
\begin{expl}\label{example: Chern-Simons H^cl in Lorenz gauge}
    Consider Chern--Simons theory on a 3-manifold $M$ with structure quadratic Lie algebra $\g$, with action 
    \begin{equation}
    S=\int_M \frac12 \langle \mc{A},\dr \mc{A} \rangle+\frac16 \langle \mc{A},[\mc{A},\mc{A}] \rangle,
    \end{equation}
    with $\mc{A}\in \Omega^\bt(M,\g)[1]=\F$ the AKSZ superfield. Equipping $M$ with a metric $g$ we have a Hodge decomposition
    $\Omega^\bt=\mr{Harm}\oplus \mr{im}(\dr)\oplus \mr{im}(\dr^*)$ into harmonic, exact and coexact forms. It induces a splitting $\F=\F'\oplus \F''$ with $\F'=\mr{Harm}\otimes \g[1]\simeq H^\bt(M,\g)[1]$ and $\F''$ given by exact and coexact forms. For $\LL$ we choose the coexact forms. In this setting, if we use non-normalized chain homotopy $\bar\kappa=\dr^*$ (cf. Remark \ref{rem: non-normalized chain homotopy}), the vector field (\ref{Hhat classical}) becomes 
    \begin{equation}
        H^\cl=\langle \underbrace{\Delta \mc{A}}_{\mr{free\;part}}+\underbrace{\dr^*\frac12[\mc{A},\mc{A}]+[\mc{A},\dr^*\mc{A}]}_\mr{perturbation},\frac{\delta}{\delta\mc{A}} \rangle
        =
       \underbrace{\langle \dr^* F_\mc{A},\frac{\delta}{\delta\mc{A}} \rangle}_{H^\cl_1}+
        \underbrace{\langle \dr_\mc{A} \dr^*\mc{A},\frac{\delta}{\delta \mc{A}} \rangle}_{H^\cl_2},
    \end{equation}
    with $F_\mc{A}=\dr\mc{A}+\frac12[\mc{A},\mc{A}]$ and $\dr_\mc{A}=\dr+[\mc{A},-]$. 
    The zero-locus of $H^\cl$ restricted to ghost number zero is
    \begin{equation}
        \mc{Z}|_{\mr{gh}=0}=\{A\in \Omega^1(M,\g)\;|\; \dr^*F_A=0,\; \dr^*A=0\}.
    \end{equation}
    %-- connections satisfying the Yang-Mills equation and the Lorenz gauge condition. \marginpar{Appearance of Yang-Mills equation here is kind of surprising..}
    %\marginpar{\bl May 11 edit}
    Note that in the abelian case $\g=\mathbb{R}$, this is the space of connections satisfying the abelian Yang--Mills %Maxwell's 
    equation $\dr^* \dr A=0$ and the Lorenz gauge condition.

    %\marginpar{\bl May16: nonlinear heat eqn added}
    The flow $\Phi_T$ maps a connection 1-form $A$ to a time-dependent connection $A_T$ -- a solution of a nonlinear heat equation
    \begin{equation}\label{nonlinear heat eq on A_T}
        \frac{\partial}{\partial T} A_T = - \dr^* F_{A_T}-\dr_{A_T} \dr^* {A_T}
    \end{equation}
    with initial condition $A_{T=0}=A$. In the abelian case $\g=\RR$, equation (\ref{nonlinear heat eq on A_T}) becomes the linear heat equation on a 1-form, $\partial_T A_T=-\Delta A_T$.
\end{expl}

% Denote 
% \begin{equation}
%     \Phi_T=\mr{Flow}_T(-H^\cl)\colon \F^\formal \ra \F^\formal
% \end{equation}
% the flow of the vector field $-H^\cl$ in time $T\in \RR$.

The following is an immediate consequence of Theorem \ref{thm: interacting SDR via TQM} considered modulo $\hbar$.
\begin{cor}\!\!\!\footnote{A version of this statement appeared, independently of our work, in \cite{Getzler_talk}. Also, the statement that $i_\ii^\cl$, defined by the homological perturbation series $\sum_{n\geq 0} (- K Q^\cl_\ii)^n i$, is a morphism of commutative dg algebras (and so defines an $L_\infty$ morphism (\ref{pi^cl_int L_infty morphism}) was proven in \cite{Berglund}.}
\label{cor: class limit of interacting ipK via flow of H^cl}
    In the SDR (\ref{SDR for classical interacting theory}):
    \begin{enumerate}[(a)]
        \item $p_\ii^\cl\colon (\wh{S}\F^*,Q^\cl,\cdot)\ra (\wh{S} \F'^*,Q'^\cl,\cdot)$ is a morphism of dg commutative algebras. It is the pullback by the nonlinear formal map
        \begin{equation}\label{iota_int^cl}
            \iota_\ii^\cl\colon \F'^\formal\ra \F^\formal,\quad x'\mapsto \lim_{T\ra \infty} %\mr{Flow}_T(-H^\cl)
            \Phi_T(x')=(\pi|_\mc{Z})^{-1}x'=x'+\Cr(x'),
        \end{equation}
        with $\pi|_\mc{Z}$ as in (\ref{pi|_Z}). Thus, $\iota_\ii^\mr{cl}$ maps $x'$ to the critical point of $S^\cl$ restricted to $x'+\LL$.\footnote{
        By properties of $H^\cl$, $\Phi_T$ moves $x'$ along $x'+\LL$, converging as $T\ra\infty$ to the fiberwise critical point of $S$, $x'+\Cr(x')$. %{\color{blue} REDUNDANT FOOTNOTE?}
        }
        \item \label{cor: class limit of interacting ipK via flow of H^cl (b)}
        $i_\ii^\cl \colon (\wh{S}\F'^*,Q'^\cl,\cdot)\ra (\wh{S} \F^*,Q^\cl,\cdot)$ is a morphism of dg commutative algebras. It is the pullback of the nonlinear formal map
        \begin{equation}\label{pi_int^cl}
            \pi_\ii^\cl\colon \F^\formal\ra \F'^\formal,\quad x\ra \lim_{T\ra\infty} \pi \circ %\mr{Flow}_T(-H^\cl)
            \Phi_T (x).
        \end{equation}
        \item The chain homotopy in (\ref{SDR for classical interacting theory}) is
        \begin{equation}
            K_\ii^\cl(O)=\int_0^\infty \dr T\, \wh\kappa\, \Phi_T^* O
        \end{equation}
        for $O\in \wh{S}\F^*$, or, equivalently,
        \begin{equation}
            K_\ii^\cl(O)(x)=\int_0^\infty O\Big((\mr{id}+\dr T\, \kappa)\circ\Phi_T(x)\Big)
        \end{equation}
        for $x\in \F^\formal$.
    \end{enumerate}
\end{cor}

\begin{rem}[Language of $L_\infty$ algebras and morphisms]
\label{rem: L_infty language}
    The differential $Q^\cl$ equips $\F[-1]$ with the structure of an $L_\infty$ algebra (and $\omega$ additionally makes it a cyclic $L_\infty$ algebra with pairing of degree $-3$).\footnote{
    One reads off the multilinear operations $l_n\colon \wedge^n \F[-1]\ra \F[-1]$ of degrees $2-n$ from the Taylor expansion of $Q^\cl$: $l_1=d$ the differential and, for $n\geq 2$, $l_n$ is the dual of $Q^\cl_{n+1}$ in the notations of (\ref{Q_int = sum Q_n}).
    } Likewise, $Q'^\cl$ equips $\F'[-1]$ with the structure of an $L_\infty$ algebra (again, cyclic of degree $-3$ by virtue of $\omega'$) -- the homotopy transfer of the $L_\infty$ algebra $(\F[-1],Q^\cl)$ onto $\F'[-1]$, using SDR (\ref{SDR F to F'}). In this language, maps (\ref{iota_int^cl}), (\ref{pi_int^cl}) are $L_\infty$ morphisms,
    \begin{gather}
        \iota_\ii^\cl\colon (\F'[-1],Q'^\cl)\xra{L_\infty\;\mr{mor}} (\F[-1],Q^\cl), \label{iota^cl_int L_infty morphism}\\
        \pi_\ii^\cl \colon (\F[-1],Q^\cl)\xra{L_\infty\;\mr{mor}} (\F'[-1],Q'^\cl). \label{pi^cl_int L_infty morphism}
        \end{gather}
    They provide nonlinear ($L_\infty$) deformations of the original linear chain maps $\iota$, $\pi$ from (\ref{SDR F to F'}).
\end{rem}

% {\color{gray} 
% %To add:  
% %example of 1d BF and $\pi_\ii^\cl$ given by $\log\mr{Hol}$. Also: 
% Are $\pi_\ii^\cl$ and $\iota_\ii^\cl$ still mutually transpose in some sense? Not by the usual $\omega$, I think, but maybe by some up-to-homotopy cyclic structure.. (Getzler was talking about something in that direction in \cite{Getzler_talk}.)}

%\marginpar{added Apr 26}
\begin{expl}
\label{example: H^cl for ab BF+B^2 to ab YM}
    Let us revisit the example of Section \ref{ss: example - lifting ab Wilson loop} -- BV pushforward from abelian $BF+B^2$ theory to abelian Yang--Mills. The vector field $H^\cl$ restricted to fields of ghost degree zero is:
    \begin{equation}
        H^\cl_{\mr{gh}=0}=\int_M (\mr{id}-P_\mr{harm})B_- \frac{\delta}{\delta B_-}+2G P_+ \dr \dr^* B_- \frac{\delta}{\delta B_+} + \lambda 2 G\dr^*B_- \frac{\delta}{\delta A}.
    \end{equation}
    As it is a linear vector field on the space of triples $(A,B_+,B_-)\in \Omega^1\oplus \Omega^2_+\oplus \Omega^2_-$, its flow in time $T$ is easy to compute:
    \begin{multline}
        \Phi_T(A,B_+,B_-)=(A+(1-e^{-T})\lambda 2G \dr^* B_-\;,\\ B_+ +(1-e^{-T}) 2G P_+ \dr \dr^* B_-\;,\;
        e^{-T}(\mr{id}-P_\mr{harm})B_-+P_\mr{harm}B_-).
    \end{multline}
    Therefore, we have
    \begin{multline}\label{ab BF+BB: pi_int}
        \pi_\ii =\lim_{T\ra\infty} \pi\circ\Phi_T \colon (A,B_+,B_-)\mapsto \\
        \mapsto (A+\lambda 2G \dr^* B_-\;,\; B_++2G P_+ \dr \dr^* B_- \;,\; P_\mr{harm}B_-).
    \end{multline}
    Hence, the observable-lifting map, in the classical limit and restricted to $\mr{gh}=0$ fields, reads
    \begin{multline}
        i_\ii^\cl\colon O'(A,B_+,B_{-\mr{harm}}) \mapsto\\
        \mapsto O(A,B_+,B_-)= O'(A+\lambda 2G \dr^* B_-\;,\;B_++2G P_+ \dr \dr^* B_-\;,\;P_\mr{harm}B_-).
    \end{multline}
    Applying this map to the Wilson loop, we recover the result of Proposition \ref{prop: lifting ab Wilson loop} (there are no quantum corrections to the lifting map $i_\ii^\cl$ in this case).

    We remark that instead of using the chain homotopy (\ref{ab BF+BB: kappa}) one can use an equivalent non-normalized chain homotopy $\bar\kappa$:
    \begin{equation}
    %\F_2\colon\qquad 
    \bar\kappa\colon\quad
    \begin{array}{ccccccccccc} 
       \Omega^0 &\xla{\dr^*} &\Omega^1 &\xla{\dr^* }& \Omega^2_- &\xla{0}& \Omega^2_-& \xla{P_-\dr^* } &\Omega^3 &\xla{\dr^*} &\Omega^4.
    \end{array}
\end{equation}
The corresponding vector field (restricted to $\mr{gh}=0$ fields) is
\begin{equation}
    H^\cl_{\mr{gh}=0}= [\wh{\bar\kappa},Q]\Big|_{\mr{gh}=0} =\int_M \frac12 \Delta B_- \frac{\delta}{\delta B_-}+P_+ \dr \dr^* B_-\frac{\delta}{\delta B_+}+\lambda \dr^* B_-\frac{\delta}{\delta A}
\end{equation}
with $\Delta$ the Laplace-de Rham operator.
Its flow is
\begin{multline}
    \Phi_T(A,B_+,B_-)=\left( A+\lambda \dr^* \frac{\mr{id}-e^{-\frac12 T\Delta}}{\frac12 \Delta}B_-\;,\right.\\ \left.
    B_++P_+ \dr \dr^* \frac{\mr{id}-e^{-\frac12 T\Delta}}{\frac12 \Delta}B_-\;,\; e^{-\frac12 T\Delta}B_-
    \right).
\end{multline}
Here $e^{-\frac12 T\Delta}$ is the heat flow operator.
Taking the limit $T\ra\infty$ and projecting to $\F'$, we obtain the same formula for $\pi_\ii$ as before (\ref{ab BF+BB: pi_int}), using the normaized homotopy $\kappa$.
\end{expl}

\subsection{Example: 1d $BF$ theory on an interval}\!\!\!\footnote{The setup of this example is taken from \cite{simpBF}, \cite{discrBF}.}
\label{sss: 1d BF on an interval}
%\begin{expl}
Consider 1d $BF$ theory on the interval $I=[0,1]$ with structure Lie algebra $\g$, which we assume to be unimodular.
%\footnote{In this section we are mainly concerned with the classical theory; for the consistency of quantum theory (in particular, for the quantum master equation to hold), one must require that $\g$ is unimodular.}
It is defined by the BV action 
\begin{equation}\label{1d BF action}
S=\int_I \langle \mc{B},\dr \mc{A}+\frac12[\mc{A},\mc{A}] \rangle\end{equation} 
-- a function on the space of fields 
\begin{equation}\label{1d BF space of fields}
\F=T^*[-1] \Omega^\bt(I,\g)[1]\cong \Omega^\bt(I,\g)[1]\oplus \Omega^\bt_\mr{distr}(I,\g^*)[-1]\ni (\mc{A},\mc{B})
\end{equation}
with $(-1)$-symplectic form $\omega=\int_I \langle \delta \mc{B}\stackrel{\wedge}{,} \delta\mc{A}\rangle $. Subscript ``distr'' in the second summand means that we allow distributional forms, more specifically, delta-like 1-forms supported at the endpoints and discontinuity of 0-forms at the endpoints.\footnote{
$BF$ theory with field $\mc{B}$ treated as a distribution (de Rham current) -- one can call it ``canonical'' $BF$ theory, cf. \cite{discrBF}, \cite[Section 6.3]{CR} -- on a manifold with boundary is a pure BV theory rather than a ``BV-BFV theory'' \cite{CMRpert} (a family of BV theories parametrized by boundary conditions). Here the effective action (and the action on $\mc{F}$, in a regularized sense) satisfy the usual quantum master equation, rather than one corrected by a boundary BFV operator. In this ``canonical'' setting, only the field $\mc{A}$ can be pulled back by the inclusion of the boundary, whereas $\mc{B}$ has the opposite functoriality (can be extended from the boundary into the bulk). Thus, a canonical $BF$ theory does not behave as a functorial field theory in the sense of Segal's axioms: it does not assign symplectic phase spaces (classically) or spaces of states (at the quantum level) to boundaries -- rather, it assigns to boundaries an odd-symplectic space with a boundary BV action. Nevertheless, there is a version of gluing formula for effective actions in canonical $BF$ theory by an inclusion-exclusion formula, cf. \cite[Section 5.4]{discrBF}.
}
%\marginpar{\bl May 14: footnote added}
%1-forms $\delta(t)dt,\delta(1-t)dt$ with support at the endpoints 

%We denote 0- and 1-form components of $\mc{A}$ by $c$, $A$ (with ghost numbers $1$ and $0$). Also, denote 0- and 1-form components of $\mc{B}$ by $A^+$, $c^+$ (with ghost numbers $-1$ and $-2$).
We split the ``superfields'' $\mc{A},\mc{B}$ into components according to form degree on $I$:
\begin{eqnarray}
    \A &=& c+A,\\
    \B &=& A^+ + c^+,
\end{eqnarray}
where the first field on the right is a 0-form on $I$ and the second is a 1-form. The ghost numbers for $c,A,A^+,c^+$ are: $1,0,-1,-2$.

Let 
\begin{equation}
\F'=C^\bt(I,\g)[1]\oplus C_{-\bt}(I,\g^*)[-2]
\end{equation}
with $C_\bt$ the cellular chains of the interval with standard CW decomposition into two 0-cells (endpoints) and one 1-cell (bulk), and $C^\bt$ the cellular cochains. We denote the cellular basis in cochains $C^\bt(I)$ by $e_0,e_1,e_I$ and the dual basis in chains $C_\bt(I)$ by $e^0,e^1,e^I$. Thus, an infrared field is 
\begin{equation}
    (\A'=e_0 c_0 + e_1 c_1 +e_I A_I, \B'= e^0 c^+_0+e^1 c^+_1+e^I  A^+_I),
\end{equation}
where
\begin{equation}
    c_0,c_1\in \g[1],\; A_I\in \g,\; c^+_0,c^+_1\in \g^*[-2],\; A^+_I\in \g^*[-1].
\end{equation}

Consider the SDR
\begin{equation}\label{1d BF SDR A}
    \SDR{(\Omega^\bt(I,\g)[1],\dr)}{(C^\bt(I,\g)[1],\delta)}{\iota_\A}{\pi_\A}{\kappa_\A},
\end{equation}
with $\delta$ the cellular coboundary operator,
given by
\begin{eqnarray}
    \iota_\A\colon & e_0 c_0 + e_1 c_1 + e_I A_I & \mapsto \quad
    (1-t) c_0 + t c_1 + \dr t A_I,\\
    \pi_\A \colon & c+A &\mapsto\quad  e_0 c(0) + e_1 c(1) + e_I \int_0^1 A,\\
    \kappa_\A \colon & c+A &\mapsto\quad 
    %\int_0^t A(t') -t \int_0^1 A(t').
    \int_{[0,1]\ni t'} (\theta(t-t')-t) A(t').
\end{eqnarray}
One also has the dual SDR
\begin{equation}\label{1d BF SDR B}
    \SDR{(\Omega^\bt(I,\g^*)[-1],\dr)}{(C_{-\bt}(I,\g^*)[-2],\dd)}{\iota_\B}{\pi_\B}{\kappa_\B},
\end{equation}
where $\dd$ is the cellular boundary operator. Here:
\begin{gather}
    \iota_\B\colon  e^0 c^+_0 + e^1 c^+_1 + e^I A^+_I \mapsto \quad
    \delta(t)\dr t\, c_0^+ + \delta(1-t)\dr t\, c_1^+ + A_I^+,\\
    \pi_\B\colon  A^++c^+ \mapsto\quad 
    e^0 \int_0^1 (1-t) c^+ +e^1 \int_0^1 t c^+ + e^I \int_0^1 \dr t\, A^+,\\
    \kappa_\B\colon  A^++c^+ \mapsto \quad
    -\int_{[0,1]\ni t'} (\theta(t'-t)-t')c^+(t').
\end{gather}
%\marginpar{sign of $\kappa_\B$?}
Putting together SDRs (\ref{1d BF SDR A}) and (\ref{1d BF SDR B}) ($\A$- and $\B$-sectors), one has the full SDR of fields onto infrared fields,
\begin{equation}\label{1d BF interval SDR}
    \SDR{(\F,d)}{(\F',d'=\delta\oplus\dd)}{\iota=\iota_\A\oplus \iota_\B}{\pi=\pi_\A\oplus \pi_\B}{\kappa=\kappa_\A\oplus \kappa_\B}
\end{equation}
We call this SDR the Whitney--Dupont gauge-fixing, since $\iota_\A$ realizes cell cochains by Whitney elementary forms on 1-simplex and $\kappa_\A$ is the Dupont chain contraction operator on forms on a simplex.\footnote{This construction generalizes to higher-dimensional simplices, see \cite{Getzler_Lie}, \cite{simpBF}, \cite{discrBF}).}

The effective BV action on $\F'$ for this model was computed in \cite{simpBF}, \cite{discrBF} to be 
\begin{equation}
    S'=\underbrace{\sum_{i=0}^1\langle c^+_i,\frac12[c_i,c_i] \rangle
    %+ \langle c^+_1,\frac12[c_1,c_1] \rangle 
    +
    \langle A^+_I, \mathbb{F}_+(\ad_{A_I})\circ c_1 - \mathbb{F}_-(\ad_{A_I})\circ c_0 \rangle}_{S'^\cl} - \iii\hbar \underbrace{\mr{tr}_\g \mathbb{G}(\ad_{A_I})}_{S'^{\mr{1-loop}}},
\end{equation}
where
\begin{eqnarray}
    \mathbb{F}_+(x) &=&\frac{x}{1-e^{-x}}=\sum_{n\geq 0}\frac{B_n^+}{n!}x^n =1+\frac{x}{2}+\frac{x^2}{12}-\frac{x^4}{720}+\cdots,\\ 
    \mathbb{F}_-(x)&=&\frac{x}{e^x-1}=\sum_{n\geq 0}\frac{B_n^-}{n!}x^n=1-\frac{x}{2}+\frac{x^2}{12}-\frac{x^4}{720}+\cdots,\\ 
    \mathbb{G}(x)&=&\log \frac{\sinh (x/2)}{x/2} \label{G(x)}\\
    \nonumber &&=\sum_{n\geq 0,n\neq 1}\frac{B_n}{n\cdot n!}x^n=1+\frac{x^2}{2\cdot 12}-\frac{x^4}{4\cdot 720}+\cdots,
\end{eqnarray}
with $B_n^\pm$ the Bernoulli numbers. The induced differential $Q'^\cl$ on $\F'$ is, by Theorem \ref{thm: Q' and p_int via BV pushforward}, the Hamiltonian vector field on $\F'$ generated by $S'^\cl$ or, equivalently, the cotangent lift of 
\begin{equation}
Q'_\A= \sum_{i=0}^1\langle \frac12 [c_i,c_i],\frac{\dd}{\dd c_i} \rangle+
\langle \mathbb{F}_+(\ad_{A_I})\circ c_1-\mathbb{F}_-(\ad_{A_I})\circ c_0 , \frac{\dd}{\dd A_I}\rangle
\end{equation}
-- a formal vector field on $C^\bt(I,\g)[1]$.

The vector field $H^\cl$ (\ref{Hhat classical}) in this example is
%\marginpar{Labels for terms except (a) can be removed if we don't need them}
\begin{multline}\label{1d BF Hhat classical}
    H^\cl= 
    \int_0^1 \left\langle \kappa (F_\A) +d_\A \kappa(\A),\frac{\delta}{\delta\A} \right\rangle + \left\langle \kappa(d_\A \B)+d_\A \kappa(\B)),\frac{\delta}{\delta \B} \right\rangle\\
    =
    \int_0^1
    \left\langle \kappa (d_A c) + [c,\kappa (A)],\frac{\delta}{\delta c} \right\rangle+
    \underbrace{\left\langle d_A \kappa(A),\frac{\delta}{\delta A}
    \right\rangle}_{(a)}+\\
    +
    \left\langle \kappa (d_A A^+ +[c,c^+])+[c,\kappa(c^+)],\frac{\delta}{\delta A^+} \right\rangle
    +\left\langle d_A \kappa (c^+),\frac{\delta}{\delta c^+} \right\rangle.
    %\langle d_A \kappa(A),\frac{\delta}{\delta A} \rangle+ \langle \kappa d_A c,\frac{\dd}{\dd c}  \rangle+\cdots
\end{multline}

\subsubsection{The $L_\infty$ projection $\pi_\ii^\cl$} 
The flow $\Phi_T$ of the vector field $-H^\cl$ applies to a connection 1-form $A$ a $T$-dependent gauge transformation relative to the endpoints of the interval $I$ (cf. term (a) in (\ref{1d BF Hhat classical})) which, as $T\ra\infty$, ``straightens'' $A$ to a gauge-equivalent constant connection\footnote{Note that a connection is constant iff it is annihilated by $\kappa$. 
%{\color{blue} NOT SURE THIS FOOTNOTE  IS NECESSARY}
} 
on $I$. Thus, we have the following. 
\begin{prop}
$\pi^\cl_\ii$ applied to $A$ yields
\begin{equation}\label{1d BF pi_int^cl(A) = log hol(A)}
    \pi_\ii^\cl(A)=\lim_{T\ra\infty} \pi(\Phi_T(A)) = -e_I\log 
    P\mr{exp}\left(-\int_0^1 A \right)
        %\mr{hol}_I(A)
    \in C^1(I,\g)\cong \g
\end{equation}
-- log of the parallel transport (path ordered exponential) of $A$ along $I$. 
\end{prop}
Note that (\ref{1d BF pi_int^cl(A) = log hol(A)}) is a nonlinear deformation of the integration map $\pi \colon A\mapsto \int_I A$.
% \footnote{
% On the other hand, one has the homological perturbation formula $i_\ii^\cl=\sum_{n\geq 0} (-K Q_\ii)^n i$. 
% %Applying it to an observable $O'(A_I)$ depending only on the infrared connection and 
% Comparing it with (\ref{1d BF pi_int^cl(A) = log hol(A)}), one recovers the Magnus expansion  of the log of parallel transport (continuous version of Baker-Campbell-Hausdorff formula). {\color{blue} More details? Also, cite Bandiera-Sch\"atz.}
% }
%\end{expl}

More generally, applied to a general point in $\F^\formal$, $\pi_\ii^\cl$ gives
\begin{equation}
    \pi_\ii^\cl(c+A,A^++c^+)=e_0 c(0)+e_1 c(1)- e_I \log P\exp\left(-\int_0^1 A\right)+\alpha(c,A,A^+,c^+),
\end{equation}
with $\alpha\in C_{-\bt}(I,\g^*)[-2]$ a certain chain with complicated dependence on the input fields.\footnote{In special cases $\alpha$ simplifies: if $c^+\in \mr{im}(\iota)$, we have $\alpha(c+A,A^+)=e^I \int_0^1 \dr t\, A^+$. If $c=0$, we have $ \alpha(A,A^++c^+)=e^I \int_0^1 \dr t\, A^+ + \beta(A,c^+)$, with $\beta$ a complicated 0-chain.
}
% It does simplify in special cases:
% \begin{itemize}
%     \item If $c^+\in \mr{im}(\iota)$, $\Phi_T$ moves $A^+$ by a $\kappa$-exact term and preserves $c^+$, so we have 
%     \begin{equation}
%         \alpha(c+A,A^+)=e^I \int_0^1 dt A^+.
%     \end{equation}
%     \item If $c=0$, $\Phi_T$ again moves $A^+$ by a $\kappa$-exact term (while the trajectory of $c^+$ is nontrivial and depends on $A$), and we have
%     \begin{equation}
%         \alpha(A,A^++c^+)=e^I \int_0^1 dt A^+ + \beta(A,c^+),
%     \end{equation}
%     with $\beta$ a complicated 0-chain.
% \end{itemize}

\begin{rem}
    Let $A_T=\Phi_T(A)$ be the flow of $-H^\cl$ in time $T$, applied to a connection $A$. The flow equation 
    \begin{equation}
        \frac{\dd}{\dd T}A_T = -\dr_{A_T} \kappa(A_T)
    \end{equation}
    with initial condition $A_{T=0}=A$ can be rewritten in terms of Fourier modes $A_{T,n}$ (seen as functions of $T\in [0,\infty)$) of the connection $A_T=\sum_{n\in \ZZ} A_{T,n} e^{2\pi i n t}dt$, as an infinite system of quadratic ODEs
    \begin{equation}
        \frac{\dd}{\dd T} A_{T,n}= - A_{T,n}(1-\delta_{n,0})-\sum_{k\neq 0} \frac{1}{2\pi i k} A_{T,k} (A_{T,n-k}-A_{T,n}).
    \end{equation}
    subject to the initial condition $A_{T=0,n}=A_n$.
\end{rem}

\subsubsection{The $L_\infty$ inclusion $\iota^\cl_\ii$} 
We find
\begin{multline}\label{1d BF iota_int^cl}
    \iota_\ii^\cl(\A'=e_0c_0+ e_1 c_1+ e_I A_I, \B'=e^0c^+_0+e^1 c^+_1+ e^I A^+_I)=\\
    =\lim_{T\ra \infty} \Phi_T \iota(\A',\B')\\
    =
    \frac{e^{-t \ad_{A_I}}-e^{-\ad_{A_I}}}{1-e^{-\ad_{A_I}}}\circ c_0+
     \frac{1-e^{-t \ad_{A_I}}}{1-e^{-\ad_{A_I}}}\circ c_1 + \dr t\, A_I\\
     +
     \delta(t)\dr t\, c^+_0 + \delta(1-t)\dr t\, c^+_1 +
     \frac{\ad^*_{A_I}}{1-e^{-\ad^*_{A_I}}} e^{-t \ad^*_{A_I}}\circ A^+_I
\end{multline}
Here the $c$ component is the solution of $\dr_{A_I}c=\mr{const}$, $c(0)=c_0, c(1)=c_1$; 
%(cf. term (b) in (\ref{1d BF Hhat classical})); 
the $A^+$ component is the solution of $\dr_{A_I}A^+=0$, $\int_0^1 \dr t \, A^+= A^+_I$ (this is the conditional critical point equation (\ref{conditional crit point eq}) supplemented by the condition that the field projects onto the respective infrared field). %(cf. term (c) in (\ref{1d BF Hhat classical})).

%\begin{rem}
\subsubsection{Magnus expansion}\!\!\!\footnote{The observation that one can recover Magnus expansion from homological perturbation lemma was made in \cite{BS} and in \cite{MPIM_interns}.}
    Comparing formula (\ref{1d BF pi_int^cl(A) = log hol(A)}) with the homological perturbation formula 
    \begin{equation}\label{1d BF i_int^cl HPL}
      i^\cl_\ii= \sum_{n\geq 0}(- K Q_\ii)^n i,  
    \end{equation}
     one recovers the Magnus expansion of the log of parallel transport (the ``continuous version'' of the Baker--Campbell--Hausdorff formula).

    As an illustration, let us compute the terms $n=0,1$ in the series (\ref{1d BF i_int^cl HPL}) for a linear observable $O'=\langle \xi, A_I\rangle$, with $\xi\in \g^*$ a parameter:
    \begin{multline}
        \langle \xi, A_I \rangle \xra{i} \left\langle \xi, \int_0^1 \dr t \ul{A}(t) \right\rangle \xra{Q_\ii} \left\langle \xi,-\int_0^1 \dr t [\ul{A}(t),c(t)] \right\rangle \\
        \xra{K} 
        \left\langle \xi, \frac12 \int_0^1 \dr t \left[ \ul{A}(t)+\int_0^1 \dr t'' \ul{A}(t''), \int_0^1 \dr t' (\theta(t-t')-t)\ul{A}(t') \right] \right\rangle\\
        =\Big\langle \xi, \frac12 \int_0^1 \dr t \int_0^1 \dr t' [\ul{A}(t),\ul{A}(t')] (\underbrace{\theta(t-t')-t-t'+\frac12}_{f(t,t')}) \Big\rangle \\
        =
        \Big\langle \xi,\frac12\int_0^1 \dr t\int_0^{t} \dr t' [\ul{A}(t),\ul{A}(t')] \Big\rangle
    \end{multline}
    Here we were writing $A=\dr t \ul{A}$ with $\ul{A}$ a $\g$-valued function on the interval. In the transition to the last line we used that 
    $[\ul{A}(t),\ul{A}(t')]$ is skew-symmetric under the swap of $t$ and $t'$ and that skew-symmetrizations of $f(t,t')$ and $\theta(t-t')$ agree.
    %$f(t,t')-f(t',t)=\theta(t-t')-\theta(t'-t)$. 
    In the transition to the second line we used that $(\kappa^\vee)^\mr{sym}(a\otimes b)=\frac{\mr{id}+ip}{2}(a)\otimes \kappa^\vee(b)+\kappa^\vee(a)\otimes \frac{\mr{id}+ip}{2}(b)$, cf. (\ref{Ksym}).  Thus, we have
    \begin{multline}
        i_\ii^\cl \langle \xi,A_I \rangle = (i-KQ_\ii i+\cdots)\langle \xi,A_I \rangle \\= \left\langle \xi,  \int_0^1 \dr t \ul{A}(t)-\frac12 \int_0^1 \dr t \int_0^t \dr t' [\ul{A}(t),\ul{A}(t')]  +\cdots\right\rangle
    \end{multline}
    where in the r.h.s. one recognizes $\xi$ paired with the Magnus expansion of $\displaystyle -\log P\exp\left(-\int_0^1 A\right)$.

%    {\color{gray} Cite Bandiera-Sch\"atz. Maybe cite Magnus.}
%\end{rem}

\subsection{Example: 1d $BF$ theory on a circle: two gauge-fixings} %\leavevmode\\
\label{ss: 1d BF on S^1}
\subsubsection{Whitney--Dupont gauge (circle with a distinguished point)}
The example of Section \ref{sss: 1d BF on an interval} can be straightforwardly modified to 1d $BF$ theory on a circle, with the action (\ref{1d BF action}) and fields (\ref{1d BF space of fields}), with the interval $I$ replaced by a circle $S^1$, which we parametrize by $t\in \mathbb{R}/\mathbb{Z}$. For the infrared fields we take
\begin{equation}\label{1d BF on S^1, IR fields}
    \F'=C^\bt(S^1,\g)[1]\oplus C_{-\bt}(S^1,\g^*)[-2]
\end{equation}
-- the cellular cochains and chains of $S^1$ with CW decomposition into one 0-cell $t=0$ and one 1-cell. We denote the cellular bases in cochains and chains $e_0, e_{S^1}$ and $e^0,e^{S^1}$, respectively. Thus, an infrared field is
\begin{equation}
    (\A'=e_0c_\zm+e_{S^1}A_\zm,\B'=e^0 c^+_\zm+ e^{S^1} A^+_\zm),
\end{equation}
with $c_\zm\in \g[1]$, $A_\zm\in\g$, $c^+_\zm\in \g^*[-2]$, $A^+_\zm\in \g^*[-1]$. The subscript zm stands for ``zero-mode.''

One has SDR for the theory on $S^1$ obtained from (\ref{1d BF interval SDR}) by gluing the endpoints of the interval,
\begin{equation}
    \SDR{(\F,d)}{(\F',0)}{\iota_\mr{WD}}{\pi_\mr{WD}}{\kappa_\mr{WD}}
\end{equation}
with
\begin{multline}
    \iota_\mr{WD}\colon  (e_0 c_\zm+ e_{S^1} A_\zm, e^{S^1} A^+_\zm+ e^0 c^+_\zm) \mapsto (c_\zm+\dr t\, A_\zm,\; A^+_\zm+\delta(t)\dr t\, c^+_\zm),\\
    \pi_\mr{WD}\colon (c+A,A^++c^+)\mapsto \Big(e_0 c(0)+e_{S^1} \oint A,\; e^{S^1} \oint \dr t A^+ + e^0 \oint c^+\Big),\\
    \kappa_\mr{WD}\colon (c+A,A^++c^+)\mapsto \Big(\oint (\theta(t-t')-t)A(t'),\;\oint (-\theta(t'-t)+t')c^+(t')\Big).
\end{multline}
Here $\oint$ is the integral over the entire $S^1$; subscript $\mr{WD}$ stands for ``Whitney--Dupont.''

The effective action of the model on the circle is
\begin{equation}
    S'=\underbrace{\langle c^+_\zm,\frac12[c_\zm,c_\zm] \rangle}_{S'^\cl}+ \langle A^+_\zm,[A_\zm,c_\zm] \rangle - \iii\hbar \underbrace{\mr{tr}_\g \mathbb{G}(\ad_{A_{S^1}})}_{S'^{\mr{1-loop}}}
\end{equation}
with $\mathbb{G}$ as in (\ref{G(x)}). The induced differential on $\F'$ is the Hamiltonian vector field generated by $S'^\cl$ or, equivalently, the cotangent lift of
\begin{equation}
    Q'_\A = \langle \frac12 [c_\zm,c_\zm],\frac{\dd}{\dd c_\zm}\rangle +
    \langle [A_\zm,c_\zm],\frac{\dd}{\dd A_\zm} \rangle.
\end{equation}

The vector field $H^\cl$ is again given by (\ref{1d BF Hhat classical}), and yields the following $L_\infty$ maps $\iota_{\mr{WD},\ii}^\cl,\pi_{\mr{WD},\ii}^\cl$:
\begin{multline}
    \iota_{\mr{WD},\ii}^\cl(e_0 c_\zm+e_{S^1}A_\zm, e^{S^1} A^+_\zm+ e^0 c^+_\zm)
    \\=
    c_\zm+\dr t\, A_\zm + \frac{\ad^*_{A_\zm}}{1-e^{-\ad^*_{A_\zm}}}e^{-t \ad^*_{A_\zm}}\circ A^+_\zm + \delta(t)\dr t\, c^+_\zm,
\end{multline}
\begin{equation}
    \pi_{\mr{WD},\ii}^\cl(c+A,A^++c^+)=e_0 c(0)-e_{S^1} \log P\exp \left(-\int_0^1 A\right) + \alpha,
\end{equation}
with $\alpha\in C_{-\bt}(S^1,\g^*)[-2]$ a chain with complicated dependence on $c,A,A^+,c^+$.

\subsubsection{Lorenz gauge (circle without a distinguished point)}
\label{sss 1d BF on S^1 Lorenz gauge}
One can identify infrared fields (\ref{1d BF on S^1, IR fields}) with doubled cohomology of the circle,
\begin{equation}\label{1d BF on S^1 IR fields as doubled cohomology}
    \F'=H^\bt(S^1,\g)[1]\oplus H^\bt(S^1,\g^*)[-1].
\end{equation}
Thus, we identify $e_0$ and $e^{S^1}$ as the generator $1$ of $H^0(S^1)$ and we identify $e_{S^1}$ and $e^0$ as the generator $[dt]$ of $H^1(S^1)$.

One can consider the following SDR
\begin{equation}
    \SDR{(\F,d)}{(\F',0)}{\iota_\mr{L}}{\pi_\mr{L}}{\kappa_\mr{L}}
\end{equation}
with
\begin{multline}
    \iota_\mr{L}\colon (e_0 c_\zm+e_{S^1} A_\zm,\; e^{S^1} A^+_\zm+e^0 c^+_\zm) \mapsto (c_\zm+\dr t\, A_\zm, A^+_\zm+\dr t\, c^+_\zm),\\
    \pi_\mr{L}\colon (c+A,A^++c^+)\mapsto \Big( e_0 \oint \dr t\, c + e_{S^1} \oint A,\; e^{S^1}\oint \dr t\, A^+ + e^0 \oint c^+\Big),\\
    \kappa_\mr{L} \colon (c+A,A^++c^+)\mapsto \Big(\oint (\theta(t-t')-t+t'-\frac12)A(t'),\; \oint (\theta(t-t')-t+t'-\frac12)c^+(t') \Big).
\end{multline}
Here $\mr{L}$ stands for ``Lorenz.'' One has a non-normalized chain homotopy compatible with this SDR (cf. Remark \ref{rem: non-normalized chain homotopy}), given by
\begin{equation}\label{1d BF on S^1 barkappa}
    \bar\kappa_\mr{L}\colon (\A,\B) \mapsto (\dr^* \A, \dr^* \B),
\end{equation}
with $\dr^*$ the Hodge codifferential associated with the metric $(\dr t)^2$ on $S^1$, hence the name ``Lorenz SDR'': gauge-fixing $\bar\kappa(\A,\B)=0$ is the usual Lorenz gauge condition.

The vector field $H^\cl$ associated to $\bar\kappa$ is
%the non-normalized chain homotopy 
\begin{multline}
    H^\cl=\oint \left\langle \dr^* \dr_A c+[c,\dr^*A],\frac{\delta}{\delta c} \right\rangle + \left\langle \dr_A \dr^* A ,\frac{\delta }{\delta A} \right\rangle
    +\\
    +\left\langle \dr^*(\dr_A A^+ +[c,c^+])+[c,\dr^*c^+],\frac{\delta}{\delta A^+} 
    \right\rangle + \left\langle \dr_A \dr^* c^+,\frac{\delta }{\delta c^+} \right\rangle.
\end{multline}
The corresponding $L_\infty$ maps $\iota_\mr{L,int}^\cl,\pi_\mr{L,int}^\cl$ are:
\begin{equation}
    \iota_\mr{L,int}^\cl(e_0 c_\zm+e_{S^1}A_\zm,e^{S^1}A^+_\zm+e^0 c^+_\zm)= (c_\zm+\dr t\, A_\zm,\; A^+_\zm+\dr t\, c^+_\zm),
\end{equation}
\begin{multline}\label{1d BF on S^1 pi_int in Lorenz gauge}
    \pi_\mr{L,int}^\cl(c+A,A^++c^+)\\
    = %e_0 \beta(A,c) -e_{S^1} U(A)\left(\log P\exp\left(-\int_0^1 A\right)\right) U(A)^{-1}+ e^{S^1}\gamma(c+A,A^++c^+)+e^0\delta(A,c^+).
    e_0 (\cdots) -e_{S^1} U(A)\left(\log P\exp\left(-\int_0^1 A\right)\right) U(A)^{-1}+ e^{S^1}(\cdots)+e^0(\cdots).
\end{multline}
Here $U(A)$ is an element of the group $G$ integrating $\g$ with complicated dependence on $A$; 
%$\beta\in \g$, $\gamma\in \g^*[-1]$ and $\delta \in \g^*[-2]$ are objects with complicated dependence on inputs.
coefficients $(\cdots)$ also have complicated dependence on the fields.

As a consequence of (\ref{1d BF on S^1 pi_int in Lorenz gauge}), if we have an observable $O'(A_\zm)$ given by a $\g$-invariant function of the infrared connection $A_\zm\in \g$, then
\begin{equation}\label{i_int^cl(O') for a class function on S^1}
i_\ii^\cl(O')(A)=((\pi_{\mr{L},\ii}^\cl)^*O')(A)=O'\left(-\log P\exp\left(-\int_0^1A\right)\right)
\end{equation}
-- the natural lift of a gauge-invariant observable. E.g., if $O'=\mr{tr}_R e^{-A_\zm} $, with $\mr{tr}_R$ the trace in some representation $R$ of $G$, then
\begin{equation}\label{i_int^cl(IR Wilson loop) on S^1}
    i_\ii^\cl(O')(A)= \mr{tr}_R P\exp\left(-\oint A\right).
\end{equation}

We summarize the differences between Whitney--Dupont gauge and Lorenz gauge:
\begin{itemize}
    \item in WD gauge, the point $t=0$ plays a distinguished role, in Lorenz gauge it does not.
    \item $\iota_\mr{WD}$ and $\iota_\mr{L}$ differ in the $c^+$ term. $\pi_\mr{WD}$ and $\pi_\mr{L}$ differ in the $c$ term.
    \item Lorenz SDR treats the $\A$ and $\B$ sectors symmetrically, whereas Whitney--Dupont SDR does not.
    \item $\iota_\mr{L}$ lands in smooth forms and $\kappa_\mr{L}$ maps smooth forms to smooth forms. On the other hand, $\iota_\mr{WD}$ lands in distributional forms and $\kappa_\mr{WD}$ maps smooth $c^+$ forms to 0-forms discontinuous at $t=0$.
    \item In Lorenz gauge, $\iota$ does not get undeformed by the interaction, while in WD gauge the $A^+$ component attains a deformation.
    \item $\pi$ gets deformed in both gauges, but in WD case the answer is more explicit (e.g. the $A$-component is known exactly, not up to an implicit conjugation).
\end{itemize}

% {\color{gray} To add: 
% %explicit SDR in $\mc{A}$-sector of 1d BF on the interval. 
% 1d BF on a circle - there one can use non-normalized homotopy $\bar\kappa=d^*$. Actually, with this gauge-fixing (not singling out a basepoint on the circle, one doesn't get $\log hol(A)$.. So maybe it is not worth adding. But $i_\ii^\cl$ of a \emph{$G$-invariant} IR observable $O'(A_{zm})$ is $O(A)=O'(-\log P \exp -\oint A)$ which is nice.
% %Also: cite Bandiera-Sch\"atz. 
% %Mention Magnus expansion
% }

\section{TQM as a 1d AKSZ theory}
\label{ss: TQM as 1d AKSZ}
%{\color{gray} Boundary (BV-BFV) problems for $i_\ii,p_\ii,K_\ii$. Feynman diagrams in $\tau$= cable diagrams for HPL series for these objects. Examples: $\F$ $n$-dim AKSZ theory $\Rightarrow$ $\tau$ is an $(n+1)$-dim AKSZ (of MQ/instantonic type?). Sub-examples: 1d BF $\ra$ 2d cohomological YM, 3d CS $\ra$ 4d BF+BB probably?}

Given an interacting BV theory $(\F,\omega,S=S_0+S_\ii)$, consider a 1d AKSZ theory $\tau$ on the interval $%I_T=
[0,T]$ with target $\N=T^*\F$.\footnote{This theory is a classical Lagrangian description of the topological quantum mechanics $\tau$ of Section \ref{ss: TQM perspective}, see Remark \ref{rem: 1d AKSZ can quantization}, hence we use the same label $\tau$ for it.}
%\marginpar{Are we assuming that $S$ does not depend on $\hbar$? What if it has quantum corrections?}

We will assume for simplicity that $S$ does not depend on $\hbar$. The quantum master equation for $S$ splits into
\begin{eqnarray}
    \{S,S\} &=& 0,\\
    \Delta S &=& 0. \label{Delta S=0}
\end{eqnarray}

We denote $x^i$ the coordinates on $\F$ and $p_i$ the dual coordinates on the cotangent fiber $\F^*$. The target $\N$ is equipped canonical with $0$-symplectic structure $\omega_\N=\langle \delta p,\delta x \rangle$ and AKSZ Hamiltonian
%\marginpar{$C^\infty(\N)$ or $C^{\infty,\formal}(\N)\simeq \wh{S}\N^*$?}
\begin{equation}\label{Theta_N}
\begin{aligned}
    \Theta_\N(x,p)&=\omega^{-1}\left(p,\frac{\dd S}{\dd x}\right) + \frac12 \omega^{-1}(p,p) \\
    &= \langle p, d_\F x \rangle+\underbrace{\omega^{-1}\left(p,\frac{\dd S_\ii}{\dd x}\right)}_{\Theta_\ii} + \frac12 \omega^{-1}(p,p)\;\; \in C^\infty(\N)_1
\end{aligned}
\end{equation}
of ghost number $1$.
In this subsection we denote the differential on $\F$ by $d_\F$ to avoid confusion with other differentials. The Hamiltonian vector field generated by $\Theta_\N$ is the cohomological vector field on $\N$
\begin{equation}\label{Q_N}
    Q_\N=Q^\mr{cot.\,lift}+\omega^{-1}\left(p,\frac{\dd}{\dd x}\right),
\end{equation}
where the first term is the cotangent lift of $Q=\{S,-\}$ from $\F$ to $\N$.\footnote{
One has an isomorphism $T^*\F\simeq T[1]\F$ induced fiberwise by $(\omega^\#)^{-1}$. Thus, the target $\N$ can be identified with $T[1]\F$ with base coordinates $x^i$ and shifted tangent fiber coordinates $\theta^i$. 
Under the identification $C^\infty(\N)=C^\infty(T[1]\F)\cong \Omega^\bt(\F)$, the target AKSZ Hamiltonian (\ref{Theta_N}) is $\Theta_\mc{N}=\mathsf{d}S+\omega$ -- a form on $\F$ of mixed degree 1 and 2. The cohomological vector field (\ref{Q_N}), as a derivation of $\Omega^\bt(\F)$, is $Q_\mc{N}=\LL_Q+\mathsf{d}$, with $\mathsf{d}$ the de Rham operator on forms on $\F$. The target $0$-symplectic form is $\omega_\N=\omega(\delta\theta,\delta x)$.
%In terms of $(x,\theta)$, one has 
%Denote the shifted tangent fiber coordinates on $T[1]\F$ by $\theta^i$, with $x^i$ the base coordinates. Then
}

Thus, the space of fields of the 1d AKSZ theory $\tau$ is
\begin{equation}\label{1d AKSZ F^tau}
    \F^\tau = \mr{Map}(T[1][0,T]
    ,T^*\F)=\Omega^\bt([0,T],\F)\oplus \Omega^\bt([0,T],\F^*).
\end{equation}
It is parametrized by two AKSZ superfields 
\begin{equation}
\til{x}=x+p^+ \in \Omega^\bt([0,T],\F),\quad
\til{p}=p+x^+ \in \Omega^\bt([0,T],\F^*)
\end{equation}
with $x,p$ being 0-forms along $[0,T]$ and $p^+,x^+$ being 1-forms. The BV action of the 1d AKSZ theory $\tau$ is
\begin{multline}\label{S^tau}
    S^\tau=\int_0^T \langle p, \dr_t x\rangle + \Theta_\N(\til{x},\til{p})
    \\=
    \int_0^T \langle p,\dr_t x \rangle + \langle \til{p},d_\F \til{x} \rangle +\frac12 \omega^{-1}(\til{p},\til{p})+\Theta_\ii(\til{x},\til{p}).
\end{multline}
The $(-1)$-symplectic form on $\F^\tau$ is 
\begin{equation}\label{omega^tau}
    \omega^\tau=\int_0^T \langle \delta \til{p} ,\delta\til{x} \rangle=
    \int_0^T \langle \delta x,\delta x^+ \rangle + \langle \delta p,\delta p^+\rangle.
\end{equation}
We consider the following gauge-fixing Lagrangian $\LL^\tau\subset \F^\tau$:
\begin{equation}\label{L^tau}
    \LL^\tau=\Big\{\Big(\til{x}=x+\dr t\, \kappa(x),\til{p}=p-\dr t\, \kappa^\vee(p)\Big)\;\Big|\; (x,p)\in \Omega^0([0,T],\F\oplus \F^*)\Big\},
\end{equation}
i.e., on $\LL^\tau$ the 0-form components of superfields are free and 1-form components are dependent on 0-form components.\footnote{
One can express $\LL^\tau$ %this Lagrangian 
by realizing $\F^\tau$ as the cotangent bundle of the ``trivial Lagrangian'' $\LL^\tau_0=\Omega^0([0,T],\F\oplus\F^*)$ (given by setting 1-forms to zero) and then deforming $\LL^\tau_0$ to the graph Lagrangian
$\LL^\tau=\mr{graph}(\delta\Psi)$ for the gauge-fixing fermion $\Psi=\int_0^T \dr t \,\langle p,\kappa(x)\rangle\in C^\infty(\LL^\tau_0)_{-1}$. 
}
Here $t$ is the coordinate on the interval $[0,T]$.

The gauge-fixed action is:
\begin{multline}\label{S^tau on L^tau}
    S^\tau|_{\LL^\tau} = \int_0^T \dr t \big( \langle p, (\dd_t-P'')x \rangle - \omega^{-1}(p,\kappa^\vee(p)) \big)+ \\
    + \Theta_\ii(x+\dr t\,\kappa(x),p-\dr t\, \kappa^\vee(p))\\
    = \int_0^T \dr t (\langle p,\dd_t x \rangle-\mathbb{H}(x,p)).
\end{multline}
Here $P''=\mr{id}-\iota\circ \pi$ is the projector onto the second summand in $\F=\F'\oplus \F''$ and
\begin{multline}\label{HH}
    \mathbb{H}(x,p)=\{\Theta_\N,\big\langle p,\kappa(x) \big\rangle\}_\N\\
    = \omega^{-1}\left(p,\kappa^\vee \frac{\dd S}{\dd x}\right)+
    \omega^{-1}\left(p,\Big\langle \kappa(x),\frac{\dd}{\dd x} \Big\rangle \frac{\dd S}{\dd x}\right)+\omega^{-1}(p,\kappa^\vee(p))\;\; \in C^\infty(\N)_0
\end{multline}
is the classical Hamiltonian.

\begin{rem}
    The phase space that the AKSZ theory $\tau$ assigns to a point is the target Hamiltonian dg manifold 
    \begin{equation}\label{1d AKSZ target quadruple}
    (\N=T^*\F,\omega_\N,\Theta_\N,Q_\N).
    \end{equation}
\end{rem}

\begin{rem}\label{rem: 1d AKSZ can quantization}
    The canonical quantization $x\mapsto x\cdot$, $p\mapsto -\iii\hbar \frac{\dd}{\dd x}$ prescription makes the following assignments:
    \begin{enumerate}[(i)]
        \item The AKSZ target $\N$ (as a symplectic manifold)
        %(or equivalently the phase space %$\Phi^\tau_\mr{pt}=T^*\F$ 
        %that the AKSZ theory $\tau$ assigns to a point)
        quantizes to $\fun$.
        \item The target AKSZ Hamiltonian $\Theta_\N$ quantizes to the BV differential $-\iii\hbar Q_\hbar$.\footnote{The two natural orderings -- ``$p$ to right of $x$'' vs. ``$p$ to the left of $x$'' yield the same quantization due to $\Delta S=0$ (\ref{Delta S=0}).}
        Thus, the target as Hamiltonian dg manifold (\ref{1d AKSZ target quadruple})
        %the phase space for a point (or equivalently the AKSZ target) of $\tau$, seen as a dg manifold, 
        quantizes to the BV complex $(\fun,Q_\hbar)$.
        \item The classical Hamiltonian $\mathbb{H}$ quantizes to $-\iii\hbar H$, with
        $H=\LL_E-\iii\hbar \eta + [\wh\kappa,Q_\ii]$, as in (\ref{Hhat}).\footnote{More specifically, the three terms in the second line in (\ref{HH}) quantize to $-\iii\hbar Q\wh\kappa$, $-\iii\hbar \wh\kappa Q$ and $(-\iii\hbar)^2\eta$.}
    \end{enumerate}
\end{rem}

\begin{rem} %\marginpar{Is this remark ok? }
    More appropriately, one should say that the target is a formal neighborhood of zero in $T^*\F$, $\Theta_\N$ and $\mathbb{H}$ are formal functions, $Q_\N$ is a formal vector field. To simplify the exposition we are suppressing the ``formal'' qualifier.
\end{rem}

\begin{rem}
    For gauge-fixing of the 1d AKSZ theory in (\ref{L^tau}), (\ref{S^tau on L^tau}), (\ref{HH}), one can
    use a non-normalized chain homotopy $\bar\kappa$ instead of a normalized one $\kappa$, cf. Remark \ref{rem: non-normalized chain homotopy}.
\end{rem}

%\marginpar{Add remark: formally integrating out $\til{p}$ yields $d_tS$ $\ra$ descends to the boundary $t=0,T$.}

\begin{rem}\label{rem: S^tau as bulk theory for S}
    One can rewrite the action (\ref{S^tau}) -- completing it to a full square in $\til{p}+\cdots$ -- as
    \begin{multline}\label{S^tau full square}
        S^\tau=\int_0^T \frac12 \omega^{-1}\left(\til{p}+\omega^\# \dr_t x + \frac{\dd S}{\dd x}(\til{x}), \til{p}+\omega^\# \dr_t x + \frac{\dd S}{\dd x}(\til{x})\right) \\
        -\frac12 \underbrace{\omega^{-1}\left(\frac{\dd S}{\dd x}(\til{x}),\frac{\dd S}{\dd x}(\til{x})\right)}_{0\;\mr{by\;CME}}-\dr_t S(\til{x})
    \end{multline}
    %\marginpar{\bl May 14 edit}
    Here the second term vanishes by classical master equation $\{S,S\}=0$. Integrating out $\til{p}$ we obtain a theory with the action 
    \begin{equation}\label{int d_t S}
    -\int_0^T \dr_t S(\til{x})=S(x|_{t=0})-S(x|_{t=T}). 
    \end{equation}
    Thus, fixing $x_\mr{in},x_\mr{out}\in \F$ and fixing some extension $\til{x}$ as a form on $[0,T]$ restricting to $x_\mr{in}$, $x_\mr{out}$ at $t=0,T$, we have
    \begin{equation}\label{rem: S^tau as bulk theory for S, eq1}
        \int
        \mc{D}\til{p}\; e^{\frac\iii\hbar S^\tau} =
        e^{\frac\iii\hbar S(x_\mr{in})} e^{-\frac{\iii}{\hbar}S(x_\mr{out})}.
    \end{equation}
    We stress that this is a different gauge-fixing than $\LL^\tau$ (\ref{L^tau}) used in the rest of this section.

\begin{comment}
    In more detail: consider for simplicity the case $\F'=0$ and a Lagrangian \begin{equation}
    \ol{\LL}^\tau=\{\til{x}\in \mr{im}(\kappa)\oplus \dr t\, \mr{im}(d),\; \til{p}\in \mr{im}(d^\vee)\oplus \dr t\, \mr{im}(\kappa^\vee)\}\quad \subset \F^\tau.
    \end{equation}
    It is chosen so that shifts of $\til{p}$ in (\ref{S^tau full square}) are parallel to $\ol{\LL}^\tau$. Then formally we have
    \begin{equation}
        \int_{\ol\LL^\tau} e^{\frac{\iii}{\hbar}S^\tau}= \int_\LL e^{\frac{\iii}{\hbar}S}\cdot \int_\LL e^{-\frac{\iii}{\hbar}S}=
        \left| \int_\LL e^{\frac{\iii}{\hbar}S}\right|^2.
    \end{equation}
    Here the l.h.s. is the gauge-fixed path integral for the theory $\tau$ and in the middle one has the product of boundary contributions at $t=0$ and $t=T$, cf. (\ref{int d_t S}) -- the product of the gauge-fixed integral of the original theory $(\F,\omega,S)$, with original gauge-fixing $\LL=\mr{im}(\kappa)$, and its complex conjugate.

    {\bl to finish/think through: (a) the integral over the 1-form part of $\til{x}$ seems degenerate, so $\ol{\LL}^\tau$ probably does not fully fix the gauge (but maybe that is a typical behavior in theories like $\int_{M^4} BF+BB \xra{\mr{int\;out\; B}} \int_{M^4} tr(F^2)=\int_{M^4} d(CS)=\int_{\dd M^4} CS$). (b) including $\F'$?}
\end{comment}
\end{rem}

\begin{expl}\label{example: toy interacting scalar: 1d AKSZ theory tau}
    Consider the toy interacting scalar field of Section \ref{sss example: toy interacting quantum scalar}. In this case, the 1d AKSZ theory $\tau$ has target 
    \begin{equation}
       \mc{N}=\underset{\pi,p}{T^*} (\underset{\xi}{T^*[-1]} \underset{x}{\RR}),
    \end{equation}
    with coordinates $x,\xi$ on $T^*[-1]\mathbb{R}$ of ghost degrees $0,-1$ and coordinates $p,\pi$ (of ghost degree $0,1$) in the cotangent fiber over $(x,\xi)$. The structure on $\mc{N}$ is:
    \begin{equation}
        \omega_\N=\delta p\, \delta x+ \delta \pi\, \delta\xi,\quad \Theta_N=\pi S'(x)+\pi p.
    \end{equation}
    The AKSZ superfields are 
    \begin{equation}
    \til{x}=\underset{0}{x}+\dr t\,\underset{-1}{p^+}, \; \til\xi=\underset{-1}{\xi}+\dr t\,\underset{-2}{\pi^+},\; \til{p}=\underset{0}{p}+\dr t\, \underset{-1}{x^+},\; \til\pi=\underset{1}{\pi}+\dr t\, \underset{0}{\xi^+}
    \end{equation}
    where we indicated the ghost degrees of the components.
    The AKSZ action is
    \begin{multline}
        S^\tau=\int_0^T \dr t (p\dot{x}-\pi \dot\xi+\Theta_\N(\til{x},\til{\xi},\til{p},\til{\pi}))\\
        =
        \int_0^T \dr t (p\dot{x}-\pi\dot\xi + \xi^+ S'(x)+p^+ \pi S''(x)+x^+ \pi+\xi^+ p)
    \end{multline}
    The gauge-fixing is
    \begin{equation}
        \LL^\tau\colon \til{x}=x+\dr t\,\xi,\; \til\xi=\xi,\; \til{p}=p,\; \til\pi=\pi+\dr t\, p.
    \end{equation}
    The gauge-fixed AKSZ action is therefore
    \begin{equation}\label{toy interacting scalar: S^tau|L^tau}
        S^\tau|_{\LL^\tau}=\int_0^T \dr t (p\dot{x}-\pi \dot\xi+
        \mathbb{H}(x,\xi,p,\pi)
        %\underbrace{pS'(x)+ \xi\pi S''(x)+p^2}_{\mathbb{H}(x,\xi,p,\pi)}
        )
    \end{equation}
    with 
    \begin{equation}
        \mathbb{H}(x,\xi,p,\pi)=\{\Theta_\N,p\xi\}_\N
        =pS'(x)+ \xi\pi S''(x)+p^2.
    \end{equation}
\end{expl}

\subsection{Path integral formulae for $i_\ii$, $p_\ii$, $K_\ii$ in terms of 1d AKSZ theory}

\begin{thm}\label{thm: i_int,p_int,K_int from 1d AKSZ}
One has the following path integral formulae for $i_\ii,p_\ii,K_\ii$ in terms of gauge-fixed 1d AKSZ theory (\ref{S^tau on L^tau}):
\begin{equation}\label{i_int as 1d AKSZ PI}
    i_\ii(O')(x_\mr{out})=\lim_{T\ra \infty} \int_{p(0)=0,\; x(T)=x_\mr{out}} \mc{D} [x(t)] \mc{D} [p(t)]\, O'(\pi(x(0))) \, e^{\frac{\iii}{\hbar}S^\tau|_{\LL^\tau}},
\end{equation}
\begin{equation}\label{p_int as 1d AKSZ PI}
    p_\ii(O)(x'_\mr{out})=\lim_{T\ra \infty} \int_{p(0)=0,\; x(T)=\iota(x'_\mr{out})} \mc{D} [x(t)] \mc{D} [p(t)]\, O(x(0)) \, e^{\frac{\iii}{\hbar}S^\tau|_{\LL^\tau}},
\end{equation}
\begin{multline}\label{K_int as 1d AKSZ PI}
    K_\ii(O)(x_\mr{out})\\
    =\int_0^\infty \int_{p(0)=0,\; x(T)=x_\mr{out}} 
    \mc{D} [x(t)] \mc{D} [p(t)]\,
   %\left. 
   \underbrace{O(\til{x}|_{t=0, \dr t=\dr T})}_{O\big(x(0)+\dr T\kappa(x(0))\big)} e^{\frac{\iii}{\hbar} S^\tau
   %(\til{x},\til{p})
   }%\right
   \Big|_{\LL^\tau},
\end{multline}
for any input observables $O\in \wh{S}\F^*$, $O'\in \wh{S}\F'^*$. %$x_\mr{out}\in \F$, $x'_\mr{out}\in \F'$.
The outer integral in (\ref{K_int as 1d AKSZ PI}) is in $T$ (one extracts the $\dr T$ component from the integrand, coming from the observable).
\end{thm}
Here the path integrals are understood perturbatively, writing $S^\tau|_{\LL^\tau}$ as the last line of (\ref{S^tau on L^tau}) and understanding $\mathbb{H}$ as perturbation.

To prove the theorem, first we need the following lemma.

\begin{lem}\label{lemma 5.20}
One has
    \begin{equation}\label{e^(-TH)) = 1d AKSZ PI}
    (e^{-TH}O)(x_\mr{out})=
           \int_{p(0)=0,\; x(T)=x_\mr{out}} \mc{D} [x(t)] \mc{D} [p(t)]\, O(x(0)) \, e^{\frac{\iii}{\hbar}S^\tau|_{\LL^\tau}},
    \end{equation}
    with any $T>0$, $O\in \wh{S}\F^*$ and with $H$ as in (\ref{Hhat}) -- the quantum Hamiltonian of the TQM.
\end{lem}

\begin{proof}
First, write 
\begin{equation}\label{x=xout+xhat}
x(t)=\ol{x}_\mr{out}(t)+\xx(t),
\end{equation}
where the first term is 
$\ol{x}_\mr{out}(t)=\begin{cases}
    x_\mr{out},\; t=T,\\
    0,\; t<T
\end{cases}$
-- the discontinuous extension of $x_\mr{out}$ by zero to $t\in [0,T)$.\footnote{Cf. discontinuous extension of boundary fields into the bulk in \cite[Section 2.4]{CMRpert}.
} The r.h.s. of (\ref{e^(-TH)) = 1d AKSZ PI}) then reads
\begin{equation}\label{lemma 5.20 eq1}
    \int_{p(0)=0, \, \xx(T)=0} \mc{D} [\xx(t)] \mc{D} [p(t)]\, O(\xx(0)) e^{\frac{\iii}{\hbar} \left(\langle p(T),x_\mr{out} \rangle+\int_0^T \dr t \left(\langle p,\dd_t \xx \rangle-\mathbb{H}(\xx,p)\right) \right)}
\end{equation}

Consider free expectation values 
\begin{equation}
        \langle f \rangle_0 = \int_{p(0)=0,\; \xx(T)=0} 
        \mc{D} [\xx(t)] \mc{D} [p(t)]\, e^{\frac{\iii}{\hbar}\int_0^T \dr t \langle p, \dd_t \xx\rangle}\cdot f,
\end{equation}
with $f$ a function on the space of paths $(\xx(t),p(t))$. %AKSZ fields.
One has propagators (two-point correlation functions)
\begin{equation}
        \langle p(t)\otimes \xx(t') \rangle_0=-\iii\hbar \theta(t-t')\otimes\mr{id},\quad
        \langle \xx(t)\otimes \xx(t')\rangle_0 = 0,\quad
        \langle p(t)\otimes p(t')\rangle_0 =0.
\end{equation}

Computing (\ref{lemma 5.20 eq1}) by Wick's lemma, we obtain
\begin{equation}
    \left| e^{\langle x_\mr{out}, \frac{\dd}{\dd \xx}\rangle} e^{-T H} O(\xx)  \right|_{\xx=0} = (e^{-TH}O)(x_\mr{out}),
\end{equation}
which proves the lemma.

Here we were understanding $\mathbb{H}$ as a normally-ordered expression, prohibiting self-contractions, which is tantamount to displacing $p$ in $\mathbb{H}(\xx,p)$ to a slightly earlier time $t$ than $\wh{x}$: 
\begin{equation}
:\mathbb{H}(\xx(t),p(t)):=\lim_{\epsilon\ra +0} \mathbb{H}(\xx(t),p(t-\epsilon)).
\end{equation}
\end{proof}

\begin{proof}[Proof of Theorem \ref{thm: i_int,p_int,K_int from 1d AKSZ}]
    Theorem \ref{thm: i_int,p_int,K_int from 1d AKSZ} follows from Theorem \ref{thm: interacting SDR via TQM} and Lemma \ref{lemma 5.20}. Indeed: 
    \begin{itemize}
        \item Set $O=i(O')$ in (\ref{e^(-TH)) = 1d AKSZ PI}) and take the limit $T\ra\infty$. 
        % \begin{equation}
        %     \lim_{T\ra\infty}(e^{-TH}i(O'))(x_\mr{out})=
        %    \lim_{T\ra\infty}\int_{p(0)=0,\; x(T)=x_\mr{out}} \mc{D} [x(t)] \mc{D} [p(t)]\, O'(\pi(x(0))) \, e^{\frac{\iii}{\hbar}S^\tau|_{\LL^\tau}},
        % \end{equation}
        Then the l.h.s. becomes the r.h.s. of (\ref{thm 5.3 i}) applied to $O'$ and evaluated at $x_\mr{out}$, and the r.h.s. becomes the r.h.s. of (\ref{i_int as 1d AKSZ PI}), which proves (\ref{i_int as 1d AKSZ PI}).
        \item Set $x_\mr{out}=\iota(x'_\mr{out})$ in (\ref{e^(-TH)) = 1d AKSZ PI}) and take the limit $T\ra \infty$. Then the l.h.s. becomes the r.h.s. of (\ref{thm 5.3 p}) applied to $O$ and evaluated at $x'_\mr{out}$, whereas the r.h.s. becomes the r.h.s. of (\ref{p_int as 1d AKSZ PI}).
        \item Set $O=\dr T\, \wh\kappa O$ in (\ref{e^(-TH)) = 1d AKSZ PI}) and integrate in $T$ from $T=0$ to $T=\infty$. Then the l.h.s. becomes the r.h.s. of (\ref{thm 5.3 K}) (applied to $O$ and evaluated at $x_\mr{out}$) while the r.h.s. becomes the r.h.s. of (\ref{K_int as 1d AKSZ PI}).
    \end{itemize}
\end{proof}

\begin{expl}
    For the toy interacting scalar field (Example \ref{example: toy interacting scalar: 1d AKSZ theory tau}), the path integral formula for the chain homotopy (\ref{K_int as 1d AKSZ PI}), evaluated on an observable $O(x)$, reads
    \begin{multline}
        K_\ii(O)(x_\mr{out},\xi_\mr{out})=\\
        \int_0^\infty \int_{p(0)=\pi(0)=0,\; x(T)=x_\mr{out},\;\xi(T)=\xi_\mr{out}} \mc{D} [x(t)] \mc{D} [\xi(t)] \mc{D} [p(t)] \mc{D} [\pi(t)] \; O(x(0)+\dr T\, \xi(0))\; e^{\frac\iii\hbar S^\tau|_{\LL^\tau}},
    \end{multline}
    with $S^\tau|_{\LL^\tau}$ as in (\ref{toy interacting scalar: S^tau|L^tau}).
\end{expl}

\subsection{Cable diagrams for $i_\ii$ as Feynman graphs for 1d AKSZ path integral} \label{sss: cable diagrams for i_int via 1d AKSZ}

Consider the path integral (\ref{e^(-TH)) = 1d AKSZ PI}):
\begin{multline}\label{1d AKSZ PI}
    (e^{-TH}O)(x_\mr{out})=
           \int_{p(0)=0,\; x(T)=x_\mr{out}} \mc{D} [x(t)] \mc{D} [p(t)]\, O(x(0)) \, e^{\frac{\iii}{\hbar}S^\tau|_{\LL^\tau}}\\
           = \int_{p(0)=0,\; \xx(T)=0}
           \mc{D} [\xx(t)] \mc{D}[p(t)] \, O(\xx(0))\cdot \exp\frac{\iii}{\hbar}\Big(
           \int_0^T \dr t\, \Big(
           \langle p,(\dd_t-P'')\xx\rangle -\\
           -\underbrace{\omega^{-1}\left(\kappa^\vee(p),\frac{\dd S}{\dd x}\right)}_{\mathbb{H}_a} -
           \underbrace{\omega^{-1}\left( p, \left\langle \kappa(x),\frac{\dd}{\dd x} \right\rangle \frac{\dd S}{\dd x}\right)}_{\mathbb{H}_b}
           -\underbrace{\omega^{-1}(p,\kappa^\vee(p))}_{\mathbb{H}_\eta}
           \Big)\Big|_{x\ra \xx}
           +\underbrace{\langle p(T),x_\mr{out}}_{\mathbb{B}} \rangle
           \Big).
\end{multline}
Here we used the splitting  $x(t)=\bar{x}_\mr{out}(t)+\xx(t)$ as in (\ref{x=xout+xhat}). Let us consider $\int_0^T \dr t\langle p,(\dd_t-P'')\wh{x} \rangle$ as the leading (kinetic) term and the rest of the exponential as a perturbation. We denote $\langle \cdots \rangle$ the Gaussian average with this kinetic term.

The propagator in this perturbation theory -- the line in Feynman graphs -- is
\begin{equation}\label{1d AKSZ propagator alpha}
    \alpha= \frac{\iii}{\hbar}\langle p(t)\otimes \xx(t')  \rangle = 
    %-\iii\hbar 
    \theta(t-t') e^{-(t-t')P''}\otimes \mr{id} \quad \in \Omega^0_\mr{distr}([0,T]^{\times 2},\F^*\otimes \F).
\end{equation}
We will draw it in Feynman graphs as an edge directed from $\xx$ to $p$.

Perturbations in (\ref{1d AKSZ PI}) correspond to vertices in Feynman graphs:

\begin{center}
\begin{tabular}{c|c}
   term in (\ref{1d AKSZ PI})  & vertex in Feynman graphs \\ \hline
    $\mathbb{H}_a\in \mr{Hom}(\F^*,S \F^*)$ & 
    $\vcenter{\hbox{\includegraphics[scale=0.8]{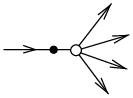}}}$\\
    $\mathbb{H}_b\in \mr{Hom}(\F^*,S \F^*)$ & 
    $\vcenter{\hbox{\includegraphics[scale=0.8]{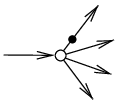}}}$ \\
    $\mathbb{H}_\eta\in \mr{Hom}(S^2\F^*,\mathbb{R})$ & 
    $\vcenter{\hbox{\includegraphics[scale=0.8]{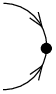}}}$ \\
    $\mathbb{B}\in \mr{Hom}(\F^*,\mathbb{R})$ & 
    $\vcenter{\hbox{\includegraphics[scale=0.8]{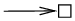}}}$ \\
    $O\in S\F^*$ & 
    $\vcenter{\hbox{\includegraphics[scale=0.8]{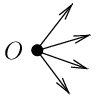}}}$
\end{tabular}
\end{center}
Here incoming half-edges are decorated by $p$ and outgoing ones by $\xx$. Black dot decorating one incident half-edge is a part of the graphical notation for $\mathbb{H}_a$, $\mathbb{H}_b$ (and corresponds to $\kappa^\vee$ or $\kappa$ hitting one $p$ or $x$ argument in the formulae for $\mathbb{H}_a$, $\mathbb{H}_b$); black dot in $\mathbb{H}_\eta$ vertex refers to $\kappa^\vee$ in the formula for $\mathbb{H}_\eta$.  We also adopt the (optional) graphical arrangement convention that outgoing half-edges go to the right, and incoming ones come from the left.

The perturbative formula for the path integral (\ref{1d AKSZ PI}) is then
\begin{equation}\label{1d AKSZ PI perturbative}
    (\ref{1d AKSZ PI})= \sum_\Gamma \frac{(-\iii\hbar)^{-\chi(\Gamma)}}{|\mr{Aut}(\Gamma)|} \Phi^\tau_\Gamma(O,x_\mr{out},T),
\end{equation}
where the sum is over oriented connected graphs $\Gamma$ 
%without short loops (edges connecting a vertex to itself), 
with a unique vertex $O$ and arbitrary finite numbers of $\HH_a$, $\HH_b$, $\HH_\eta$, $\mathbb{B}$ vertices (see Figure \ref{fig: 1d AKSZ Feynman graph}). 
\begin{figure}
    \centering
    \includegraphics[width=0.75\linewidth]{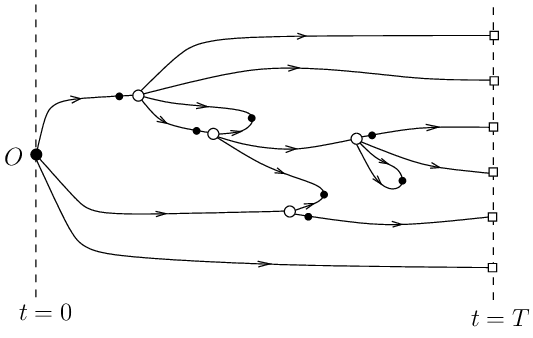}
    \caption{A typical Feynman graph contributing to the perturbative 1d AKSZ path integral (\ref{1d AKSZ PI perturbative}). 
    A particular configuration of times $t_v$ in the integrand in (\ref{Phi_Gamma^tau}) corresponds to a particular horizontal placement of vertices.
    }
    \label{fig: 1d AKSZ Feynman graph}
\end{figure}
We do not allow short loops (edges connecting a vertex to itself) in $\Gamma$. $\chi(\Gamma)$ is the Euler characteristic of the graph and $|\mr{Aut}(\Gamma)|$ is the order of the automorphism group. The weight of the graph is
%\marginpar{\bl May 10: $\pi^*_{uv}\alpha$ -- pullback of a distribution -- mathematically not very clean.. comment?}
\begin{multline}\label{Phi_Gamma^tau}
    \Phi_\Gamma^\tau(O,x_\mr{out},T) \\= \int_{[0,T]^{V_\mr{bulk}}} 
    \left.\left\langle \bigotimes_{v\in V_\mr{bulk}} \dr t_v \HH_{\lambda(v)} \otimes O\otimes \mathbb{B}^{\otimes V_\mathbb{B}},
    \bigotimes_{\mr{edges}\; e=(vu)} \pi^*_{uv}\alpha
    \right\rangle_\Gamma\right|_{t_{v_O}=0,t_{v_\mathbb{B}}=T}
\end{multline}
Here:
\begin{itemize}
    \item The set of vertices of $\Gamma$ is $V=V_\mr{bulk}\sqcup \{v_O\}\sqcup V_\mathbb{B}$ -- ``bulk'' vertices (i.e. those integrated in $t$) of types $\HH_a,\HH_b,\HH_\eta$, the unique vertex $O$ and the vertices $\mathbb{B}$; $\lambda(v)\in \{a,b,\eta\}$ is the type of a bulk vertex.
    \item $\pi_{uv}\colon [0,T]^V \ra [0,T]^{\times 2}$ is the map $(t_1,\ldots,t_V)\mapsto (t_u,t_v)$. 
    \item $\langle-,- \rangle_\Gamma$ is the canonical pairing 
    between $\F$ and $\F^*$ extended to tensor powers, with factors matched according to the combinatorics of $\Gamma$.
    %$(\F^*)^{\otimes \mr{HE}_x}\otimes \F^{\otimes \mr{HE}_p}$
    \item The integral is over the variables $t_v\in [0,T]$, $v\in V_\mr{bulk}$.
\end{itemize}
Note that, due to the form of the propagator (\ref{1d AKSZ propagator alpha}), the integrand in (\ref{Phi_Gamma^tau}) can be nonzero only if for every edge $\langle p(t) \xx(t') \rangle$ in the graph one has $t> t'$ (which is consistent with edges oriented left-to-right and $t$ increasing left-to-right in Figure \ref{fig: 1d AKSZ Feynman graph}).

\subsubsection{Specialization to $i_\ii$}
\label{sss: specialization to i_int}
Evaluating (\ref{1d AKSZ PI}), (\ref{1d AKSZ PI perturbative}) on an infrared observable $O=i(O')$ and taking the limit $T\ra \infty$, we obtain
\begin{equation}\label{1d AKSZ i_int Feynman graph expansion}
    i_\ii(O')(x_\mr{out})=\lim_{T\ra\infty} \sum_\Gamma \frac{(-\iii\hbar)^{-\chi(\Gamma)}}{|\mr{Aut}(\Gamma)|} \Phi^\tau_\Gamma(i(O),x_\mr{out},T)
\end{equation}
Here the contributing graphs -- see Figure \ref{fig:1d AKSZ i_int Feynman graph} -- are the same as in (\ref{1d AKSZ PI perturbative}), except $O$ is now replaced with $O'$, connected by edges $\frac{\iii}{\hbar}\langle p(t) \pi(\xx(0)) \rangle=P'$ to other vertices; we draw these ``infrared'' edges as dashed. 
\begin{figure}
    \centering
    \includegraphics[width=0.75\linewidth]{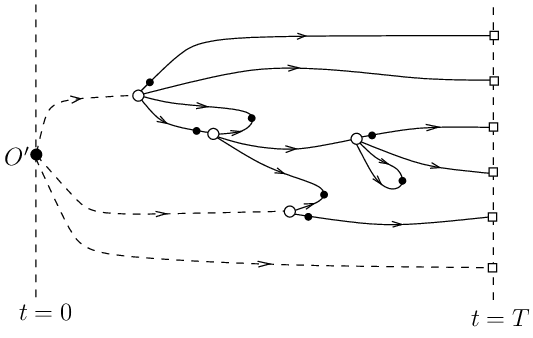}
    \caption{Typical Feynman graph contributing to the expansion of $i_\ii$ (\ref{1d AKSZ i_int Feynman graph expansion}.}
    \label{fig:1d AKSZ i_int Feynman graph}
\end{figure}
Note that $O'$ cannot be connected by edges to $\HH_a$ vertices or $\HH_\eta$ vertices, otherwise the graph vanishes trivially. 

Also note that, due to the exponential decay of the propagator (\ref{1d AKSZ propagator alpha}), the integral (\ref{Phi_Gamma^tau}) in the limit $T\ra\infty$ is supported on configurations $\{t_v\}$ where all $t_v$ are within $\mc{O}(1)$ of $T$. Thus, a typical configuration is where the dashed edges are very long and the other edges are finite.
%\marginpar{Is this written ok?}

Finally, in the Feynman graphs appearing in (\ref{1d AKSZ i_int Feynman graph expansion}) we immediately recognize the cable diagrams for $i_\ii$, Figure \ref{fig:i_int cable diagram}. Thick edges in the cable diagram correspond to edges without a black dot in Figure \ref{fig:1d AKSZ i_int Feynman graph}.

\begin{rem}
    One can also recover cable diagrams for $p_\ii$, by specializing (\ref{1d AKSZ PI}), (\ref{1d AKSZ PI perturbative}) to $x_\mr{out}=\iota(x'_\mr{out})$ and taking $T\ra\infty$:
    \begin{equation}
        (p_\ii O)(x'_\mr{out})=\lim_{T\ra\infty}\sum_\Gamma \frac{(-\iii\hbar)^{-\chi(\Gamma)}}{|\mr{Aut}(\Gamma)|} \Phi^\tau_\Gamma(O,\iota(x'_\mr{out}),T).
    \end{equation}
    In this case, one has dashed ``infrared'' edges $\frac{\iii}{\hbar}\langle \iota^\vee(p(T)) \xx(t)\rangle=P'$ connecting to the $\mathbb{B}$-vertices. In this case the graphs $\Gamma$ cannot contain $\HH_b$ vertices, otherwise the graph vanishes. See Figure \ref{fig:1d AKSZ Feynman graph for p_int} for a typical contributing graph. The integral (\ref{Phi_Gamma^tau}) in the limit $T\ra\infty$ is supported on configurations $t_v=\mc{O}(1)$, i.e., the dashed edges are long while the other edges are finite.
\begin{figure}
    \centering
    \includegraphics[width=0.75\linewidth]{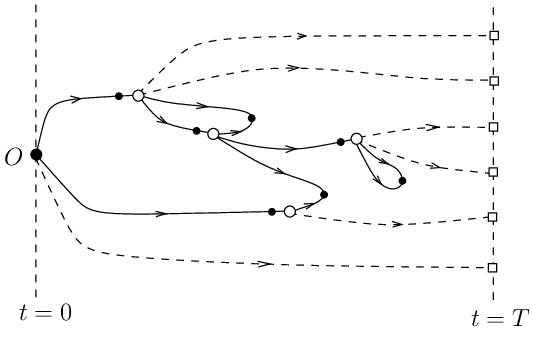}
    \caption{1d AKSZ Feynman graph for $p_\ii$.}
    \label{fig:1d AKSZ Feynman graph for p_int}
\end{figure}

Likewise, one recovers cable diagrams for $K_\ii$ by replacing $O\ra dT \wh\kappa (O)$ in (\ref{1d AKSZ PI}), (\ref{1d AKSZ PI perturbative}) and integrating in $T\in [0,\infty)$:
\begin{equation}
    (K_\ii O)(x_\mr{out})=\int_0^T  \sum_\Gamma \frac{(-\iii\hbar)^{-\chi(\Gamma)}}{|\mr{Aut}(\Gamma)|} \Phi^\tau_\Gamma(\dr T\, \wh\kappa(O),x_\mr{out},T).
\end{equation}
See Figure \ref{fig:1d AKSZ Feynman graph for K_int} for a typical contributing graph.
%\marginpar{\bl May 10: edit $dT\ra \dr T$ in figure}
\begin{figure}
    \centering
    \includegraphics[width=0.75\linewidth]{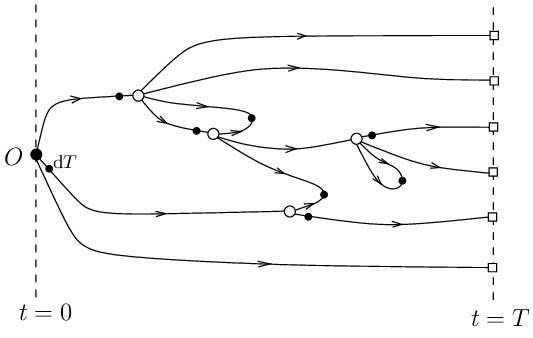}
    \caption{1d AKSZ Feynman graph for $K_\ii$. A new element compared to Figure \ref{fig: 1d AKSZ Feynman graph} is an extra black dot %(standing for $\kappa$ and coming with a $dT$ factor) 
    adjacent to $O$ (understood in this case as $\dr T\,\kappa$).}
    \label{fig:1d AKSZ Feynman graph for K_int}
\end{figure}
\end{rem}

%\marginpar{\bl rem added May 14}
\begin{rem}
    In the case $\F'=0$, we can consider the path integral (\ref{p_int as 1d AKSZ PI}) using the gauge-fixing of Remark \ref{rem: S^tau as bulk theory for S} instead of $\LL^\tau$:
    \begin{equation}\label{rem p_int PI with alternative gauge-fixing, eq1}
        \int \mc{D}x_\mr{in} \int \mc{D}\til{p} \; O(x_\mr{in})\; e^{\frac\iii\hbar S^\tau} \underset{\eqref{rem: S^tau as bulk theory for S, eq1}}{=}  \int \mc{D} x_\mr{in}\; O(x_\mr{in})\; e^{\frac{\iii}{\hbar} S(x_\mr{in})}.
    \end{equation}
    Here we understand that for each $x_\mr{in}$ we fix some extension $\til{x}$ of it from $t=0$ to the interval $[0,T]$, restricting to $x_\mr{out}=0$ at $t=T$. The integral over $\til{p}$ in the l.h.s.\ is with free boundary conditions. The integral in the r.h.s.\ still needs to be gauge-fixed, and then it yields the BV pushforward of $O$. This yields another (formal) proof of Theorem \ref{thm: Q' and p_int via BV pushforward} (ii) in the case $\F'=0$. Indeed, using the formal independence of the BV integral in theory $\tau$ on gauge-fixing and (\ref{p_int as 1d AKSZ PI}), we have that the r.h.s.\ (\ref{rem p_int PI with alternative gauge-fixing, eq1}) equals $p_\ii$ as defined by homological perturbation series.  Note that with this gauge-fixing the integral (\ref{rem p_int PI with alternative gauge-fixing, eq1}) does not depend on $T$ and taking the limit $T\ra\infty$ is a tautological operation.
\end{rem}

\subsection{A class of examples: the case when $(\F,\omega,S)$ is itself an AKSZ theory}
\label{ss: tau for F an AKSZ theory}

Let $(\F,\omega,S)$ be an AKSZ theory on a closed $n$-manifold $M$ with
\begin{equation}
    \F= \mr{Map}(T[1]M,Y)\cong \Omega^\bt(M,Y)
\end{equation}
with target $Y$ a graded vector space equipped with: 
\begin{itemize}
    \item 
    A constant symplectic form $\omega_Y$ of ghost degree $n-1$.
    \item A target Hamiltonian $\Theta_Y\in C^\infty(Y)$ of ghost degree $n$ satisfying the classical master equation $\{\Theta_Y,\Theta_Y\}_{\omega_Y}=0$.
\end{itemize}
%We denote $y^i$ the coordinates on $Y$ and $\bar{y}^i\in \Omega^\bt(M)$ the corresponding AKSZ superfields parametrizing $\F$.
We denote $y$ a point in $Y$ and $\bar{y}$ a point in $\F$ (an AKSZ superfield).

%Let $\mr{ev}\colon T[1] M\times \F \ra Y$ be the evaluation map.
The AKSZ action is
\begin{equation}
    S(\bar{y})=\int_M \frac12 \omega_Y(\bar{y},\dr_M \bar{y})+\Theta_Y(\bar{y})
    %\iota_{d_M^\mr{lifted}} \mr{ev}^* \alpha_Y +\mr{ev}^* \Theta_Y = \int_M \alpha_Y(\bar{y})_{ij}\bar{y}^i d_M\bar{y}^j + \Theta(\bar{y})
\end{equation}
and the $(-1)$-symplectic form on the space of fields is
\begin{equation}
    \omega= \int_M \frac12\omega_Y(\delta \bar{y},\delta \bar{y}).
    %\int_M \mr{ev}^* \omega_Y = 
    %\int_M \omega_Y(\bar{y})_{ij} \delta \bar{y}^i\wedge \delta \bar{y}^j.
\end{equation}
%Here $\alpha=\alpha_{ij}(y)y^i \delta y^j$, $\omega=\omega_{ij}(y)\delta y^i\wedge \delta y^j$ are the coordinate representation of the 1-form $\alpha_Y\in \Omega^1(Y)_{n-1}$ and the symplectic form $\omega_Y=\delta \alpha_Y \in \Omega^2(Y)_{n-1}$.

We can identify the cotangent bundle $T^*\F$ with\footnote{For $\F$ an infinite-dimensional graded vector space, we understand the dual  space $\F^*$ in the weak sense, i.e., as a graded vector space $W$ equipped with a degree zero bilinear form $b\colon \F\otimes W\ra \mathbb{R}$ such that the induced map $b^\#\colon \F\ra W^*$ is injective.}
\begin{equation}
    T^*\F\simeq \mr{Map}(T[1]M,T^*[n]Y)\cong \Omega^\bt(M,Y)\oplus \Omega^\bt(M,Y^*)[n] \;\; \ni (\bar{y},\bar\psi),
\end{equation}
where the canonical symplectic form on $T^*\F$ is 
$\int_M \langle \delta \bar\psi,\delta\bar{y} \rangle$. In our notations, $\psi$ is the cotangent fiber coordinate on $T^*[n]Y$. %\cong Y\oplus Y^*[n]$.
Note that, comparing with notations of (\ref{S^tau}), we have $x=\bar{y}$, $p=\bar\psi$.

The AKSZ theory $\tau$ then has the space of fields (\ref{1d AKSZ F^tau}):
\begin{multline}
    \F^\tau=\mr{Map}(T[1][0,T],\underbrace{\mr{Map}(T[1]M,T^*[n]Y)}_{T^*\F})
    \\ \cong \mr{Map}(T[1]([0,T]\times M), T^*[n]Y).
\end{multline}
In particular, it can be seen as an $(n+1)$-dimensional theory on a cylinder $N=[0,T]\times M$. The target AKSZ structure on $\mc{Y}=T^*[n]Y$ is: 
\begin{itemize}
    \item The canonical degree $n$ symplectic structure $\omega_\mc{Y}=\langle \delta\psi,\delta y \rangle$.
    \item The target AKSZ Hamiltonian 
    \begin{equation}
        \Theta_\mc{Y}=\omega_Y^{-1}\left(\psi,\frac{\dd \Theta_Y}{\dd y}\right)+\frac12 \omega_Y^{-1}(\psi,\psi)\quad \in C^\infty(\mc{Y})
    \end{equation}
    of ghost degree $n+1$.
\end{itemize}
Then the AKSZ action
\begin{equation}
    S^\tau(\til{y},\til\psi)=\int_{N} \langle \til\psi,\dr_N\til{y} \rangle + \Theta_\mc{Y}(\til{y},\til\psi)
\end{equation}
is the action (\ref{S^tau}). We denoted $\til{y},\til\psi$ the AKSZ superfields on $[0,T]\times M$ (as opposed to $\bar{y},\bar\psi$ -- forms % the AKSZ superfields 
on $M$). The $(-1)$-symplectic structure (\ref{omega^tau}) is
\begin{equation}
    \omega^\tau= \int_{N} \langle \delta\til\psi, \delta \til{y}\rangle.
\end{equation}

\begin{expl}
Let $(\F,\omega,S)$ be Chern--Simons theory on a 3-manifold $M$ with coefficients in a quadratic Lie algebra $\g$ (cf. Example \ref{example: Chern-Simons H^cl in Lorenz gauge}) -- an AKSZ theory with target $Y=\g[1]\ni y$, $\Theta_Y=\frac16\langle y,[y,y] \rangle$, $\omega_Y=\frac12\langle \delta y,\delta y \rangle$. Then $\tau$ is the 4d AKSZ theory on the cylinder $N=[0,T]\times M$ with target
\begin{equation}
\mc{Y}\cong \g[1]\oplus \g[2]\ni (y,\psi),
\end{equation}
with symplectic structure $\omega_\mc{Y}=\langle \delta\psi,\delta y\rangle$ and the Hamiltonian 
\begin{equation}
\Theta_\mc{Y}=\langle \psi,\frac12[y,y] \rangle+\frac12 \langle \psi,\psi \rangle.
\end{equation}
The AKSZ theory $\tau$ has the space of fields
\begin{equation}
    \F^\tau=\Omega^\bt(N,\g)[1]\oplus \Omega^\bt(N,\g)[2]\quad \ni (\A,\B)
\end{equation}
with $(-1)$-symplectic form $\int_N \langle \delta\B,\delta \A \rangle$ and the action
\begin{equation}
    S^\tau=\int_N \langle \B,\dr_N\A+\frac12[\A,\A] \rangle + \frac12 \langle \B,\B\rangle.
\end{equation}
One recognizes the 4d ``$BF+B^2$'' theory, or ``$BF$ theory with cosmological term.''
\end{expl}

\begin{expl}\label{example: theory tau for BF}
    Let $(\F,\omega,S)$ be $BF$ theory on an $n$-manifold $M$ with coefficients in a unimodular Lie algebra $\g$ -- an AKSZ theory with target $Y=T^*[n-1]\g[1]= \g[1]\oplus \g^*[n-2] \ni (y,p_y)$, with $\omega_Y=\langle \delta p_y, \delta y\rangle$ and $\Theta_Y=\langle p_y,\frac12[y,y]\rangle$. Then $\tau$ is the AKSZ theory on the $(n+1)$-dimensional cylinder $N=[0,T]\times M$ with target
    \begin{multline}
        \mc{Y}=T^*[n] T^*[n-1] \g[1] \\
        \cong (\underset{y}{\g[1]}\oplus \underset{P_\psi}{\g^*[n-2]})\oplus (\underset{P_y}{\g^*[n-1]}\oplus \underset{\psi}{\g[2]}) \cong T^*[n] (\underset{y}{\g[1]}\oplus \underset{\psi}{\g[2]}),
    \end{multline}
    where we indicated the names $y,\psi,P_y,P_\psi$ for the components of an element in $\mc{Y}$. The target $\mc{Y}$ is equipped with $n$-symplectic structure 
    \begin{equation}
        \omega_\mc{Y}=\langle \delta P_y, \delta y \rangle + \langle \delta P_\psi,\delta\psi \rangle
    \end{equation}
    and the Hamiltonian
    \begin{equation}
        \Theta_\mc{Y}=\langle P_y,\frac12[y,y]\rangle + \langle P_\psi,[y,\psi]\rangle +\langle P_y,\psi \rangle
    \end{equation}
    of degree $n+1$. AKSZ theory $\tau$ has the space of fields
\begin{equation}
    \F^\tau=\Omega^\bt(N,\g)[1] \oplus \Omega^\bt(N,\g)[2]\oplus \Omega^\bt(N,\g^*)[n-1] \oplus \Omega^\bt(N,\g^*)[n-2] \;\; \ni(\A,\alpha,\B,\beta)
\end{equation}
equipped with $(-1)$-symplectic form 
\begin{equation}
    \omega^\tau=\int_N \langle \delta\B,\delta\A \rangle + \langle \delta \beta,\delta \alpha\rangle
\end{equation}
and the action
\begin{equation}
    S^\tau=\int_N \langle \B, \dr_N \A+\frac12[\A,\A] \rangle + \langle \beta,\dr_N\alpha+[\A,\alpha]\rangle + \langle \B,\alpha \rangle.
\end{equation}
This is $BF$ theory on $N$ with coefficients in $\mathfrak{h}=\g\oplus \g[1]$ -- the Lie algebra $\g$ extended by the (shifted) adjoint module, seen as a \emph{differential graded} Lie algebra with differential $\g \xleftarrow{\mr{id}} \g[1]$.
\end{expl}

%\marginpar{\bl added May 14}
\begin{expl}\label{example: lifting IR Wilson loop via 2d theory tau}
    Consider Example \ref{example: theory tau for BF} with $M=S^1$, i.e., $(\F,\omega,S)$ is $BF$ theory on a circle. Consider the Lorenz gauge of Section \ref{sss 1d BF on S^1 Lorenz gauge} and an infrared observable $O'(A_\zm)$ given by a $G$-invariant function on $\g$. In this case, $\tau$ is the 2d $BF$ theory on a cylinder $N=S^1\times [0,T]$ with coefficients in the dg Lie algebra $\mathfrak{h}=\g\oplus \g[1]$. Formula (\ref{i_int as 1d AKSZ PI}) becomes
    \begin{equation}\label{i_int(O')(A_out) as integral in BF theory on the cylinder}
        i_\ii(O')(A_\mr{out})=\lim_{T\ra\infty}\int 
        %\mc{D}\mc{A}\mc{D}\mc{B} \mc{D}\alpha \mc{D}\beta\; 
        O'\left(\oint_{S^1} A\big|_{t=0}\right)\; e^{\frac\iii\hbar S^\tau|_{\LL^\tau}}.
    \end{equation}
    Here the outer integral is over fields $\A,\B,\alpha,\beta$ on the cylinder $N$, subject to boundary conditions $\A|_{t=T}=A_\mr{out}, \beta|_{t=T}=0$, $\alpha|_{t=0}=\mc{B}|_{t=0}=0$, and restricted to the gauge-fixing Lagrangian $\LL^\tau$ (which fixes the components of fields which are 1-forms along $[0,T]$).\footnote{
    Due to a conflict between notations for fields convenient for 2d theory vs.\ 1d theory, the fields of 1d $BF$ theory on $S^1$ in this example are $\mc{A},\beta$ instead of $\mc{A},\mc{B}$ as in Section \ref{sss 1d BF on S^1 Lorenz gauge}.
    }

    The equations of motion generated by the gauge-fixed action $S^\tau|_{\LL^\tau}$ extend $A_\mr{out}$ to a flat connection in the cylinder\footnote{
    More explicitly: let $s$ be the coordinate on $S^1$ and $t$ be the coordinate on $[0,T]$. Decomposing the $\g$-valued 1-form on the cylinder as $A=A_s ds + A_t dt$, flatness reads $\partial_t A_s = \partial_s A_t + [A_s,A_t]$. The equation of $\LL^\tau$ (\ref{L^tau}), using the non-normalized chain homotopy $\bar\kappa=\dr^*_s$ associated to Lorenz gauge (\ref{1d BF on S^1 barkappa}), 
    %of Section \ref{sss 1d BF on S^1 Lorenz gauge}, 
    yields $A_t=-\partial_s A_s$. Thus, we have a non-linear 
    %perturbation of the 
    heat equation for $A_s$:  ${\partial_t A_s= -\partial^2_s A_s -[A_s,\partial_s A_s]}$.
    } (so the conjugacy class of its holonomy is constant for each cross-section of the cylinder). In the asymptotics $T\ra\infty$, this extension restricts to a constant connection on the in-circle $S^1\times \{t=0\}$, which proves that semi-classically one has the result (\ref{i_int^cl(O') for a class function on S^1}):
    \begin{equation}
        i_\ii^\cl(O')(A_\mr{out})=O'(-\log P \exp \left(-\oint_{S^1} A_\mr{out}\right)).
    \end{equation}
    \begin{figure}[h]
        \centering
        \includegraphics[width=0.5\linewidth]{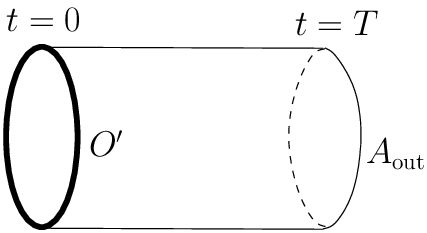}
        \caption{Illustration for Example \ref{example: lifting IR Wilson loop via 2d theory tau}: $BF$ theory on a cylinder with boundary condition for $A$ on one side ($t=T$) and observable $O'\left(A_\zm=\oint_{S^1\times \{0\}}A\right)$ on the other side yields in the limit $T\ra\infty$ the lift of the observable $i_\ii(O')(A_\mr{out})$, cf. (\ref{i_int(O')(A_out) as integral in BF theory on the cylinder}).
        }
        \label{fig:placeholder}
    \end{figure}
\end{expl}

\end{document}